%% file: arxiv0402.tex
\PassOptionsToPackage{unicode}{hyperref}
\PassOptionsToPackage{hyphens}{url}
\PassOptionsToPackage{dvipsnames,svgnames,x11names}{xcolor}
\documentclass[
  12pt]{article}

\usepackage{amsmath,amssymb}
\usepackage{iftex}
\ifPDFTeX
  \usepackage[T1]{fontenc}
  \usepackage[utf8]{inputenc}
  \usepackage{textcomp} 
\else 
  \usepackage{unicode-math}
  \defaultfontfeatures{Scale=MatchLowercase}
  \defaultfontfeatures[\rmfamily]{Ligatures=TeX,Scale=1}
\fi
\usepackage{lmodern}
\ifPDFTeX\else  
\fi
\IfFileExists{upquote.sty}{\usepackage{upquote}}{}
\IfFileExists{microtype.sty}{
  \usepackage[]{microtype}
  \UseMicrotypeSet[protrusion]{basicmath} 
}{}
\makeatletter
\@ifundefined{KOMAClassName}{
  \IfFileExists{parskip.sty}{%
    \usepackage{parskip}
  }{
    \setlength{\parindent}{0pt}
    \setlength{\parskip}{6pt plus 2pt minus 1pt}}
}{
  \KOMAoptions{parskip=half}}
\makeatother
\usepackage{xcolor}
\makeatletter
\ifx\paragraph\undefined\else
  \let\oldparagraph\paragraph
  \renewcommand{\paragraph}{
    \@ifstar
      \xxxParagraphStar
      \xxxParagraphNoStar
  }
  \newcommand{\xxxParagraphStar}[1]{\oldparagraph*{#1}\mbox{}}
  \newcommand{\xxxParagraphNoStar}[1]{\oldparagraph{#1}\mbox{}}
\fi
\ifx\subparagraph\undefined\else
  \let\oldsubparagraph\subparagraph
  \renewcommand{\subparagraph}{
    \@ifstar
      \xxxSubParagraphStar
      \xxxSubParagraphNoStar
  }
  \newcommand{\xxxSubParagraphStar}[1]{\oldsubparagraph*{#1}\mbox{}}
  \newcommand{\xxxSubParagraphNoStar}[1]{\oldsubparagraph{#1}\mbox{}}
\fi
\makeatother

\usepackage{longtable,booktabs,array}
\usepackage{calc} 
\usepackage{etoolbox}
\makeatletter
\patchcmd\longtable{\par}{\if@noskipsec\mbox{}\fi\par}{}{}
\makeatother
\IfFileExists{footnotehyper.sty}{\usepackage{footnotehyper}}{\usepackage{footnote}}
\makesavenoteenv{longtable}
\usepackage{graphicx}
\makeatletter
\def\maxwidth{\ifdim\Gin@nat@width>\linewidth\linewidth\else\Gin@nat@width\fi}
\def\maxheight{\ifdim\Gin@nat@height>\textheight\textheight\else\Gin@nat@height\fi}
\makeatother
\setkeys{Gin}{width=\maxwidth,height=\maxheight,keepaspectratio}
\makeatletter
\def\fps@figure{htbp}
\makeatother

\makeatletter
\@ifpackageloaded{caption}{}{\usepackage{caption}}
\AtBeginDocument{%
\ifdefined\contentsname
  \renewcommand*\contentsname{Table of contents}
\else
  \newcommand\contentsname{Table of contents}
\fi
\ifdefined\listfigurename
  \renewcommand*\listfigurename{List of Figures}
\else
  \newcommand\listfigurename{List of Figures}
\fi
\ifdefined\listtablename
  \renewcommand*\listtablename{List of Tables}
\else
  \newcommand\listtablename{List of Tables}
\fi
\ifdefined\figurename
  \renewcommand*\figurename{Figure}
\else
  \newcommand\figurename{Figure}
\fi
\ifdefined\tablename
  \renewcommand*\tablename{Table}
\else
  \newcommand\tablename{Table}
\fi
}
\@ifpackageloaded{float}{}{\usepackage{float}}
\floatstyle{ruled}
\@ifundefined{c@chapter}{\newfloat{codelisting}{h}{lop}}{\newfloat{codelisting}{h}{lop}[chapter]}
\floatname{codelisting}{Listing}

\makeatother
\makeatletter
\@ifpackageloaded{caption}{}{\usepackage{caption}}
\@ifpackageloaded{subcaption}{}{\usepackage{subcaption}}
\makeatother

\ifLuaTeX
  \usepackage{selnolig}  
\fi
\usepackage[]{natbib}
\usepackage{bookmark}

\IfFileExists{xurl.sty}{\usepackage{xurl}}{} 
\hypersetup{
  pdftitle={Title},
  pdfauthor={Author 1; Author 2},
  pdfkeywords={3 to 6 keywords, that do not appear in the title},
  colorlinks=true,
  linkcolor={blue},
  filecolor={Maroon},
  citecolor={Blue},
  urlcolor={Blue},
  pdfcreator={LaTeX via pandoc}}

\newcommand{\anon}{1}

\usepackage{enumerate}

\usepackage{mathtools,
  booktabs,
  braket,
  mathrsfs,
  latexsym,
  epsfig,
  lineno,
  placeins,
  tikz
}

\usepackage{amsthm}
 
\usepackage[flushmargin]{footmisc}

\usepackage{scalerel}
\usepackage{multirow}
\usepackage{arydshln}
 
\usepackage{soul}
 
\usepackage{float}

\usepackage{latexsym}
\usepackage{caption}
\usepackage[margin=20pt]{subcaption}
 
\usepackage{comment}
\theoremstyle{definition}
\newtheorem{assumption}{Assumption}

\newtheorem{theorem}{Theorem}

\newtheorem{proposition}{Proposition}
\newtheorem{lemma}{Lemma}
\newtheorem*{lemma*}{Lemma}

\newtheorem{definition}{Definition}

\newtheorem*{corollary*}{Corollary}

\usepackage{etoolbox} 

\usepackage{listings}
\usepackage{lscape}

\RequirePackage[normalem]{ulem}

\allowdisplaybreaks

\graphicspath{ {./Graphs/} }
\input{cmd_0208.tex}

\input{cmd_housekeeping.tex}

\def\spacingset#1{\renewcommand{\baselinestretch}%
{#1}\small\normalsize} 
\usepackage{hyperref}
\usepackage{xr}
\begin{document}

\def\ttl{Two-stage Least Squares with Clustered Data under the Local Average Treatment Effect Framework}

\if1\anon
{
  \title{\bf \ttl}
  \author{Anqi Zhao\thanks{
Peng Ding thanks the U.S. National Science Foundation (grants \# 1945136 and \# 2514234) for the support. 
We thank Bruce Hansen for discussions on the regularity condition, and thank Hongbin Huang, Weihang Liu, and Songliang Chen for research assistance.} \\
 Fuqua School of Business, Duke University\\
 \\
 Peng Ding\\
 Department of Statistics, University of California, Berkeley\\
 \\
 Fan Li\\
 Department of Statistical Sciences, Duke University}

   \date{}
  \maketitle
} \fi

\if0\anon
{
  \bigskip
  \bigskip
  \bigskip
  \begin{center}
 {\LARGE\bf \ttl}
\end{center}
  \medskip
} \fi

\bigskip

\bigskip
\begin{abstract}
To estimate the causal effect of an endogenous treatment using clustered data,
the {\it canonical \tslsf} (\tsls) estimates a linear regression of the outcome on treatment status using an instrumental variable (IV) and conducts inference with \crses.
When both the treatment and the IV vary within clusters, 
an alternative {\it \tsfef} (\tsfe) additionally includes cluster indicators in the regression, thereby incorporating cluster information into point estimation as well. 
This paper studies the trade-off between these approaches within the local average treatment effect (LATE) framework. 
When clusters are homogeneous, we show that both approaches yield valid large-sample inference for the LATE, and that \tsfe\ is more efficient than \cts\ only when \condeff\ and \condiv. 
When clusters are heterogeneous, we show that  \tsfe\ identifies a weighted average of cluster-specific LATEs, whereas the \cts\ generally does not. 
We further propose a test for detecting \ch. 
\end{abstract}

\noindent%
\vfill

\newpage
\spacingset{1.8} 

\section{Introduction}\label{sec:intro}

Clustered data, in which units of observation are nested within higher-level groups, are pervasive in empirical economics. 
Common examples include repeated observations on units over time and cross-sectional units grouped by geographic, institutional, or demographic characteristics, such as states, schools, and families.
When the objective is to estimate the causal effect of a potentially endogenous treatment, the {\it canonical \tslsf} (\tsls) estimates a linear regression of the outcome on treatment status using \tsls\ with an instrumental variable (IV) and use \crses\ to account for clustering in inference \citep[see, e.g.,][]{\baumid, cameron2015practitioner, andrews2019weak}. 
When both the treatment and IV vary within clusters, 
an alternative {\it \tsfef} (\tsfe) additionally includes cluster indicators in the regression, thereby incorporating the cluster information in point estimation as well \citep[see, e.g.,][]{angrist2004does, acconcia2014mafia, autor2013china}.

Despite the widespread use of these two procedures, their causal interpretation within the local average treatment effect (LATE) framework \citep{imbens1994identification, angrist1996identification} remains underexplored.
Our contributions are threefold.  
First, we establish the validity of both the \cts\ and \tsfe\ for large-sample inference of the LATE when clusters are homogeneous, and clarify their efficiency trade-offs. Specifically, 
\begine[(i)] 
\item Point estimators from the \cts\ and \tsfe\ are both consistent and asymptotically normal for estimating the LATE, and their associated \crses\ \citep{liang1986longitudinal} are consistent for the respective asymptotic variances. 
\item \dgp, \tsfe\ is more efficient than the \cts\ only when \condeff\ and \condiv.  
\item Although covariate adjustment does not necessarily improve efficiency with clustered data \citep{su2021model}, we analyze the special case under the above linear IV model, and quantify the efficiency gain from adjusting for cluster-level covariates in the \cts.
In contrast, \tsfe\ absorbs all cluster-level variation through cluster indicators, thereby precluding the inclusion of cluster-constant covariates. 
This suggests a potential advantage of the \cts\ when informative cluster-level covariates are available.
\ende

Second, we show that with heterogeneous clusters, \tsfe\ identifies a weighted average of cluster-specific LATEs, whereas the \cts\ generally does not. 
Consequently, in settings with potentially heterogeneous clusters,
\tsfe\ provides a more interpretable estimand than the \cts. 

Third, to guide empirical choice between the two procedures, 
we derive a joint central limit theorem for their respective point estimators under homogeneous clusters, and propose a test based on the difference between the two estimators for detecting \ch\ \citep{hausman1978specification}. 
Rejection of the test provides evidence of heterogeneity across clusters and suggests the use of \tsfe\ rather than the \cts. 

\paragraph*{Notation.} Let $1_{\{\cdot\}}$ denote the indicator function. 
Let $\rs$ denote convergence in distribution, and let $\op$ denote a sequence that converges to zero in probability. 
For $u_i\in \mbr$, $\{(v_{i1}, \ldots, v_{iK}): v_{ik} \in \mbr^{p_k}\}$, and $\{(w_{i1}, \ldots, w_{iL}): w_{il} \in \mbr^{q_l}\}$ defined for a population indexed by $\otn$, let $\tslst(u_i \sim v_{i1}+\cdots+v_{iK} \mid w_{i1}+\cdots+w_{iL})$ denote the \tsls\ regression of $u_i$ on $v_i = (v_{i1}^\T, \ldots, v_{iK}^\T)^\T$ over $\otn$, instrumented by $w_i = (w_{i1}^\T, \ldots, w_{iL}^\T)^\T$. We allow each $v_{ik}$ and $w_{il}$ to be a scalar or a vector and use + to denote concatenation of regressors. 
Throughout, we use $\tslst(\cdot)$ to denote the numerical outputs of \tsls\ without imposing any assumptions about the corresponding linear model. 
For a random variable $Y \in \mbr$ and a random vector $X \in \mbr^n$, let 
$\proj(Y\mid X)$ denote the linear projection of $Y$ onto $X$, defined as $\proj(Y\mid X) = b_0^\T X$, where $b_0 = \argmin_{b\in \mbr^{n}} \E\{(Y -  b^\T X)^2\} = \{\E(X X^\T)\}^{-1} \E(X Y)$. 

\section{Two 2SLS specifications for clustered data}\label{sec:setup}

\subsection{Clustered data with endogenous treatment}
Consider a study with two treatment levels, indexed by $d = 0,1$, and a study population of $N$ units nested within $G$ groups, or {\it clusters}, indexed by $g = \ot{G}$.
Let $\ng$ denote the size of cluster $g$, with $\ng \geq 1$ and $\sumg \ng = N$.
For a given study population, we view $G$ and $(n_1, \ldots, n_G)$ as fixed design parameters.
 
Index the $j$th unit within cluster $g$ by $gj$.
Let $\mi = \mif$ denote the set of all units, and $\mig = \{gj: \otng\}$ the set of units within cluster $g$.
When cluster membership is not central, we also use a single index $i$ to denote units in $\mi$, with $i = gj$. 
Let $\ci$ denote the cluster index of unit $i \in \mi$, with $\ci = g$ for $\ig$. 
Let $\cci = \ccif$ denote the corresponding vector of cluster indicators.

For each unit $i \in \mi$, we observe its treatment status $\di  \in \{0,1\}$,  outcome $\yi \in \mbr$,  a binary instrument $Z_i \in \{0,1\}$, and a vector of baseline covariates $\xxi\in\mbr^p$.
The goal is to estimate the causal effect of treatment on the outcome when the treatment status is endogenous.  

\subsection{2SLS for clustered data}\label{sec:setup_procedures}
We now formalize the definitions of the \cts\ and \tsfe, along with their corresponding point estimators and \crses. 
%
%

%
We begin with the unadjusted specifications to present our main theory and quantify the impact of covariate adjustment in \sec~\ref{sec:x}.  
As a foundation, \defcrse\ below reviews the numerical formulation of general \tsls\ following \citet[Sections 2.4--2.5]{baum2003instrumental}.

\begin{definition}\label{def:crse}
For  $u_i \in \mbr$ and $v_i, w_i  \in \mbr^m$ defined for $\otn$, let $\tslst(u_i\sim v_i \mid w_i)$ denote the \tslsr\ of $u_i$ on $v_i$ instrumented by $w_i$.
Let $\hb \in \mbr^m$ denote the estimated \coeffv\ of $v_i$ from $\tslst(u_i\sim v_i \mid w_i)$. 
Let $r_i = u_i - v_i^\T\hb$ denote the residual for unit $i$, and $\rg = \vjh{r} \in \mbr^{\ng}$ the residual vector for cluster $g$. 
The cluster-robust covariance of $\hb$ equals
\begineqs
\widehat{\cov}(\hb) = \xpwxi (V^\T \pw \ho \pw V )\xpwxi,
\endeqs
where $V = \vgj{v} \in \mbr^{N\times m}$, $\pw = \pwf$ with $W = \vgj{w} \in \mbr^{N\times m}$, and $\ho = \diag(r_{[1]} r_{[1]}^\T,  \ldots,  r_{[G]} r_{[G]}^\T) \in \mbr^{N\times N}$ denotes the block diagonal matrix with $\{r_{[g]} r_{[g]}^\T: g = \ot{G}\}$ on the diagonal. 
The \crses\ of the elements of $\hb$ are the square roots of the diagonal elements of $\widehat{\cov}(\hb)$. 
\end{definition}

Definition~\ref{def:ls} formalizes the \cts.

\begin{definition}[Canonical \tsls]\label{def:ls}
Fit $\tslst(Y_i \sim 1 + D_i \mid 1 + Z_i)$, the \tslsr\ of $Y_i$ on $v_i = (1, D_i)^\T$ instrumented by $w_i = (1, Z_i)^\T$, following \defcrse. 
Let $\htls$ denote the estimated coefficient on $\di$, and let $\hsels$ denote the corresponding \crse. 
\end{definition}

The \cts\ in \defls\ uses no cluster information in the regression specification. 
Consequently, clustering is ignored in point estimation and addressed only through the \crse\ $\hsels$.
As noted in \sec~\ref{sec:intro}, an alternative is to include cluster indicators $\cci = \ccif$ in the \tsls\ specification, thereby accounting for clustering in point estimation as well. We call this procedure {\it \tsfef} (\tsfe), formalized in \deffe\ below.

\begin{definition}[Two-stage least squares with fixed effects (\tsfe)]\label{def:fe}
Fit $\tslst(Y_i \sim D_i + \cci \mid Z_i + \cci)$, 
 the \tslsr\ of $Y_i$ on $v_i = (D_i, \cci^\T)^\T$ instrumented by $w_i = (Z_i, \cci^\T)^\T$, following \defcrse. 
Let $\htfe$ denote the estimated coefficient on $\di$, and let $\hsefe$ denote the corresponding \crse. 
\end{definition}

Following \cite{cameron2015practitioner}, \tsfe\ employs the \crse, $\hsefe$, even after including $C_i$ in the regression specification.
The coefficients on $C_i$ are commonly interpreted as cluster fixed effects in regression-based fixed-effects analyses \citep{wooldridge2010econometric}.
A restriction of \tsfe\ is that it is well defined only when both $Z_i$ and $D_i$ vary within clusters. It is therefore not applicable when the treatment, the IV, or both are constant within clusters.
A notable example is cluster randomized experiments with noncompliance \citep{frangakis2002clustered, jo2008cluster, agbla2020estimating}, in which the random treatment assignment serves as a cluster-level IV.


\subsection{Assumptions under the LATE framework}\label{sec:validity}
We formulate the sampling and IV assumptions using the potential outcomes framework \citep{imbens1994identification, angrist1996identification}.  
For $z, d = 0,1$,  
let $\dizz$ denote the potential treatment status of unit $i$ if $\zi=z$, and 
let $\yizd$ denote the potential outcome of unit $i$ if $(\zi, \di ) = (z,d)$.
The observed treatment status and outcome satisfy $\di  = \di (\zi)$ and $\yi = \yi(\zi, \di ) = \yi(\zi, \di(\zi))$.
Following standard terminology, we classify the units into four compliance types, denoted by $U_i$, based on the joint values of $\dio$ and $\diz$. Unit $i$ is 
an {\it always-taker} if $\dio = \diz = 1$, denoted by $\ui = \aa$; 
a {\it complier} if $\dio = 1$ and $\diz = 0$, denoted by $\ui = \cc$; 
a {\it defier} if $\dio = 0$ and $\diz = 1$, denoted by $\ui = \textup{d}$; and 
a {\it never-taker} if $\dio = \diz = 0$, denoted by $\ui = \nn$.

Let $\mqif$ denote the collection of all variables associated with unit $i$, and let $\mqg = \{\mqi: \ig\}$ denote the corresponding collection for cluster $g$.
\assmcs\ below states the standard {\it cluster sampling} assumption \citep{wooldridge2010econometric, hansen2019asymptotic}, requiring independence across clusters while allowing arbitrary dependence within clusters. 

\begin{assumption}[Cluster sampling]\label{assm:cs}
$\mqg \indep \mqgp$ for $ g \neq g' \in \{\ot{G}\}$. 
\end{assumption}

\assmiv\ below states the IV assumptions for clustered data.

\begin{assumption}[IV assumptions for clustered data]\label{assm:iv} 
For $\otn$, the following holds:  
\begine[(i)]
\item\label{it:assm_iv_er} \textit{Exclusion restriction:}
$\yi(0,d) = \yi(1,d) = \yi(d)$  for $d = 0,1$, where $\yi(d)$ denotes the common value.
\item \label{it:assm_iv_random} \textit{Cluster-level random assignment:} For $g = \ot{G}$,\\ 
$\left\{\zi: \ig
\right\} \indep \left\{\pypd, \xxi: z, d= 0,1; \ \ig \right\}.$
\item \textit{Overlap:} $0 < \pr(\zi = 1) < 1$. 
\item\label{it:assm_iv_relevance} \textit{Relevance:} $\E\{\di(1) - \di(0)\} \neq 0$.
\item\label{it:assm_iv_mono} \textit{Monotonicity:} $\di(1) \geq \di(0)$.  
\ende 
\end{assumption}

\assmiv\ follows the classical IV assumptions in \cite{angrist1996identification} except for \assmiv\eqref{it:assm_iv_random}.
Unlike the classical counterpart $
\zi \indep \{\pypd, X_i:   z, d= 0,1 \}$, 
which requires random assignment at the unit level, \assmiv\eqref{it:assm_iv_random} requires random assignment of the IV at the cluster level, 
such that the IV for a given unit is independent not only of its own potential treatment status and outcomes but also of those of other units within the same cluster.
Aside from this, 
\assmiv\eqref{it:assm_iv_random} imposes no constraints on the joint distribution of $\{Z_i: \ig\}$ within each cluster. 
It therefore accommodates (a) IVs that are independent within clusters, (b) cluster-level IVs that are constant within clusters, and (c) IVs as-if assigned via stratified randomization, so that the proportion of  units with $Z_i = 1$ within each cluster is fixed at prespecified values $e_g$.

\assms~\ref{assm:im}--\ref{assm:hetero_marginal} below state two alternative assumptions on the marginal distributions of $\mqi$ across clusters.
\assmim\ requires a common marginal distribution across all units.
This assumption holds under a {\it two-stage}, or {\it grouped sampling}, design \citep{arkhangelsky2024fixed}, in which clusters are randomly sampled from a superpopulation and, conditional on selection, $\ng$ units are randomly sampled from each sampled cluster $g$. 
Assumption~\ref{assm:hetero_marginal} relaxes \assmim\  by allowing the marginal distribution of $\mqi$ to vary across clusters.
We refer to \assms~\ref{assm:im} and \ref{assm:hetero_marginal} as the {\it homogeneous-clusters}  and {\it heterogeneous-clusters} assumptions, respectively. 

\begin{assumption}[Homogeneous clusters]\label{assm:im}
 $\{\mqi:\otn\}$ are identically distributed.
\end{assumption}

\begin{assumption}[Heterogeneous clusters]\label{assm:hetero_marginal} 
For each $g \in \{\ot{G}\}$, $\{\mqgj: \otng\}$ are identically distributed within the cluster, whereas $\mq_{gj}$ and $\mq_{g'j'}$ may differ in distribution for $g \neq g'$. 
\end{assumption}

We assume \assmhomo\ in Section~\ref{sec:homo} to develop theory under homogeneous clusters, and discuss relaxations of \assmim\ to \assm~\ref{assm:hetero_marginal}  in \sec~\ref{sec:hetero}.

\section{Theory under homogeneous clusters}\label{sec:homo}
Assume \assmhomo\ throughout this section. Recall that $\yid$  denotes the common value of $\yi(0,d) = \yi(1,d)$ under exclusion restriction in \assmiv\eqref{it:assm_iv_er}.
%
Define $\tau_i = \yio - \yiz$ as the individual treatment effect of unit $i$, and define 
\begineq\label{eq:tc}
\tc  = \E(\tau_i \mid \uc)
\endeq as the {\it local average treatment effect} (LATE) on compliers under {\assmtth}. 
We establish below the validity and relative efficiency of the \cts\ and \tsfe\ for estimating $\tc$. 
%

\subsection{Validity of Wald-type inference}
\propclt\ below establishes central limit theorems for $(\htls, \hsels)$ and $(\htfe, \hsefe)$ under \assmhomo, thereby validating the \cts\ and \tsfe\ for large-sample Wald-type inference on the LATE with homogeneous clusters. 
We adopt the asymptotic framework of \cite{hansen2019asymptotic}, reviewed in \assmngm\ below. 
The framework requires clusters to be asymptotically negligible, implying that the number of clusters $G$ diverges. 
Beyond this, it allows for unequal and diverging cluster sizes $\ng$.

\begin{assumption}\label{assm:ng_main}
As $\ntinf$, $\ds\max_{\otg}  \,  {\ngsq}/N \to 0$ and $ {\ds \limsup_{N\to\infty}} \, N^{-1}\sumg n_g^2   < \infty$. 
\end{assumption}

Let $\sz =  \meani(\zi - \bz_{\ci})^2$ denote the sample within-cluster variance of $Z_i$, where $\bzg = \meanig \zi$.
Under \assmtth, let $
\vz = \var(Z_i) > 0$ and $\pc = \pr(\uc) > 0$ denote the IV probability and complier proportion, respectively, and define 
\beginy\label{eq:kn}
\kn = \esz / \vz.   
\endy
Intuitively, $\kappa_N$ captures  the within-cluster variance of treatment assignments and the degree of balance in treatment proportions, and is related to the balance parameter in the covariate-adaptive randomization literature \citep{bugni2018inference}. 
We can show that
\begineq\label{eq:kn_01}
\kn \in [0,1],
\endeq 
where $\kn = 0$ if and only if $Z_i$ is constant within each cluster, so that \tsfe\ is degenerate; 
and $\kn = 1$ if and only if 
\begineq\label{eq:sre}
\text{$\var(\bzg) = 0$ for all $g$;}
\endeq 
see \lemz\ in the \sm. 
Condition~\eqref{eq:sre} implies that the IV proportion $\bzg$ is fixed within each cluster, so that under \assmiv, the IV is as if randomly assigned via stratified randomization.
Let $\ai = \yi - \di\tc$, $\mua = \E(\ai)$, and $\bag = \meanig \ai$. Let $e = \pr(Z_i = 1)$ denote the IV probability under \assmim.

\begin{proposition}\label{prop:clt}
Under \assmhomo, \ref{assm:ng_main}, and proper moment and rank conditions,  
$$
\sqrtn (\hts - \tc)/\sqrt{\vs}  \rs  \sn,\quad  N\hses^2 = \vs + \op, \quad (\hts - \tc)/\hses  \rs  \sn  
$$
for $* = \fe, \ls$, 
where
\begineq\label{eq:vls_vfe_def}
\beginar{rcl}
\vls 
&=& \ds \dfrac{1}{\vzpcsq} \cdot\nifsumg \var\left\{\sumig (\zi - e) (\cai)\right\},\\
\vfe
&=& \ds \dfrac{1}{\kvzpcsq} \cdot\nifsumg \var\left\{\sumig (\zi - \bzg) (\ai - \bag)\right\}. 
\endar
\endeq 
\end{proposition}

\propclt\ establishes the consistency and asymptotic normality of $\htls$ and $\htfe$ for estimating $\tc$, and justifies using \crses\ to estimate the corresponding asymptotic variances. 
The result for $(\htfe, \hsefe)$ echoes \cite{cameron2015practitioner}, underscoring the need to employ \crses\ even when cluster indicators are included.
We provide the formal moment and rank conditions in \thmjoint\ in Section~\ref{sec:hetero}, where we establish the joint asymptotic distribution of $(\htls, \htfe)$ for testing potential violations of \assmim. 
Both central limit theorems follow from combining the asymptotic theory for clustered data  \citep{hansen2019asymptotic} with the causal interpretation of \tsls\ coefficients under the LATE framework, without requiring correct specification of the outcome models. 
Establishing the asymptotic distribution of $\htfe$ requires additional technical arguments, as the dimension of $C_i$ diverges.

\subsection{Efficiency comparison}\label{sec:eff}
From \eqref{eq:vls_vfe_def}, the ratio $\vls/ \vfe$ measures the asymptotic efficiency of $\htls$ relative to $\htfe$ under arbitrary potential outcomes model satisfying \assmhomo. 
This ratio depends on the IV assignment mechanism through $(\kn, \zi, \bzg)$ and on the outcome model through $(\aai, \bag)$. In particular, the asymptotic variance of $\htfe$, $\vfe$, has $\kn$ in the denominator, implying that efficiency requires sufficient within-cluster variation in the IV; c.f. \eqref{eq:kn}--\eqref{eq:kn_01}. 
To provide insight, we quantify below the respective contributions of these components to relative efficiency under a linear IV model. 

As we show below, \tsfe\ is more efficient than the \cts\ only when \condeff\ and \condiv, and adjusting for cluster-level covariates improves the efficiency of the \cts.
These results caution against the routine inclusion of cluster indicators in regression specifications without evaluating the associated efficiency trade-offs, particularly when informative cluster-level covariates are available.

\subsubsection{Linear IV model for comparing efficiency}\label{sec:model}
\assmy\ below specifies the outcome model that we use to study relative efficiency. 

\begin{assumption}\label{assm:y}
For $\otn$, the potential outcomes satisfy
\begine[(i)]
\item\label{it:assm_y_linear iv} $\tau_i = \tc$ for compliers;
\item\label{it:assm_y_A} $\yi  = \di \tc + \aci + \epi$, where $\{\alpha_g \in \mbr: \otg\}$ are  unobserved \cfes\ and $\{\ep_i\in \mbr: \otn\}$ are unobserved unit-level errors with 
\begine[-]
\item $\E(\ag) = \mual$ and $\var(\ag) = \va$ for unknown constants $(\mual,\va)$;
\item  $\epg = (\ep_{g1}, \ldots, \ep_{g,\ng})^\T$ are independent across $\otg$;
\item  $ \E(\epg \mid \ag) = 0_{\ng}$ and $ \cov(\epg \mid \ag) = \ve \ing$ for an unknown constant $\sigesq > 0$. 
\ende 
\ende
\end{assumption}

\assmy\eqref{it:assm_y_linear iv} imposes constant treatment effects for compliers, but otherwise places no restrictions on always-takers and never-takers.
It implies $Y_i(D_i(1)) - Y_i(D_i(0)) = \tc\{\dio - \diz\}$, which is the linear IV model proposed by \cite{imbens2005robust}. 
Under \assmy\eqref{it:assm_y_linear iv}, it follows that $\aai = \yi - \di\tc = \yiz + \oua(\tau_i - \tc)$,  so that $\aai$ is a baseline unobserved attribute independent of $\zi$.
\assmy\eqref{it:assm_y_A} then imposes a cluster structure on $\aai$.
Under \assmy\eqref{it:assm_y_A}, the observed outcome shares the same functional form as the fixed effects model \citep{wooldridge2010econometric, arkhangelsky2024fixed}.
However, let $\dg = (D_{g1}, \ldots, D_{g,\ng})^\T$ denote the treatment vector for cluster $g$. The fixed effects model further requires 
\begineq\label{eq:fe_assm}
\E(\epg\mid \dg, \ag) = 0 \for g = \ot{G}  
\endeq
in addition to \assmy. In contrast, \assmy\ places no restrictions on the dependence between $\dg$ and $\{ \ag, \epg : g = \ot{G}\}$.
%
When \eqref{eq:fe_assm} is violated, standard fixed effects estimation via the \olsr\ of $Y_i$ on $(D_i, \cci)$ may be inconsistent. 

\propclt\ implies that both $\htls$ and $\htfe$ are consistent for $\tc$ regardless of whether \eqref{eq:fe_assm} or \assmy\ holds.
We establish below their relative efficiency under \assmy. 

\subsubsection{Results}
Recall that $\kn  = \esz / \vz \in [0,1]$ as defined in \eqref{eq:kn}--\eqref{eq:kn_01}. Theorem~\ref{thm:eff} below provides the simplified forms of $\vls$ and $\vfe$ under \assmy, and establishes conditions for efficiency gain by \tsfe. 
To simplify the presentation, we assume $\inf_N\kn > 0$, so that \tsfe\ is nondegenerate for all $N$. 

\begin{theorem}\label{thm:eff}  
Assume \assmyall. 
Let $
\cn = N\var(\bz)/\vz$, where $\bz = \meani \zi$. Then
\begineq\label{eq:vr_def}
\vls  = \dfrac{\ve + \va \cdot \cn}{\vz \cdot \pcsq}, 
\quad 
\vfe = \dfrac{\ve}{\kn \cdot \vz \cdot \pcsq} ,  
\quad
\vr = \vrf, 
\endeq
where 
\begine[(i)]
\item $\vls > \vfe$ if and only if 
\begineq\label{eq:eff_cutoff}
\vavef  >  \dfrac{1-\kn}{\kn } \cdot \dfrac{1}{ \cn}.
\endeq
\item\label{it:eff_equal to 1} $\vls  = \vfe$ if \condsre\ holds.   
\ende
\end{theorem}

From \eqref{eq:vr_def}, under \assmy, the asymptotic variance of $\htls$, $\vls$, depends on both $\ve$ and $\va$, whereas that of $\htfe$, $\vfe$, depends only on $\ve$. 
From \eqref{eq:kn_01} and \eqref{eq:vr_def}, the ratio $\vls/\vfe$ may be either greater or less than 1, implying that including cluster indicators via \tsfe\ does not necessarily improve efficiency relative to the \cts. 
In particular, \eqref{eq:vr_def} implies that $\vls/\vfe$ depends on both $\vave$, which captures \vaveword\ in the outcome model under \assmy, and 
 $(\kn,\cn)$, which are determined by the IV assignment mechanism.  
We discuss their respective impacts below. 
  
\paragraph*{\bf Impact of $\vave$ on $\vls/\vfe$.} 
From \eqref{eq:vr_def}, the ratio $\vls/\vfe$ increases with $\vave$. Therefore, greater cluster-level heterogeneity increases the efficiency of $\htfe$ relative to $\htls$. 
In particular, \eqref{eq:eff_cutoff} provides a lower bound on $\vave$ that ensures the asymptotic efficiency of $\htfe$ relative to $\htls$. 
As two limiting cases: 
\begini
\item The classical fixed effects model views fixed effects as unrestricted constants \citep{raudenbush2002hierarchical}.
Heuristically, this can be viewed as assuming $\va \approx \infty$ in \assmy.
From this perspective, the condition in \eqref{eq:eff_cutoff} necessarily holds, implying the asymptotic efficiency of $\htfe$ relative to $\htls$. 
In this setting, clusters can also be viewed as heterogeneous with distinct outcome distributions; see \assm~\ref{assm:hetero_marginal} and \sec~\ref{sec:hetero}.
\item When $\va=0$, \eqref{eq:kn_01} and \eqref{eq:vr_def} imply 
\begin{equation}\label{eq:eff_va=0}
\vr = \kn \leq 1,
\end{equation} where the equality holds if and only if \condsre\ holds, under which the IV proportion $\bzg$ is fixed within each cluster. This implies that $\htfe$ is less efficient than $\htls$ in the absence of \cfes.
\end{itemize} 


\paragraph*{\bf Impact of $(\kn,\cn)$ on $\vls/\vfe$.} 
\lemz\ in the \sm\ implies that $\kn$ and $\cn$ are both functions of the average pairwise correlation of $\zi$ within cluster $g$, denoted by $\bar\rho_g = \{\ng(\ng-1)\}^{-1}\sum_{i\neq i'\in \sg}\corr(\zi, \zi') \in [-1/(\ng-1),1]$, with both $\vlsn$ and $\vfen$ increase with $\bar\rho_g$. 
The net effect of $\bar\rho_g$ on $\vls/\vfe$ is therefore indeterminate. 
To provide insight, consider a special case in which (a) clusters are of equal sizes with $n_g = \nb$ for all $\otg$; (b) IVs are equicorrelated within clusters with $\corr(Z_i, Z_{i'}) =  \rz$ for all $i \neq i'$ in the same cluster. This is appropriate under exchangeable observations and includes the case in which IVs are independently assigned within clusters.
Then $\kn = (1-\bni)(1-\rz)$ and $\cn = (\bn - 1)\rz + 1$ by \lemz, so that 
 $\vr > 1$ if and only if 
\begineq\label{eq:cor_z}
-\dfrac{1}{\bn-1} < \rz < 1-\dfrac{\ve}{\va}\cdot\dfrac{1}{\bn-1},
\endeq
and $\vr = 1$ when $\rz$ attains the boundary values in \eqref{eq:cor_z}. 
These results characterize the nonlinear relationship between relative efficiency and within-cluster IV variation, as captured by $\rz$.  
The necessary and sufficient condition in \eqref{eq:cor_z} implies that \tsfe\ is more efficient only if $\rz < 1-  \dfrac{\ve}{\va} \cdot \dfrac{1}{\bn-1}$, but this efficiency gain vanishes as $\rz$ approaches $-\dfrac{1}{\bn-1}$, the lower bound that $\rz$ can attain corresponding to the case of fixed within-cluster IV proportions in \condsre.
The interval in \eqref{eq:cor_z} is non-empty if and only if  $\va > \ve/\bn$. 

\subsection{Impact of covariate adjustment}\label{sec:x}
We next examine the impact of covariate adjustment. 
Recall from \assmiv\ that we focus on settings in which the IV is independent of the covariates. Therefore, including covariates does not affect the validity of inference, but may improve precision.
Both the \cts\ and \tsfe\ accommodate individual-level covariates. 
However, only the \cts\ accommodates cluster-level covariates, whereas \tsfe\ absorbs them through 
$C_i$. 
To simplify the exposition, we focus on cluster-level covariates, and show that under Assumption~\ref{assm:y}, adjusting for cluster-level covariates can improve the efficiency of the \cts.
This suggests a potential advantage of the \cts\ over \tsfe\ when informative cluster-level covariates are available. 

Recall that $\xxi \in \mbr^p$ denotes the vector of baseline covariates for $\otn$.
Define
\begineq\label{eq:fex}
\beginar{l}
\ctsxformula,\\
\tsfexformula
\endar
\endeq
as the \ca\ variants of the \cts\ and \tsfe, respectively. 
Let $(\htlsx, \hselsx)$ and $(\htfex, \hsefex)$ denote the resulting estimated coefficients on $\di$ and their corresponding {\crse}s.
Similar to \propclt, we can show that 
under \assmhomo, \ref{assm:ng_main}, and proper moment and rank conditions, 
\begineqs
(\htx - \tc) / \hsex \rs \sn, \quad (\htfex - \tc) / \hsefex \rs \sn. 
\endeqs
This establishes the validity of $(\htlsx, \hselsx)$ and $(\htfex, \hsefex)$ for large-sample Wald-type inference on the LATE under arbitrary potential outcomes models; see \thmlsxfex\ in the \sm\ for details. 

We now examine the impact of covariate adjustment on efficiency when the potential outcomes  model satisfies Assumption \ref{assm:y}. 
As a basis, \assmyx\ below augments \assmy\ to require exogeneity of $\xxi$ conditional on the \cfes.  
\begin{assumption}\label{assm:y_x}
In the setting of \assmy, further assume 
$ \E(\epg \mid \xg, \ag) = 0_{\ng}$ and $ \cov(\epg \mid \xg, \ag) = \ve \ing$, where $\xg = \{\xxi: \ig\}$. 
\end{assumption}

\prop~\ref{prop:x_eff} below quantifies the efficiency gain from adjusting for cluster-level covariates in the \cts.

\begin{proposition}\label{prop:x_eff}
Assume \assmhomo, \ref{assm:ng_main}--\ref{assm:y_x}, and proper moment and rank conditions.
Suppose that $\xxi$ is constant within clusters, and let $\xsg$ denote the common value of $\xxi$ within cluster $g$. 
Then the \ca\ \tsfe\ in \eqref{eq:fex} reduces to the covariate-free \tsfe\ in \deffe, with $\hsefexsq = \hsefesq$, while $\hselsxsq$ from the \ca\ \cts\ satisfies
\begina
\dfrac{\hselsxsq}{\hsesqls} &=& 1 - \var\{\proj(\ag\mid 1, \xsg)\} \cdot\dfrac{\cn}{\ve + \va\cdot\cn} + \op,\\
\dfrac{\hselsxsq}{\hsesqfe} &=& \dfrac{\hselsxsq}{\hsesqls}\cdot \vrf + \op,
\enda where $\proj(\ag\mid 1, \xsg)$ denotes the linear projection of $\ag$ onto $(1,\xsg)$, with $\var\{\proj(\ag\mid 1, \xsg)\} \in [0, \va]$. 
 
\end{proposition}

\propxeff\ implies that the variance reduction for the \cts\ from adjusting for $\xxi = \xsg$ results from the portion of variation in $\ag$ linearly explained by $\xsg$. 
In particular, the linear projection of $\ag$ onto $(1, \xsg)$, denoted by $\proj(\ag\mid 1, \xsg)$, captures the part in $\ag$ linearly explained by $\xsg$.
When $\proj(\ag\mid 1, \xsg) = 0$, so that $\xsg$ explains no variation in $\ag$, covariate adjustment yields no efficiency gain, and the \ca\ \cts\ is equivalent to its unadjusted counterpart. The efficiency comparison between $\htlsx$ and $\htfe$ reduces to \thmeff:
\begineqs
\dfrac{\hselsxsq}{\hsesqls} = 1  + \op, \quad \dfrac{\hselsxsq}{\hsesqfe} =  \vrf + \op.
\endeqs
Conversely, when $\ag$ in linear in $\xsg$, we have $\var\{\proj(\ag\mid 1, \xsg)\}=\va$ so that 
\begineqs
\dfrac{\hselsxsq}{\hsesqls} = \dfrac{\ve}{\ve + \va\cdot\cn}  + \op, \quad \dfrac{\hselsxsq}{\hsesqfe} =  \kn + \op,
\endeqs
implying the asymptotic efficiency of $\htlsx$ relative to $\htfe$. 
Intuitively, in this case, all \cfes\ are explained by $(1,\xsg)$ under the \cts\ and by cluster indicators $\cci$ under \tsfe, so the efficiency comparison reduces to the setting without \cfes\ considered in \eqref{eq:eff_va=0}.
Together, these results identify when 2SFE dominates and when cluster-level covariate adjustment in canonical 2SLS should be preferred. 

\section{Theory under heterogeneous clusters}\label{sec:hetero}
\subsection{Causal interpretation of $\htls$ and $\htfe$}\label{sec:hetero_plim}
The discussion so far has assumed identical marginal distributions of $\mqi$ across all units under \assmim. 
We now examine the causal interpretation of $\htls$ and $\htfe$ under cross-cluster heterogeneity, as formalized in \assmhm.  

Assume \assmhetero\ throughout, so that the marginal distribution is now common only within cluster. 
Define {\it cluster-specific LATE} as
\begineq\label{eq:tcg}
\tcg = \E(\tau_{gj} \mid U_{gj} = \cc)  
\endeq
for $\otg$, 
where the expectation is taken \wrt\ the common marginal distribution of units $\{gj: j = \ot{\ng}\}$ in cluster $g$.
Let $\eg = \pr(\zgj = 1)$ and $\pcg = \pr(\ugj = \cc)$ denote the cluster-specific counterparts of $e = \pr(\zi = 1)$ and $\pc = \pr(\ui = \cc)$  under \assmim. 
Then $\eg$ defines the IV propensity score conditional on cluster membership.
Let $\muyzg$, $\pag$, and $\tauag$ denote the common values of $\E\{Y_\gj(0)\}$, $\pr(\ugj = \aa)$, and $\E(\tau_\gj \mid \ugj = \aa)$ across $\otng$, respectively, for units within cluster $g$.   
Let $\phigb$ denote the average pairwise correlation of $\zi$ within cluster $g$, including the diagonal. 

\begin{theorem}\label{thm:hetero} 
As $\ntinf$, if \assmhetero\ hold, $\ds\max_{\otg} \ng/N \to 0$, and $ \sup_{\otn}  \E ( \yi^2 ) < \infty$,
then 
\begineqs
\beginar{l}
\htls 
=
\dfrac
{\ds\sumg \ng \eg(\tmug + \pcg\tcg) - \ni\left( \sumg \ng  \eg\right)\left\{ \sumg\ng (\tmug + \eg\pcg\tcg)\right\}}
{\ds\sumg \ng \eg(\pag + \pcg) - \ni\left( \sumg \ng  \eg\right)\left\{ \sumg\ng (\pag + \eg\pcg)\right\}} + \op,\medskip\\
%
%
%
\htfe 
=
\ds\sumg \kgfe \cdot \tcg + \op,
\endar
\endeqs
where $\tmug= \muyzg +\pag\tauag$ and $\kgfe = \dfrac{ \ng(1-\phig)e_g(1-e_g) \cdot \pcg}{\sumg \ng(1-\phig)e_g(1-e_g) \cdot \pcg}$
with $\kgfe \in [0,1]$ and $\sumg \kgfe = 1$. 
Further assume that $e_g = e$  for all $g$. Then  
\begineqs
\beginar{l}
\ds \htls 
=
\sumg\kgls\cdot \tcg + \op,\quad 
\htfe 
=
\ds\sumg \kgfe \cdot \tcg + \op,
\endar
\endeqs
where 
\begineqs
\kgls = 
\dfrac
{\ng \pcg }
{\sumg \ng \pcg}, \quad 
\kgfe = \dfrac{ \ng(1-\phig) \pcg}{\sumg \ng(1-\phig) \pcg} 
\endeqs
with $\kappa_{g,*} \in [0,1]$ and $\sumg \kappa_{g,*} = 1$ for $* = \textup{2sls, 2sfe}$. 
 
\end{theorem}

\thmhetero\ shows that, under cluster heterogeneity, $\htfe$ identifies a weighted average of the cluster-specific LATEs $\{\tcg: \otg\}$, with nonnegative weights summing to one, whereas $\htls$ generally does not. 
In particular, the $\eg(1-\eg)$ component of $\kgfe$ coincides with overlap weighting \citep{li2018balancing}.
In contrast, even when $\tcg = 0$ for all $g$, the probability limit of $\htls$ has $\sumg \ng \eg \tmug - \ni ( \sumg \ng  \eg )(\sumg\ng  \tmug)$ as the numerator, and is therefore generally nonzero if $e_g$ varies across clusters. 
Therefore, $\htls$ is not level-independent in the sense of \cite{blandhol2022tsls}---meaning it does not necessarily equal zero even when $\tcg = 0$ for all $g$, a property that is arguably a basic necessary condition for a quantity to admit a causal interpretation.
This distinction suggests an advantage of including fixed effects when analyzing data nested within heterogeneous clusters. 
See \cite{bugni2018inference}, \cite{bugni2023inference} and \cite{ding2021frisch} for related results under stratified randomization with a fixed number of strata as $N$ goes to infinity. 

As a special case, when the heterogeneous clusters have equal IV probabilities $e_g$, $\htls$ also converges in probability to a weighted average of $\{\tcg: \otg\}$, with nonnegative weights summing to one. However, these weights generally differ from those associated with $\htfe$. 
An example is randomized encouragement designs in which the encouragement, serving as the IV, is randomly assigned with equal probabilities across all units \citep{angrist1996identification, hoffmann2025vaccines}.  

\subsection{Test for \ch}\label{sec:test}
\thmhetero\ establishes an advantage of \tsfe\ over the \cts\ when analyzing heterogeneous clusters.
To facilitate empirical choice between the two procedures, 
we propose below a test for violations of \assmim,  motivated by the difference between the probability limits of $\htls$ and $\htfe$ in \thmhetero.  
Rejection of the test provides evidence against homogeneous clusters and favors \tsfe\ with a clear interpretation under heterogeneity.

As a theoretical foundation, 
Theorem~\ref{thm:joint} below establishes a central limit theorem for the joint distribution of $(\htls, \htfe)$ under \assmim.
Recall that $\bz = \meani \zi$ and $\bzg = \meanig\zi$. Similarly, let $\bd = \meani \di$ and $\bdg = \meanig \di$ denote the overall and within averages of $\di$.  
Let $\rils$ and $\rife$ denote the residuals from the \cts\ and \tsfe, respectively. Let $\szdt = \szdtf$ and $\szd = \szdf$ denote the sample covariance and sample within-cluster covariance of $(\zi, \di)$, respectively, normalized by $N$. 
Let $
\hsiglsfe =  
\beginp \hsesq_\ls & \hclsfe \\ \hclsfe & \hsesq_\fe\endp = \sum_{g=1}^{G} \hat{v}_g \hat{v}_g^{\top}$,
where 
$
\hat{v}_g = \dfrac{1}{N}\begin{pmatrix} S_{ZD}^{-1}  \sumig (Z_i - \bar{Z}) \rils \\ \szd^{-1} \sumig (Z_i - \bar{Z}_g) \rife \end{pmatrix}$.
This gives the explicit forms for \crses\ $\hsesqls$ and $\hsesqfe$, as well as the covariance estimator; see \thmlsxfex\ in the \sm\ for details. 
Let 
\begineq\label{eq:t-stat}
t_\diff = \dfrac{\htls - \htfe}{\hsed} 
\endeq  
denote the $t$-statistic associated with $\htls -\htfe$, where  $ \hsed^2  = (1,-1)\hsiglsfe(1,-1)^\T = \sum_{g=1}^{G}\{ (1,\ -1) \hat{v}_g\}^{2}$. 
Under \assmtth, recall the definitions of $(\ai, \mua, \bag)$ from  \prop~\ref{prop:clt}. 
Let $
\omgn = \ni\cov\left\{\dsumi \beginp \ai - \mua  \\ \zi(\ai - \mua) \\ (\zi - \bzci) (\ai - \baci)\endp \right\}$, 
and let $\lm(\omgn)$ denote the smallest eigenvalue of $\omgn$. 

\begin{theorem}\label{thm:joint}
As $\ntinf$, if \assmhomo\ and \ref{assm:ng_main} hold, $\E(Y_i^4) < \infty$, $\inf_N \kn >0$, and $\inf_N\lambda_{\min}(\omgn) > 0$, 
then 
\begineqs
\hsiglsfe^{-1/2}  \beginp 
\htls - \tc\\
\htfe - \tc
\endp 
\rs  
 \mn(0_2, I_2), \quad t_\diff \rs \sn.
 \endeqs
\end{theorem}

\thmjoint\ implies \propclt. 
Together, \thms~\ref{thm:hetero}--\ref{thm:joint} motivate a test of \assmim\ (homogeneous clusters) based on $t_\diff$. 
In practice, we can compute $t_\diff$ either using the explicit expression of $\hsed$  below \eqref{eq:t-stat} or via a cluster bootstrap approximation.
%

\section{Simulation}\label{sec:simu}
\def\xiunit{X_i'}
\subsection{Validity and efficiency with homogeneous clusters}\label{sec:simu_ca}
We first illustrate the validity of the \cts\ and \tsfe\ under \assmim, and examine the impact of covariate adjustment.
We fix the number of clusters at $G = 200$ and generate the data as follows: (i) $n_1, \ldots, n_G$ are \iid\ Poisson(10); (ii) $e_1, \ldots, e_G$ are \iid\ Uniform(0.4, 0.6); (iii) for $\otn$, 
\begini
\item $U_i$ are \iid\ with $\pr(\ua) = 0.3$,\quad $\pr(\uc) = 0.5$,\quad $\pr(\un) = 0.2$; 
\item $Z_i$ are independent with $\pr(Z_i =1 \mid e_{c(i)}) = e_{c(i)}$; 
\quad $D_i = \oua + Z_i \cdot \ouc $; 
\item  $X_i = (\xci, \xiunit)^\T$, where $X_g^*\simiid\mn(0,\sigma_X^2)$ and $\xiunit\simiid\mn(0,1)$; 
\item  $Y_i = D_i + \xci + \xiunit + \aci + \epi$, where
 $\ag = X_g^* + \eta_g$ with $\eta_g\simiid\mn(0, \sigma_\eta^2)$, and $\epi$ are independent $\mn(\mu_{U_i}, 1)$ with $(\mu_\aa, \mu_\cc, \mu_\nn) = (2, 0, -3)$.
\endi
The data-generating process implies that $\tc=1$, $\va = \sigma_X^2 + \sigma_\eta^2$, and $\sigesq = 4$. The correlation between $D_i$ and $\epi$ through $U_i$ renders $\di$ endogenous. The cluster-constant covariates $\xci$ help explain variation in $\ag$. 

We consider estimation using the \cts\ and \tsfe, indexed by  \{2sls, 2sfe\}, as well as their respective \ca\ counterparts, indexed by \{2sls-x, 2sfe-x\}. 
%
%
Table~\ref{tb:simu} reports the mean squared errors of the four \tsls\ estimators, along with the coverage rates and average lengths of the corresponding 95\% confidence intervals, over 1,000 independent replications at $(\sigma_X, \sigma_\eta)=(1,1)$, $(0.5,0.5)$, $(0.2,0.2)$. 
The main findings are threefold and align with our theory: 
\begine[(i)]
\item Across all settings, the 95\% confidence intervals from the four procedures attain correct coverage.  
\item The \ca\ variants, 2sls-x and 2sfe-x, deliver shorter average confidence interval lengths than their unadjusted counterparts, 2sls and 2sfe. 
\item The impact of covariate adjustment varies with $\va/\ve = 4^{-1}(\sigma^2_X + \sigma^2_\eta)$:

\item[$\bullet$] At $(\sigma_X, \sigma_\eta)=(1,1)$ with $\va/\ve = 0.5$, the two \tsfe\ procedures with cluster indicators (2sfe, 2sfe-x) yield shorter average confidence interval lengths than their \cts\ counterparts (2sls, 2sls-x), respectively.

\item[$\bullet$] At $(\sigma_X, \sigma_\eta)=(0.5,0.5)$ with $\va/\ve = 0.125$, 
2sfe yields a shorter average confidence interval length than 2sls, but 2sls-x outperforms both 2sfe and 2sfe-x due to adjustment for cluster-level covariates.

\item[$\bullet$] At $(\sigma_X, \sigma_\eta)=(0.2,0.2)$ with $\va/\ve = 0.02$,  2sls and 2sls-x without cluster indicators yield shorter average confidence interval lengths than  2sfe and 2sfe-x, respectively. 
\ende

\begin{table}[!t]\caption{\label{tb:simu} Mean squared errors of the four \tsls\ procedures, along with the coverage rates and average lengths of the 95\% confidence intervals over 1,000 independent replications.
\{2sls, 2sfe\} index the unadjusted \cts\ and \tsfe, and \{2sls-x, 2sfe-x\} index their respective \ca\ counterparts.}
 
\resizebox{\textwidth}{!}{
\begin{tabular}{l|rrrr|rrrr|rrrr}
  \hline
  & \multicolumn{4}{c|}{$(\sigma_X, \sigma_\eta)=(1,1)$} & \multicolumn{4}{c|}{$(\sigma_X, \sigma_\eta)=(0.5,0.5)$}
  & \multicolumn{4}{c}{$(\sigma_X, \sigma_\eta)=(0.2,0.2)$} \\
    & \multicolumn{4}{c|}{$\va/\ve = 0.5$} & \multicolumn{4}{c|}{$\va/\ve = 0.125$}
  & \multicolumn{4}{c}{$\va/\ve = 0.02$} \\\hline
 & 2sls & 2sfe & 2sls-x & 2sfe-x
 & 2sls & 2sfe & 2sls-x & 2sfe-x 
 & 2sls & 2sfe & 2sls-x & 2sfe-x \\ 
  \hline
Mean squared error 
& 0.081 & 0.044 & 0.040 & 0.037
& 0.050 & 0.045 & 0.036 & 0.038
& 0.043 & 0.045 & 0.032 & 0.035 \\ 
Coverage rate 
& 0.957 & 0.951 & 0.947 & 0.947
& 0.949 & 0.952 & 0.944 & 0.941
& 0.950 & 0.948 & 0.950 & 0.952 \\  
Average CI length 
& 1.134 & 0.831 & 0.790 & 0.743
& 0.889 & 0.835 & 0.726 & 0.746
& 0.795 & 0.825 & 0.699 & 0.738 \\ 
   \hline
\end{tabular}}
\end{table}

\subsection{Test and estimation with heterogeneous clusters}
We now illustrate the properties of the test for \ch. 
We fix the number of clusters at $G = 100$ and the cluster size at $n_g = 20$ for all $g$. 
Assume two types of clusters.
The first $50$ clusters, $g = \ot{50}$, are of the first type and have no \cfe. 
The remaining $50$ clusters, $g = 51, \ldots, 100$, are of the second type and have a \cfe\ equal to $\delta$, where $\delta$ is a tuning parameter governing the degree of heterogeneity. 
The clusters are homogeneous if and only if $\delta = 0$.
Define $\alpha_g = 1_{\{g > 50\}}\delta$ for $\otg$. 
When $\delta \neq 0$, $\ag$ represents heterogeneous \cfes\ across different cluster types. 

To facilitate illustration, we assume a constant treatment effect of $\tau = 0$ for all units. For each cluster $\otg$, we generate $(Z_i, D_i, Y_i)$ for $\ig$ as follows:  
\begini
\item $U_i$ are \iid\ with $\pr(\uc) = 0.7$ and $\pr(\un) = 0.3$; 
\item  $Z_i$ are independent $ \textup{Bernoulli}(e_g)$, where $e_g = 1/\{1 + \exp(-0.5 \alpha_g)\}$; 
\item $D_i = Z_i \cdot \ouc $; 
\item  $\yi = \yio = \yiz = \ag + \epi$, where
  $\epi \simiid\mn(0, 1)$. 
\endi
The definition of $e_g$ implies that clusters with larger $\alpha_g$ have higher IV probabilities, so the IV distributions are also heterogeneous across clusters.  

Figure~\ref{fig:t} shows the distributions of the $t$-statistic we proposed in \eqref{eq:t-stat} based on 1,000 independent replications for $\delta = 0, 1, 2$, along with the average values of $\htls$ and $\htfe$ across the three cases. 
For the case $\delta = 0$, we additionally overlay the standard normal density on the histogram.

\begin{figure}[!ht]
\begin{center} 
\begin{tabular}{ccc}
\includegraphics[width = .32\textwidth]{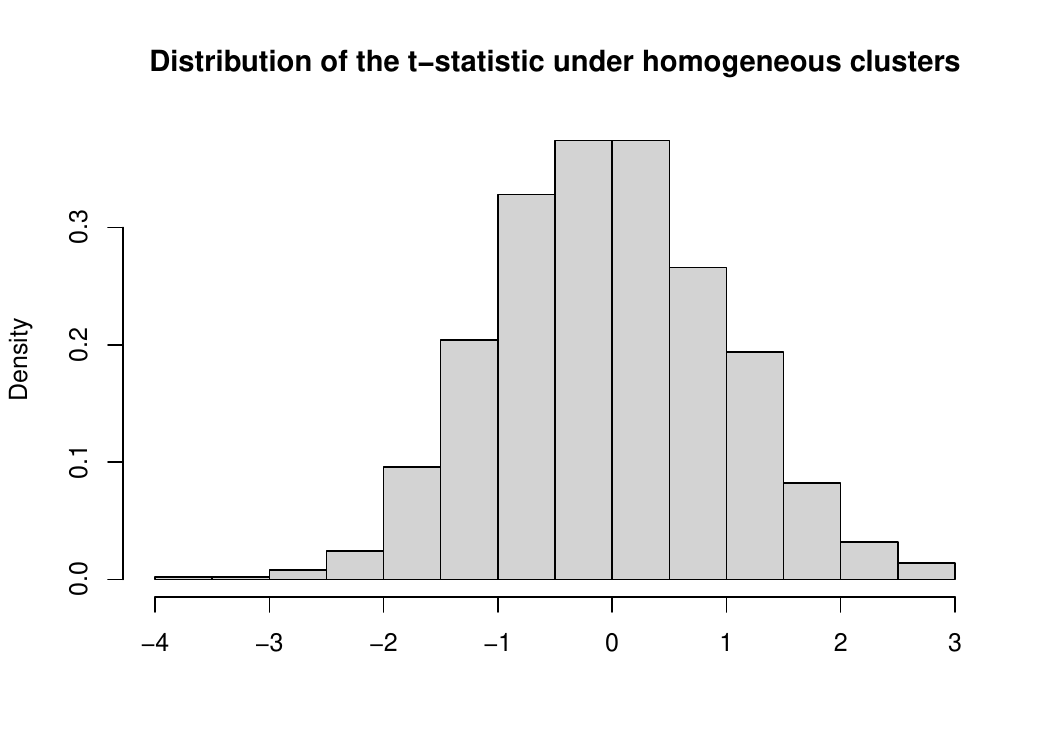}&
\includegraphics[width = .32\textwidth]{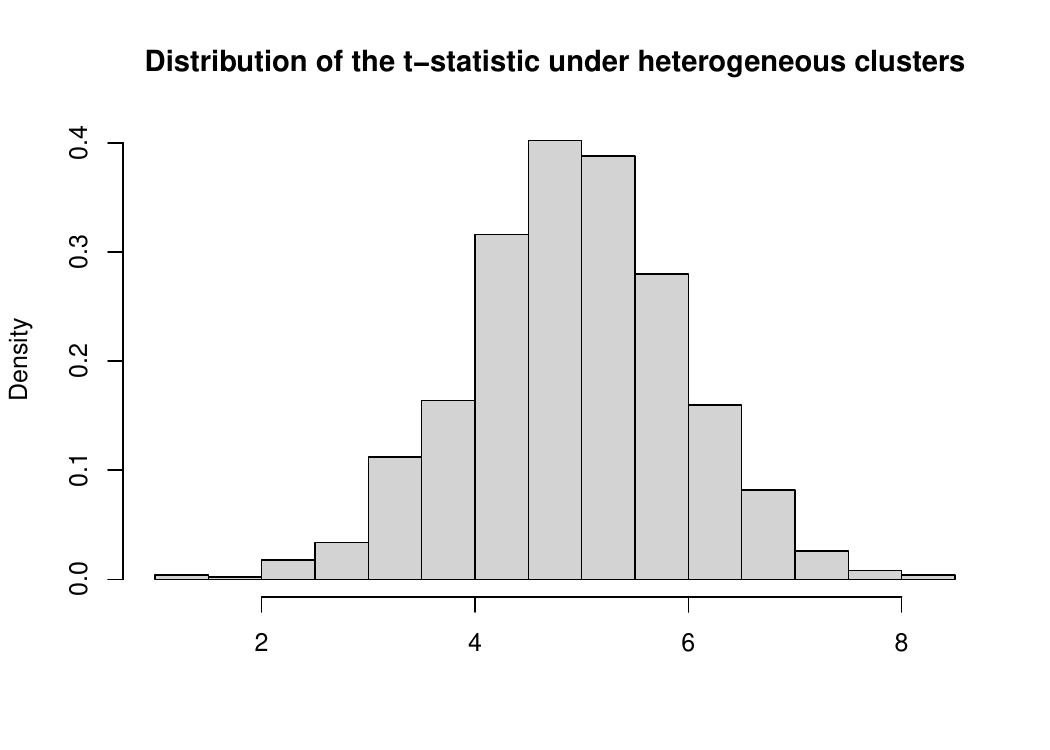} &
\includegraphics[width = .32\textwidth]{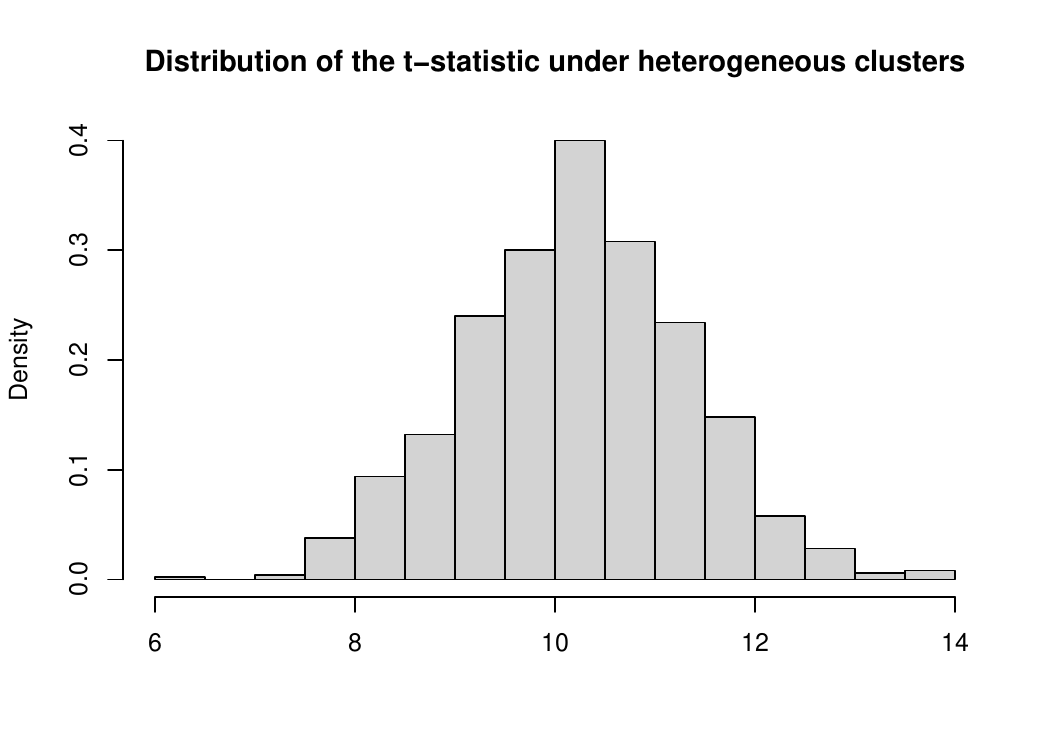} \\
(a) Homogeneous clusters & (b) Heterogeneous clusters & (c) Heterogeneous clusters \\
$\delta = 0$ & $\delta = 1$ & $\delta = 2$ \\
Ave. $\beginp
\htls\\
\htfe
\endp = \beginp 0.0004\\ 0.0008\endp$ & Ave.  $\beginp
\htls\\
\htfe
\endp = \beginp 0.175\\ -0.002\endp$ & Ave.  $\beginp
\htls\\
\htfe
\endp = \beginp 0.695\\ -0.0005\endp$
\end{tabular}

\end{center}
\caption{\label{fig:t} Distributions of the $t$-statistic we proposed in \eqref{eq:t-stat} under homogeneous and heterogeneous clusters over 1,000 independent replications. The curve in Figure~\ref{fig:t}(a) represents the standard normal density. 
The table below reports the corresponding average values of $\htls$ and $\htfe$.}
\end{figure}

When $\delta = 0$, we have $\alpha_g = 0$ and $\eg = 0.5$ for all $g$, corresponding to homogeneous clusters. 
The resulting empirical distribution closely resembles the standard normal distribution, and 95.4\% of the $t$-statistics satisfy $|t_\diff| \leq 1.96$. Both observations are coherent with the theoretical results in \thmjoint, supporting the validity of the test.  
In addition, the averages of $\htfe$ and $\htls$  are both close to the true value, $\tau = 0$, coherent with the theoretical results in \propclt. 
The cases $\delta = 1$ and $\delta = 2$ introduce heterogeneity between the first and last 50 clusters, thereby violating \assmim. 
When $\delta = 1$, the corresponding empirical distribution is centered around $4.93$,  deviating substantially from the standard normal distribution. 
99.7\% of the $t$-statistics satisfy $|t_\diff| > 1.96$, illustrating the power of the test under violations of \assmim. 
When $\delta = 2$, the corresponding empirical distribution is centered around 10.23,  deviating even further from the standard normal distribution. 
100\% of the $t$-statistics satisfy $|t_\diff| > 1.96$, indicating increased power as the degree of heterogeneity increases. 
In both cases, the average values of $\htfe$ are close to the true value, $\tau = 0$, whereas those of $\htls$ exhibit clear biases.

\section{Application}
We now illustrate our methods using data from a randomized evaluation of a microcredit program introduced in rural Morocco in 2006 \citep{crepon2015estimating}. 
The sample includes 5,898 households (units) nested within 162 villages. 
Villages were matched into 81 pairs based on observable characteristics, and within each pair, one village was randomly assigned to treatment, and the other to control. 
After randomization, credit agents from the partner microfinance institution started to promote microcredit in treatment villages.

Following \citet[Section~III.C]{crepon2015estimating}, 
we estimate the impact of microcredit take-up on outcomes using the subsample of 4,934 households with a high predicted probability of borrowing, as determined by the propensity score model in \cite{crepon2015estimating}.    
The outcome of interest, $Y_i$, is total household production from agriculture, livestock, and non-agricultural business; see \citet[Table 3]{crepon2015estimating}. 
The treatment $D_i$ is an indicator of microcredit take-up at endline. 
The IV $Z_i$ is an indicator of residence in a treated village.
Covariates $\xxi$ include household size, number of adults, age of the household head, baseline borrowing, and other household characteristics.
In total, 2,448 households (49.6\%) in the sample reside in treated villages ($\zi = 1$), and 410 households (8.3\%) took up microcredit ($\di = 1$).
Following \cite{crepon2015estimating}, we define each village pair as a stratum (cluster) when constructing the cluster indicators $\cci$.
This choice is natural because the matched pair is the randomization block that absorbs the baseline differences used in matching, and aligns with the cluster sampling assumption, which requires IV assignment to be independent across clusters. 
The $t$-statistic from the \ch\ test is 0.183, so we cannot reject \assmim. 

\begin{figure}[!ht]
\centering 

\includegraphics[width = .8\textwidth]{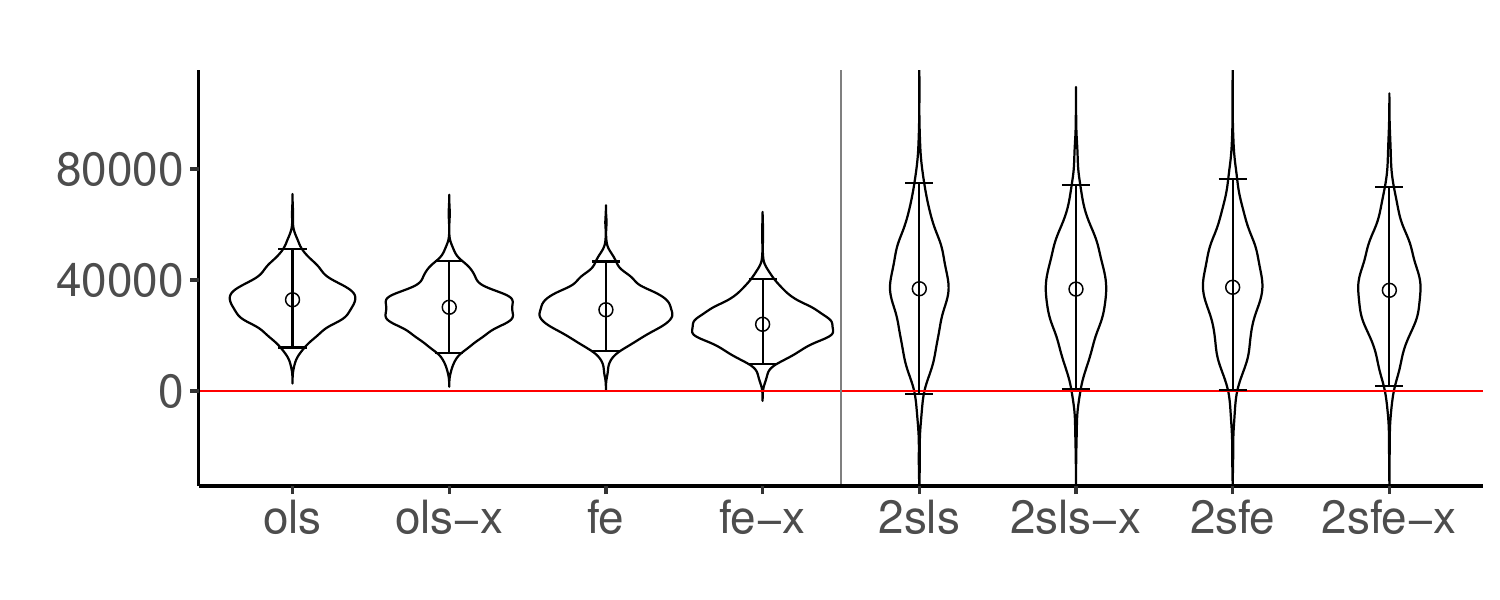} 

\begin{tabular}{l|rrrr}
  \hline
  & 2sls & 2sls-x& 2sfe & 2sfe-x\\  \hline
Point estimate  & 36806.64 & 36689.22& 37342.94 & 36252.57 \\ 
Cluster-robust standard error  & 18882.51 & 17601.90 & 18745.43 & 17508.58\\ 
95\% CI Low & -202.41  & 2190.14 &  602.57 & 1936.39 \\ 
95\% CI High & 73815.69 & 71188.31  & 74083.30 & 70568.75 \\ 
   \hline
\end{tabular}

\caption{\label{fig:app} Violin plots of the cluster bootstrap distributions of the eight point estimators over 1,000 replications. The width of each violin represents the empirical probability density, with wider sections indicating higher data concentration. Circles ($\circ$) denote point estimates from the original data. 
Error bars indicate 95\% confidence intervals computed from the \crses.}
\end{figure}

In addition to the \cts\ and \tsfe\ and their \ca\ counterparts, we also consider their \ols\ counterparts, indexed by \{ols, fe, ols-x, fe-x\}. 
Figure~\ref{fig:app} illustrates the point estimates from the eight procedures, 
along with the 95\% confidence intervals using \crses\ and their cluster bootstrap distributions over 1,000 replications. 
The four \olsc\ procedures (ols, fe, ols-x, fe-x) return smaller point estimates than their \tsls\ counterparts.
The table below reports the point estimates, associated \crses, and 95\% confidence intervals from the four \tsls\ procedures.
The \ca\ \tsfe, 2sfe-x, gives the smallest \crse\ and the narrowest confidence interval.
Note that the confidence interval associated with the unadjusted \cts\ includes zero, whereas the other three confidence intervals do not. This illustrates the efficiency gains from covariate adjustment and the inclusion of cluster indicators in this application.

\section{Discussion}\label{sec:discussion}
We examined the \cts\ and \tsfe\ as two approaches to IV analysis for clustered data when both the treatment and IV vary within clusters. 
Our main contributions are threefold. 
First, under homogeneous clusters, we established the validity of both procedures, and compared their relative efficiency when the potential outcomes model includes \cfes. 
Our results caution against the routine inclusion of cluster indicators without considering efficiency implications. 
Second, under heterogeneous clusters, we showed that the point estimator from \tsfe\ recovers a weighted average of cluster-specific LATEs, whereas that from the \cts\ generally does not. This suggests an advantage of \tsfe\ when analyzing data nested within heterogeneous clusters, such as stores in different geographic locations or firms of differing sizes.
Finally, motivated by the distinct asymptotic behavior of $\htls - \htfe$ under homogeneous and heterogeneous clusters, we developed an asymptotic theory for the joint distribution of $\htls$ and $\htfe$ when clusters are homogeneous, and proposed a test for \ch. 
We summarize the trade-off between \cts\ and \tsfe\ in Table \ref{tb:summary}

\begin{table}\caption{\label{tb:summary}Summary of the trade-off between \cts\ and \tsfe.}
\begin{center}
\resizebox{\textwidth}{!}{
\begin{tabular}{l|c|c}\hline
& Canonical \tsls & \tsfe \\\hline
Cluster-constant treatment, IV, covariates & Yes & No \\\hline
Requires cluster-robust standard errors? & Yes & Yes \\\hline
Efficiency under homogeneous clusters & \multicolumn{2}{l}{
\begin{tabular}{l}
Under \assmy, \tsfe\ is more efficient than\\ the \cts\ if and only if \eqref{eq:eff_cutoff} holds.
\end{tabular}} 
\\\hline
Efficiency gain from cluster-level covariates? & Yes & No\\\hline
\begin{tabular}{l}
Identifies weighted average of cluster-specific LATEs\\ under heterogeneous clusters?
\end{tabular} & No in general & Yes \\\hline
\end{tabular}
}
\end{center}
\end{table}

\medskip

Ordinary least squares (\ols) is a special case of \tsls\ with $Z_i = D_i$. 
Our theory therefore encompasses the comparison between the \ols\ regression 
of $Y_i$ on $(1,\di)$ and that of $Y_i$ on $(\di,\cci)$ with cluster indicators---two standard specifications for analyzing clustered data under exogenous treatment \citep{bugni2018inference}. 

\subsubsection*{Data availability statement}
The data that support the findings of this study are openly available in the replication package of \cite{crepon2015estimating} at https://www.aeaweb.org/articles?id=10.1257/app.20130535.

\bibliography{refs_cIV-FE}

\spacingset{1.5}

\newpage
\setcounter{equation}{0}
\setcounter{section}{0}
\setcounter{figure}{0}
\setcounter{example}{0}
\setcounter{proposition}{0}
\setcounter{corollary}{0}
\setcounter{theorem}{0}
\setcounter{table}{0}
\setcounter{condition}{0}
\setcounter{definition}{0}
\setcounter{assumption}{0}
\setcounter{lemma}{0}
\setcounter{remark}{0}

\renewcommand {\thedefinition} {S\arabic{definition}}
\renewcommand {\theassumption} {S\arabic{assumption}}
\renewcommand {\theproposition} {S\arabic{proposition}}
\renewcommand {\theexample} {S\arabic{example}}
\renewcommand {\thefigure} {S\arabic{figure}}
\renewcommand {\thetable} {S\arabic{table}}
\renewcommand {\theequation} {S\arabic{equation}}
\renewcommand {\thelemma} {S\arabic{lemma}}
\renewcommand {\thesection} {S\arabic{section}}
\renewcommand {\thetheorem} {S\arabic{theorem}}
\renewcommand {\thecorollary} {S\arabic{corollary}}
\renewcommand {\thecondition} {S\arabic{condition}}
\renewcommand {\thepage} {S\arabic{page}}
\renewcommand {\theremark} {S\arabic{remark}}

\renewcommand{\theHassumption}{S\arabic{assumption}}
\renewcommand{\theHlemma}{S\arabic{lemma}}
\renewcommand{\theHsection}{S\arabic{section}}

\setcounter{page}{1}
 
\section*{\centering Supplementary Material}
\sec~\ref{sec:notation} summarizes the key notation. 
\sec~\ref{sec:complete theory} states the complete theory for \ca\ \cts\ and \tsfe.  
\sec~\ref{sec:lem} states the lemmas.
\sec~\ref{sec:proof_homo} provides the proofs of the results in \sec~\ref{sec:homo} of the main paper. 
\sec~\ref{sec:proof_hetero}--\ref{sec:app_proof_of_joint} provide the proofs of the results in \sec~\ref{sec:hetero} of the main paper.
\sec~\ref{sec:proof_complete} provides the proofs of the results in Section~\ref{sec:complete theory}.  

\section{Notation}\label{sec:notation}
For positive integer $m$, 
let $1_m$ and $0_m$ denote the $m\times 1$ vectors of ones and zeros, respectively; let  
$I_m$ denote the $m\times m$ identity matrix, 
$J_m = 1_m 1_m^\T$ denote the $m\times m$ matrix of ones, and  
$
P_m  = I_m  - m^{-1} J_m$ denote the $m\times m$ projection matrix with $
P_m^2 = P_m$ and $P_m 1_m = 0_m$.  
Let $\|\cdot\|$  denote the \enf, with $\|a\| = (a^\T a)^{1/2}$ for a vector $a$.
Let $\lm(\cdot)$ denote the smallest eigenvalue of a symmetric matrix. 
Let $\spn{\cdot}$ and $\fn{\cdot}$ denote the spectral norm and Frobenius norm. 
For sequences of vectors or matrices $(a_n)_{n=1}^\infty$, we use $a_n = \op$ to denote componentwise convergence in probability to zero, and $a_n = O(1)$ and $a_n = \oop$ to denote componentwise boundedness and boundedness in probability.  

For a random variable $Y \in \mbr$ and a random vector $X \in \mbr^n$, let 
$\proj(Y\mid X)$ denote the linear projection of $Y$ onto $X$, defined as $\proj(Y\mid X) = b_0^\T X$, where $b_0 = \argmin_{b\in \mbr^{n}} \E\{(Y -  b^\T X)^2\} = \{\E(X X^\T)\}^{-1} \E(X Y)$; 
let $
\res(Y\mid X) = Y - \proj(Y\mid X)$ denote the corresponding projection residual.

\paragraph*{Clustered population.}
For a scalar or vector $a_i  \in \mbr^{p}$ defined for $\otn = \mif$, denote by 
\begini
\item $\bar a = \meani a_i \in \mbr^p$ the population average; 
\item $\bar a_g = \meanig a_i \in \mbr^p$ the within-cluster average for cluster $g$;
\item $a_\sg = \vjh{a} \in \mbr^{\ng\times p}$ the stacked vector or matrix of $a_i$ for cluster $g$; 
\item $a = \vgh{a} = \vgj{a} \in \mbr^{N\times p}$ the stacked vector or matrix of $a_i$ for all units ordered by clusters;
\item $\dsai = \sai - \bsag \in \mbr^p$ the \ccd\ variant of $\sai$; 
\item 
$\dsag = \vjh{\dot a} = \sag - \og \bar a_g^\T \in \mbr^{\ng\times 1}$ the \ccd\ variant of $a_\sg$. 
\endi
In particular, 
\begineqs
\beginar{rclclcl}
\bar Y = \meani Y_i, &\quad \byg = \dmeanig \yi, \\
\yg = \vjh{Y}, & \quad Y = \vgh{Y}, &\quad \dyg = \yg - \og\byg.
\endar
\endeqs
Similarly, define $(D, \bdg, \ddg)$, $(Z, \bzg, \dzg)$, and $(X, \bxg, \dxg)$. 
Let 
\begineq\label{eq:sz_szy_szd}
\beginar{rcl}
\sz &=&\dmeani(\zi - \bz_{\ci})^2 = \dmeani\dzi^2,\\ 
 \szy &=& \dmeani \dzifp\dyifp = \dmeani\dzi\dyi,\\
 \szd &=&\dmeani \dzifp\ddifp =\dmeani \dzi\ddi
\endar
\endeq
be the sample within-cluster variances and covariances, normalized by $N^{-1}$. 
Let
\beginy\label{eq:szg}
\szt = \dmeani(Z_i - \bz)^2, \quad \szg = \dsumig (Z_i - \bzg)^2 = \dsumig\dzi^2, 
\endy
with $\sz = \nisumg \szg.$
Let 
\begineq\label{eq:szyt_szdt}
\szyt = \szytf, \quad \szdt = \szdtf.
\endeq 

\paragraph*{Notation under \assmtth.} Under \assmim, let  
\begineqs
\begin{array}{l}
\beginar{l}
e = \pr(Z_i = 1)= \E(\zi), \quad \vz = \var(Z_i) = e(1-e),
\endar
\smallskip\\
\beginar{lllll}
\muy = \E(\yi), &\quad   \muzy = \E(\ziyi), & \quad\mud = \E(\di), \qquad \muzd = \E(\zidi),\smallskip\\
\mux = \E(\xxi), & \quad  \vx = \cov(\xxi) , & \quad\muxx = \E(\xxit) = \vx + \mux\muxt,\smallskip\\
\muxz = \E(\xxi\zi), &\quad \mudx = \E(\xxi\di), &\quad \pa = \pr(\ua),  \quad  \pc = \pr(\uc) 
\endar
\end{array}
\endeqs
be shorthand notation for moments of the common distribution.
Under \assmtth, $\pc > 0$ so that $\tc = \E(\ti \mid \uc)$ is well defined. Let 
\begineq\label{eq:ai}
\ai = \yi - \di \tc, \quad
\mua = \E(\ai) = \muy-\mud\tc, \quad \dai = \dyi - \ddi\tc.
\endeq

\section{Central limit theorems for $\htlsx$ and $\htfex$}
\label{sec:complete theory}
We provide in this section the central limit theorems for $\htlsx$ and $\htfex$ from the \ca\ \cts\ and \tsfe. 
Under \assmtth, let
\begineq\label{eq:rrilsx_def}
\bxaa =\covxi \cov(\xxi, \ai), \quad 
\rrilsx = \res(\ai\mid 1, \xxi)  
\endeq
denote the coefficient vector of $\xxi$ in $\proj(\ai\mid 1, \xxi)$ and the corresponding residual, with 
\begineq\label{eq:rrilsx_decomp}
\beginar{rcccl}
\proj(\ai \mid 1, \xxi) &=& \mua +  (\xxi - \mux)^\T\bxaa &=& \bo + \xit\bxaa,\\
\rrilsx &=&  A_i - \bo - \xit\bxaa &=&  Y_i - \bo - \di\tc - \xit\bxaa,
\endar
\endeq  
where $\bo =  \mua - \muxt\bxaa$.
Let 
\begineq\label{eq:rrifex_def}
\beginar{lll}\sx = \dmeani \dxxi, &\quad&\sxa = \dmeani \dxai, \\ 
\gxa = \gxaf, &\quad&\rrifex = \dai -\dxit\gxa.
\endar  
\endeq
Let
\begineq\label{eq:omgnlsx_omgnfex_def}
\omgnlsx = \ni\cov\left\{\dsumi \wilsxf \rrilsx\right\},\quad 
\omgnfex = \ni\cov\left\{\dsumi \wifexf \rrifex\right\}. 
\endeq
\assmasym\ below states the regularity condition we assume.
\begin{assumption}\label{assm:asym}
As $\ntinf$,
\begine[(i)]
\item\label{it:assm_asym_123} \assmhomo\ and \ref{assm:ng_main} hold;  
\item\label{it:assm_asym_Y} $\E(Y_i^4) < \infty$; 
\item\label{it:assm_asym_lsx} $\E(\|X_i\|^4) < \infty$;\quad \assmlm{\omgnlsx}; 
\item\label{it:assm_asym_fex} 
there exists a constant $s > 2$ such that $\E(\|X_i\|^{2s}) < \infty$; 

$\inf_N \kn >0$; \quad \assmlm{\omgnfex}; \quad\assmlm{\esx}.
\ende
\end{assumption}

\assmasym\eqref{it:assm_asym_123} ensures that $\omgnlsx$ and $\omgnfex$  in  \eqref{eq:omgnlsx_omgnfex_def} are well defined. 
We assume $r = 2$ to simplify the proofs. 
By \hansen, this can be relaxed to $r \in [2,\infty)$; see \assmng\ in the supplemental appendix for details. 
\assmasym\eqref{it:assm_asym_lsx} and \eqref{it:assm_asym_fex} state the rank conditions under cluster sampling for the \ca\ \cts\ and \tsfe, respectively. 

\thmlsxfex\ below give the explicit forms and central limit theorems for $\htlsx$ and $\htfex$.
Let $\zix$ denote the residual from the least-square regression of $\zi$ on $(1,\xxi)$.
Let $\zicx$ denote the residual from the least-square regression of $\zi$ on $(\cci,\xxi)$.  
Let 
\begina
\szdx = \meani \zix \di, &\quad& \szyx = \meani \zix \yi,\\
\szdcx = \meani \zicx \di, &\quad& \szycx = \meani \zicx \yi.
\enda
Denote by $(\htlsx, \hselsx, \rilsx)$ and $(\htfex, \hsefex, \rifex)$ the coefficients on $\di$, corresponding {\crse}s, and residuals from the \ca\ \cts\ and \tsfe, respectively.

\begin{theorem}\label{thm:lsx}
Consider the \tslsxf, $\tslst(\yi \sim 1 + \di + \xxi \mid 1 + \zi + \xxi)$. 
\begine[(i)]
\item\label{it:lsx_numeric} $
\htx = \dfrac{\szyx}{\szdx}, \quad 
\hsesq_\lsx = \dfrac{1}{N^2}\cdot\dfrac{1}{\szdx^2} \dsumg \left(\dsumig \zix \cdot \rilsx\right)^2.$

\item\label{it:lsx_clt} As $\ntinf$, if \assmasymlsx\ hold, then 
\begina
\vlsxn^{-1/2} \cdot \sqrtn(\htx - \tc)\rs \sn, \\
N\hsexsq/\vlsxn = 1+\op,
\\
(\htx - \tc)/ \hsex \rs \sn,
\enda
where $
\vlsxn =   \dfrac{1}{\vzpcsq}\cdot \nisumg \var\{\sumig (\zi-e)\rrilsx\}$
with $\rrilsx = \res(\aai \mid 1, \xxi)$ as defined in \eqref{eq:rrilsx_def}.
\ende
\end{theorem}

\begin{theorem}\label{thm:fex}
Consider the \tsfexf, $
\tsfexformula$.  
\begine[(i)]
\item\label{it:fex_numeric} $\htfex = \dfrac{\szycx}{\szdcx}$,
\quad
$\hsefexsq = \dfrac{1}{N^2}\cdot\dfrac{1}{\szdcx^2} \dsumg \left(\dsumig \zicx \cdot \rifex\right)^2$. 
\item\label{it:fex_clt}
As $\ntinf$, if \assmasymfex\ hold, then  
\begina
\vfexn^{-1/2} \cdot \sqrtn(\htfex - \tc)  \rs \sn, \\
N\hsefexsq/\vfexn  = 1+\op, \\
  (\htfex - \tc) / \hsefex  \rs \sn,&&
\enda
where $\vfexn =    \dfrac{1}{(\kvzpc)^2}\cdot \nisumg \var(\sumig\dzi \rrifex)$
with $\rrifex = \dai - \dxit\gxa$ as defined in \eqref{eq:rrifex_def}. 
\ende
\end{theorem}

The unadjusted \cts\ and \tsfe\ are special cases of \thmlsxfex, where $\xxi = \emptyset$.  
We have $\zix = \zi-\bz$ and $\zicx =\zi- \bzci$, with
\begineq\label{eq:num}
\begin{array}{lll}
\htls = \dfrac{\szyt}{\szdt},&\quad&
\hsesqls =  \dfrac{1}{N^2} \cdot \dfrac{1}{\szdt^2} \dsumg \left\{\dsumig (\zi - \bz) \rils\right\}^2,\\
\htfe   = \dfrac{\szy}{\szd}, &\quad& \hsesqfe   = \dfrac{1}{N^2}
 \cdot \dfrac{1}{\szd^2}  \dsumg \left\{\dsumig(Z_i - \bzg) \rife  \right\}^2, 
 \end{array}
\endeq 
where $\szyt$, $\szdt$, $\szy$, and $\szd$ are the sample total and within-cluster covariances as defined in \eqref{eq:sz_szy_szd} and \eqref{eq:szyt_szdt}.

\section{Lemmas}\label{sec:lem}
\subsection{Review of \hansen}\label{sec:hansen}
We review below the asymptotic theory for clustered units from \hansen.  
To simplify notation, we follow \hansen\ in omitting the sample size $N$ from the subscript when no confusion is likely to arise. \assmng\ below reviews \citet[Assumption~2]{\hansenid} on cluster sizes for the central limit theorem for the sample mean. %
\assm~\ref{assm:ng_main} is the most restrictive case of \assmng\ with $r = 2$.  

\begin{assumption}\label{assm:ng}
As $N \to \infty$, (i) there exists a constant $r \in [2, \infty)$ such that\\ $\limsup_{N\to\infty} N^{-1}(\sumg n_g^r)^{2/r} < \infty$; (ii) $\ds\max_{\otg}  \,  {\ngsq}/N \to 0$.
\end{assumption}

\lemsm\ below reviews the asymptotic theory for sample mean in \citet[\thms~1--3]{\hansenid} that allows for heterogeneous clusters.
A collection of random variables $\{B_n \in \mbr: n \in \mathcal C\}$ is {\it \uni} if  
\begineq\label{eq:uni_def}
\ds\lim_{M\to\infty}\sup_{n \in \mathcal C} \E\left( |B_n| \cdot 1_{\{|B_n| > M\}}\right) = 0.
\endeq 
A sufficient condition for \eqref{eq:uni_def} is that $\sup_{n \in \mathcal C} \E(|B_n|^r) <\infty$ for a constant $r > 1$.

\begin{lemma}\label{lem:hansen_sm}
Let $\bi \in \mbr^p$ be a $p \times 1$ random vector defined for $\otn$. Let $\bbn = \meani \bi$ and $\omgn = \cov(\sqrtn\bbn)$.  
Let $\tomgn = \nisumg \tbg\tbg^\T$, where $\tbg = \sumig \bi$.
Assume \assmcs. As $\ntinf$, 
\begine[(i)]
\item\label{it:lem_sm_wlln} \citep[\thm~1]{\hansenid}: If $\ds\max_{\otg} \ng/N \to 0$, and  
$\{\|\bi\|: \otn\}$ is \uni, 
then $\|\bbn - \E(\bbn)\| = \op$. 
\item\label{it:lem_sm_clt} \citep[\thm~2--3]{\hansenid}: If (a) \assmng\ holds for a constant $r \in [2,\infty)$,  (b) $\{\|\bi\|^r: \otn\}$ is \uni, and (c) \assmlm{\omgn}, 
then $\omgn^{-1/2}\cdot\sqrtn \left\{\bbn - \E(\bbn)\right\} \rs \mn(0_p, I_p).$

Further assume that $\E(\bi) = 0$ for all $\otn$. Then 
\begineqs
\omgn^{-1/2}\tomgn\omgn^{-1/2} = I_p + \op,\quad
\tomgn^{-1/2}\cdot\sqrtn \left\{\bbn - \E(\bbn)\right\}  \rs \mn(0_p, I_p).
\endeqs
\ende
\end{lemma}


\lemtsls\ below reviews the asymptotic theory for general \tsls\ in \citet[Theorem 9]{\hansenid} in the context of just-identified cases.
For $\otn$, let 
$u_i \in \mbr$ denote the dependent variable, 
$v_i \in \mbr^p$ denote the $p\times 1$ regressor vector,
and $w_i  \in \mbr^p$ denote the $p\times 1$ IV vector, respectively.
Let $
u_\sg = \vjh{u}$, $\vsg = \vjh{v}$, and $\wsg = \vjh{w}$ 
denote the concatenations of $(u_i, \vi, \swi)$, respectively, for units within cluster $g$. 
Assume that 
\beginy\label{eq:model_hansen}
u_\sg = \vsg \beta + \epg,\quad \vsg = \wsg \gamma + \delta_\sg,\quad 
\E(\wsg^\T \epsilon_\sg) = 0,
\endy
where $\epsilon_\sg = \vjh{\epsilon}$ is an $\ng \times 1$ error vector with $ \wsg^\T \epsilon_\sg = \sumig \wi \epi $.  
Let $\hb\in \mbr^p$ denote the estimated \coeffv\ of $v_i$ from $\tslst(u_i\sim v_i \mid w_i)$. 
Let $\hsign = N\widehat{\cov}(\hb)$ denote the cluster-robust covariance of $\hb$ scaled by a factor of $N$; c.f. \defcrse.
Let 
\beginy\label{eq:hansen_def}
\begin{array}{rclcl}
\ds\gn &=& \ds\nisumg \E(\wsgt\vsg),\quad
\ds\psin 
= \ds\nisumg \E(\wsgt\wsg), \\
\ds\omgn 
&=& \dnisumg \E(\wsgt \epg \epgt \wsg) 
\oeq{\eqref{eq:model_hansen}} \dnisumg \cov(\wsgt \epg), \\
\sign &=&   \gn^{-1}   \omgn (\gnt)^{-1}.
\end{array}
\endy

\begin{lemma}\label{lem:hansen_2sls}
\citep[\thm~9]{\hansenid} If as $\ntinf$,
\begine[(i)]
\item\label{it:hansen_cs} \assmcs\ holds; 
\item\label{it:hansen_ng} 
\assmng\ holds for a constant $r \in [2,\infty)$;
\item\label{it:hansen_rank} $\gn$ has full rank $p$;
\item\label{it:hansen_lambda_min}  \assmlm{\psin}, \quad \assmlm{\omgn}; 
\item\label{it:hansen_bounded} there exists a constant $s \in [r, \infty)$ so that $
\ds \supi\E(|u_i|^{2s}) < \infty$, $\ds \supi\E(\|v_i\|^{2s}) < \infty$, $\ds \supi\E(\|w_i\|^{2s}) < \infty$;  
\item\label{it:hansen_identical} either $(u_i, v_i, w_i)$ have identical marginal distributions or  $r < s$,
\ende
then for any sequence of full-rank $p\times q$ matrices $C_N$,
\begina
(\ccnt \sign\ccnt)^{-1/2} \ccnt \sqrtn (\hb - \beta) &\rs& \mn(0_q, I_q), \\
(\ccnt \sign\ccnt)^{-1/2} \ccnt \hsign\ccn(\ccnt \sign\ccnt)^{-1/2} &=& I_q + \op,\\
(\ccnt \hsign\ccnt)^{-1/2} \ccnt \sqrtn (\hb - \beta) &\rs& \mn(0_q, I_q).
\enda 
The standard errors for elements of $\ccnt \hb$ can be estimated by taking the square roots of the diagonal elements of $\ni\ccnt \hsign\ccnt$. 
\end{lemma}

\subsection{Basic lemmas and useful facts}
Recall that $\png$ denotes the $\ng\times \ng$ projection matrix. For scalars $\sai, \sbi \in \mbr$ defined on $\otn$ and $\dsai = \sai - \bsag$ and $\dsbi = \sbi - \bsbg$, a useful fact is that 
\begineq\label{eq:apb}
 \sumig \dsai \dsbi    = (\png\sag)^\T(\png \sbg) =  (\sagt\png)  \sbg = \dsumig \dsai  \sbi. 
\endeq

\begin{lemma}
\label{lem:block_mat}
Let $\abcd$ denote a block partition of a matrix. If $A$ is invertible, then 
\begineqs
\abcd^{-1} = \beginp
A^{-1} + A^{-1}B S_A^{-1} C A^{-1}
& - A^{-1}B S_A^{-1}\\
- S_A^{-1}C A^{-1} & S_A^{-1}
\endp, \quad \det\abcd = \det(A) \cdot \det(S_A),
\endeqs
where $S_A = D - C A^{-1} B$.
\end{lemma}



\begin{lemma}\label{lem:bound}
\begine[(i)]
\item\label{it:lem_bound_1} For $b_1, \ldots, b_m \in \mbr^p$ and $q \geq 1$, we have \\ $
\left\|\dfrac{b_1 + \cdots + b_m}{m}\right\|^q \le \left(\dfrac{\|b_1\| + \cdots + \|b_m\|}{m}\right)^q \le \dfrac{\|b_1\|^q + \cdots + \|b_m\|^q}{m}.$
\item\label{it:lem_bound_2} 
Let $B_1,\dots,B_n \in \mbr^p$ be identically distributed $p\times 1$ random vectors, and let $\bb = \meanin \bi$. For $q \geq 1$, we have  $\|\bi - \bb\|^q \leq
2^{q-1}(\|\bi\|^q +   \meanin \|\bi\|^q)$, 
\begineqs
\E\left(\|\bi - \bb\|^q\cdot 1_{\{\|\bi-\bb\| > M\}}\right) 
\le \E\left(\|\bi - \bb\|^q\right)
\le 2^q \cdot \E\left(\|\bi\|^q\right). 
\endeqs
\item\label{it:lem_bound_3}
For a random vector $B$ and constants $0 < p < q$,  $\left\{\E(\|B\|^p)\right\}^{1/p} \leq \left\{\E(\|B\|^q)\right\}^{1/q}$. 
\ende 
\end{lemma}

\begin{lemma}\label{lem:basic}
\begine[(i)]\item\label{it:lem_basic_as}
For a nonnegative random variable $B\geq 0$, $
\E(B) = 0$ if and only if $B = 0$ almost surely.
\item\label{it:lem_basic_oop} 
Let $(B_n)_{n=1}^\infty$ be a sequence of random variables. As $n \to \infty$,
(a) if $
B_n \rs b$
for a constant $b$, then 
$
B_n =  b +\op$;  
(b) if $\ds \limsup_{n\to\infty}\E(|B_n|) < \infty$, then $B_n = \oop$. 
\ende 
\end{lemma}

\medskip

\begin{lemma}\label{lem:lm}
Let $A \in \mbr^{p \times p}$ be a symmetric positive semidefinite matrix. 
Let $C \in \mbr^{p \times q}$ be a matrix, and let $\|C\|_2$ denote the largest singular value of $C$, also known as the spectral norm, with $\|C\|^2_2 = \lmax(\ctc)= \lmax(C C^\T)$.  
Then 
\begine[(i)]
\item\label{it:lem_lm_Cx}  $\|Cx\| \le \|C\|_2 \, \|x\|$ for any $q\times 1$ vector $x$; 
\item\label{it:lem_lm_CAC}
$\lm(\cac) \in \big[\lm(A), \lmax(A)\big]\cdot \lm(\ctc)$. 
\ende
\end{lemma}

\subsection{\fwlf\ type results for 2{\normalsize SLS}}
\lem~\ref{lem:basu} below reviews \citet[Theorem 1]{basu2024frisch} on the Frisch--Waugh--Lovell (FWL) type results for \tsls.

\begin{lemma}\label{lem:basu}
Let $\dis, \zis\in \mbr^q$ and $W_i \in \mbr^p$ be  vectors defined for $\otn$.
Consider the \tsls\ regression of $\yi \in\mbr$ on $(\dis, \wi)$ instrumented by $(\zis, \wi)$, 
\beginy\label{eq:tsls_w}
\tslst(\yi \sim \dis + \wi \mid \zis + \wi). 
\endy
Let $\htau$ denote the \coeffv\ of $\dis$, and let $\ri \ (\otn)$ denote the IV residuals. 

Let $(\yiw, \disw, \zisw)$ denote the residuals from $\olst(\yi \sim \wi)$, $\olst(\dis \sim \wi)$, and $\olst(\zis \sim \wi)$, respectively.
Consider the \tslsr\ of $\yiw$ on $\disw$ instrumented by $\zisw$,
\beginy\label{eq:tsls_fwl}
\tslst(\yiw \sim \disw \mid \zisw). 
\endy
Let $\htfwl$ denote the \coeffv\ of $\disw$, and let $\riw \ (\otn)$ denote the IV residuals. 
Then $
\htau = \htfwl$, and $r_i = \riw$ for $\otn$. 
\end{lemma}

\lem~\ref{lem:fwl} below builds on \lembasu, and gives the explicit forms of the \tsls\ coefficients and cluster-robust standard errors when $\dis$ and $\zis$ are scalars.

\begin{lemma}\label{lem:fwl}
Assume the setting of \lembasu\ with $q = 1$. 
Let $\hse$ denote the \crse\ of $\htau$ from \eqref{eq:tsls_w} by \defcrse. Let $\hsefwl$ denote the \crse\ of $\htfwl$ from \eqref{eq:tsls_fwl}. 
Let $\szdw = \meani \zisw \disw$ and $\szyw = \meani \zisw \yiw$. Then 
\begine[(i)]
\item\label{it:lem_fwl_explicit form} 
$\htau = \htw = \dfrac{\szyw}{\szdw}$,  
\quad 
$
\ds\hsesq = \hsefwl^2 
= \dfrac{1}{N^2}\cdot\dfrac{1}{\szdw^2} \sumg \left(\sumig \zisw \cdot \rifwl\right)^2. 
$
\item\label{it:lem_fwl_equivalent S}
$\szyw$ and $\szdw$ have the following equivalent forms: 
\begina
\begin{array}{lllll}
\szyw &=& \meani \zisw \yi &=& \meani \zis\yi - \hgwz^\T\meani \wi\yi,\\
\szdw &=& \meani \zisw \dis &=& \meani \zis\dis - \hgwz^\T\meani \wi\dis,
\end{array} 
\enda
where $
\hgwz = \left(\sumi\wi\wit\right)^{-1}\left(\sumi\wi\zis\right)$
denotes the \coeffv\ of $\wi$ from $\olst(\zis \sim \wi)$,  with $\zisw = \zis - \wit\hgwz$.
\ende
\end{lemma}

\begin{proof}[Proof of \lem~\ref{lem:fwl}] 
We verify below \lemfwl\eqref{it:lem_fwl_explicit form}--\eqref{it:lem_fwl_equivalent S}, respectively.

\paragraph*{Proof of \lemfwl\eqref{it:lem_fwl_explicit form}.}
\lembasu\ ensures that $\htau = \htw$ and $\ri = \rifwl \ (\otn)$.
Given $\ri = \rifwl$, the equivalence between $\hse$ and $\hsefwl$ follows from the proof of \citet[Theorem 3]{ding2021frisch}. 
We verify below the explicit forms of $(\htw,\hsefwl)$. 

Consider the first stage of \eqref{eq:tsls_fwl}, 
$\olst(\disw \sim \zisw)$. Let $\hbzfwl$ denote the coefficient on $\zisw$. 
Properties of least squares ensure
\beginy\label{eq:hbzfwl}
\hbzfwl = \dfrac{\sumi\zisw\disw}{\sumi (\zisw)^2}
 = \dfrac{\szdw}{\szw}, \where \szw = \meani (\zisw)^2. 
\endy 
The fitted values are then
\beginy\label{eq:hdisw}
 \hdisw = \zisw \hbzfwl \oeq{\eqref{eq:hbzfwl}} \zisw \cdot \dfrac{\szdw}{\szw}. 
\endy 
The second stage of \eqref{eq:tsls_fwl} equals $\olst(\yiw \sim \hdisw)$, with $\htw$ equal to the coefficient on $\hdisw$. This implies
\begina
\htw = \dfrac{\sumi \hdisw\yiw}{\sumi (\hdisw)^2} \overset{\eqref{eq:hdisw}}{=} \dfrac{\hbzfwl\cdot \sumi \zisw\yiw}{\hbzfwl^2 \cdot \sumi (\zisw)^2} 
\oeq{\eqref{eq:hbzfwl}} \dfrac{ \szyw }{\hbzfwl  \cdot \szw}
 \overset{\eqref{eq:hbzfwl}}{=} \dfrac{ \szyw }{\szdw}. 
\enda

Let
\begina
\zgfwl = (Z^*_{g1 \fwl}, \ldots, Z^*_{g,\ng \fwl})^\T,\quad
\zfwl  = (\zst_{[1] \fwl}, \ \ldots, \ \zst_{[G] \fwl})^\T,\\
\dfwl = (D^*_{11 \fwl}, \ \ldots, \ D^*_{G,n_G \fwl})^\T,\\
\hdgfwl 
 = (\hat D^*_{g1\fwl}, \ldots, \hat D^*_{g,\ng\fwl})^\T,\quad
 \hdfwl
 = (\hdst_{[1] \fwl}, \ \ldots, \ \hdst_{[G] \fwl})^\T
 \enda
denote the vectorization of $\zisw$, $\disw$, and $\hdisw$ over cluster $g$ and all units, respectively, with  
\beginy\label{eq:fwl_zd}
\hdgfwl 
 \overset{\eqref{eq:hdisw}}{=} \zgfwl \hbzfwl,
 \quad
 \hdfwl \overset{\eqref{eq:hdisw}}{=} \zfwl \hbzfwl,
 \quad
 \szw  
 \oeq{\eqref{eq:hbzfwl}} \ni\zst_{\fwl}\zfwl
\endy
from \eqref{eq:hbzfwl}--\eqref{eq:hdisw}. 
This ensures $
\ni \hdfwlt \hdfwl
\overset{\eqref{eq:fwl_zd}}{=}  \ni \hbzfwl^2   \zst_{\fwl}\zfwl
 \oeq{\eqref{eq:fwl_zd}} \hbzfwl^2 \szw$ and $N \hdfwlsi  
= (\hbzfwl^2  \szw)^{-1}$, 
so that
\begineq\label{eq:fwl_se_ddd}
N \hdfwlsi \hdgfwl
\overset{\eqref{eq:fwl_zd}}{=} (\hbzfwl^2 \szw)^{-1}  \zgfwl \hbzfwl 
=(\hbzfwl \szw)^{-1} \zgfwl 
\oeq{\eqref{eq:hbzfwl}}\dfrac{\zgfwl}{\szdw}  .
\endeq
Let $\rgfwl = (\rifwl: \ig)$ denote the residual vector for cluster $g$.
Let $\ho = \diag(\ho_g)_{g=1}^G$, where $\ho_g= \rgfwl\cdot\rgfwlt$.
We have 
\beginy\label{eq:fwl_se_dod}
\hdfwlt \ho \hdfwl  
&=& (\hdst_{[1] \fwl}, \ldots, \hdst_{[G] \fwl})\beginp r_{[1] \fwl} \cdot r_{[1] \fwl}^\T \\ & \ddots \\ && r_{[G] \fwl} \cdot r_{[G] \fwl}^\T \endp \beginp \hds_{[1] \fwl}\\ \vdots\\\hds_{[G] \fwl}\endp \nonumber\\
&=& \sumg \hdgfwlt \left(\rgfwl \cdot \rgfwlt\right)\hdgfwl.
\endy

In addition, let $\pz = \pzf$ denote the projection matrix of $\zfwl$, with
\beginy\label{eq:fwl_projection}
\hdfwl = \pz \dfwl, \quad \dpzdi = \hdfwlsqi. 
\endy
By \defcrse, 
\begina
\hsefwl^2 
&\overset{\text{\defcrse}}{=}& \dpzdi\left(\dst_\fwl \pz \ho \pz \dfwl \right)\dpzdi\\
&\oeq{\eqref{eq:fwl_projection}}& \hdfwlsqi (\hdfwlt   \ho \hdfwl )\hdfwlsqi\\
&\overset{\eqref{eq:fwl_se_dod}}{=}& \hdfwlsqi\left(\sumg \hdgfwlt \rgfwl \cdot \rgfwlt \hdgfwl\right) \hdfwlsqi\nonumber\\
&=& \sumg \left\{\hdfwlsi \cdot \hdgfwlt \rgfwl\right\}^2\nonumber\\
&\overset{\eqref{eq:fwl_se_ddd}}{=}& \dfrac{1}{N^2}\cdot\dfrac{1}{\szdw^2} \sumg \left(\zst_{\sg\fwl} \cdot \rgfwl\right)^2 
= \dfrac{1}{N^2}\cdot\dfrac{1}{\szdw^2} \sumg \left(\sumig \zisw \cdot \rifwl\right)^2.
\enda

\paragraph*{Proof of \lemfwl\eqref{it:lem_fwl_equivalent S}.}
Let $\hgwd$ denote the \coeffv\ of $\wi$ from $\olst(\dis\sim\wi)$,
with $
\disw = \dis - \wit\hgwd$. 
The first-order condition of $\olst(\zis \sim \wi)$ further ensures  $\sumi \wi\zisw = 0. $ 
These results together ensure
\begina
\szdw 
&=& \meani\zisw \disw =
\meani \zisw\dis - \left(\meani \zisw\wi^\T\right) \hgwd\\
&=&
\meani \zisw\dis \qquad\text{--- the first equivalent form}\\
&\oeqt{Def. of $\hgwz$}& \meani (\zis - \hgwz^\T\wi) \dis\\
&=& \meani\zis \dis - \hgwz^\T\meani  \wi  \dis. \quad\text{--- the second equivalent form}
\enda
The proof for $\szyw$ is identical and therefore omitted. 
\end{proof}

\subsection{Lemmas under the IV setting}

\begin{lemma}\label{lem:Z}
Let $\phigb$. 
Under \assmhomo, we have 
\beginy
&&\phig = \vzi \var(\bzg) \in [0,1],\label{eq:phig_proof}\\
&&
\E(\szt) = \vz  - \var(\bz),
\ \
\eszg  =   \vz \cdot  \ng  (1 - \phig),\nnb\\
&&\E(\sz) = \ni\sumg\eszg = \vz \left(1 - \nisumg \ng\phig\right),\nnb\\
&&\kn = \vzi \cdot \nisumg \eszg = \knf, \qquad 
 \cn =  \cnf.\qquad \qquad\label{eq:kn_cn} 
\endy
\end{lemma}

\begin{proof}[Proof of \lem~\ref{lem:Z}]
We verify below the expressions of $\phig$, $\E(\szt)$, $\eszg$, and $\cn$.
The expressions of $\esz$ and $\kn$ then follow from $\sz = \nisumg \szg$ by definition in \eqref{eq:szg} and $\kn = \E(\sz)/\vz$ by definition in \eqref{eq:kn}.  

\paragraph*{Proof of \eqref{eq:phig_proof}.}
It follows from $\corr(\zi,\zip) \leq 1$ that $\phig\leq 1$. In addition, the correlation matrix is positive semidefinite. This implies $\sum_{i,i'\in\mig} \corr(\zi,\zip) \geq 0$, so that $\phig \geq 0$. 
Under \assmtth, $\corr(\zi,\zip) = \vzi\cov(\zi,\zip)$ so that  
$ \phig = \vzi\cdot \ngsqi\sumiipg \cov(\zi, \zip) 
=\vzi\cdot\ngsqi\var\left(\sumig\zi\right) =  \vzi \var(\bzg)$. 

\paragraph*{Proof of the expression of $\E(\szt)$.}
Direct algebra ensures
\beginy\label{eq:szt_decomp}
\szt = \meani (\zi-e)^2 -  (\bz - e)^2,  
\endy
where the expectations of the first and second terms are
$\E\left\{\meani (\zi-e)^2 \right\} = \vz$ and $\E\{ (\bz - e)^2\} = \var(\bz) = \vz  N^{-1} \cn$.
This implies the result.

\paragraph*{Proof of the expression of $\eszg$.}
Applying \eqref{eq:szt_decomp} to cluster $g$ ensures $
  \szg  \oeq{\eqref{eq:szt_decomp}}  \sumig (\zi-e)^2 - \ng(\bzg - e)^2$,  so that 
$\eszg   =  \ng \vz  - \ng  \var(\bzg)
\oeq{\eqref{eq:phig_proof}}  \ng  \vz  - \ng \cdot \vz \phig 
= \vz  \ng (1-\phig). 
$

\paragraph*{Proof of the expression of $\cn$.}
Given $\phig =\vzi \var(\bzg)$ from \eqref{eq:phig_proof}, we have   
\begina
\cn &=& \vzi N \var(\bz) =\vzi N^{-1} \var\left(\sumi \zi \right)\\ 
&\oeqt{Assm.~\ref{assm:cs}}& \vzi N^{-1} \sumg \var\left(\sumig \zi\right) = \vzi N^{-1} \sumg \ngsq \var(\bzg)
\oeq{\eqref{eq:phig_proof}}   \nisumg \ngsq\phig. 
\enda
\end{proof}

\begin{lemma}\label{lem:cef}
Assume {\assmtth}.
Let $
\tyi = \yiz + \oua \ti$  denote the potential outcome of unit $i$ if $\zi = 0$. 
Let $\tmu = \E(\tyi)$ denote the common value of $\E(\tyi)$ under \assmim. 
Let $\sxz = \meani \dxi\dzi$ and $\sza = \meani \dzi\dai$. 
Then 
\begine[(i)]
\item\label{it:lem_cef_EYZ} $\E(\yi \mid \zg) = \E(\yi \mid \zi) = \tmu  + \zi \pc\tc$,

$
\muy = \tmu + e\pctc,\quad
\muzy = e(\tmu + \pctc), \quad
\cov(\zi, \yi) =  \vzpc\tc.
$

\item\label{it:lem_cef_EDZ} $\E(\di \mid \zg)= \E(\di \mid \zi)= \pa + \zi\pc$,

$\mud = \edf,
\quad 
\muzd = e(\pa + \pc),
\ \
\covzd  = \vzpc$, 

$\E(\ddi \mid \zg) = \dzi \pc$,
\ \ 
$\E(\dzi \ddi) = \pc \E(\dot Z_i^2), 
\quad 
\E(\szd) = \pc \esz = \kvzpc. 
$
 
\item\label{it:lem_cef_EAZ} $\E(\ai \mid \zg)= \E(\ai \mid \zi)= \tmu - \pa\tc = \mua$, \quad $\E(\aag \mid \zg) = \E(\aag) =\ong\mua$,

$\cov(\zg, \aag) = 0$, \quad $\covza = \E(\dzi\dai) = \E(\sza) = 0$.
 
\item\label{it:lem_cef_EXZ} $\cov(\xxi,\zi) = \E(\dxi\dzi) = \E(\sxz) = 0_p$.
\ende
\end{lemma}

\begin{proof}[Proof of \lemcef]
We verify below \lem~\ref{lem:cef} \eqref{it:lem_cef_EYZ}--\eqref{it:lem_cef_EXZ}, respectively. 

\paragraph*{Proof of \lemcef\eqref{it:lem_cef_EYZ}.} Write $\di = \oua + \zi \cdot \ouc$. 
\assmiv\eqref{it:assm_iv_er} ensures 
\begineq\label{eq:yi_tyi}
Y_i = \yiz + \di \ti 
= \yiz + \Big(\oua + Z_i \cdot \ouc\Big) \ti
= \tyi + Z_i \cdot \ouc\ti.
\endeq
 This, together with \assmiv\eqref{it:assm_iv_random}, implies  \lemcef\eqref{it:lem_cef_EYZ} as follows:
\begina
\E(\yi \mid \zg) 
&\oeq{\eqref{eq:yi_tyi}}&  
\E (\tyi \mid \zg ) + \zi \cdot \E (  \ouc \tau_i  \mid \zg )\\
&\oeqt{Assm.~\ref{assm:iv}\eqref{it:assm_iv_random}}& 
\E(\tyi) + \zi \cdot  \E ( \ouc \tau_i ) =
\tmu + \zi \cdot \pr(\uc) \cdot \E\left( \ti  \mid \uc  \right). 
\enda

\paragraph*{Proof of \lemcef\eqref{it:lem_cef_EDZ}.}
%
The first two lines in \lemcef\eqref{it:lem_cef_EDZ} follows from defining the potential outcomes as $\yi'(d) = d$ for $d = 0,1$ with $\yi' = D_i$ and $\tc' = 1$. 
It then follows from $\E(\di \mid \zg) = \pa + \zi\pc$ that 
\begineq\label{eq:eddi_1}
\E(\bd \mid \zg) = \meanig \E(\di\mid \zg) = \pa + \bz\pc, \quad \E(\ddi \mid \zg) = \dzi \pc. 
\endeq 
This implies 
$\E(\dzi\ddi) = \E\{\E(\dzi\ddi\mid \dzg)\} 
=
\E\{\dzi \E(\ddi\mid \dzg)\} \oeq{\eqref{eq:eddi_1}} \E(\dzi^2 \pc) = \pc \E(\dzi^2).$
\paragraph*{Proof of \lemcef\eqref{it:lem_cef_EAZ}.}
\lemcef\eqref{it:lem_cef_EYZ}--\eqref{it:lem_cef_EDZ} ensure  
$\E(\ai\mid \zg) = \E(\yi \mid \zg) -  \tc\cdot\E(\di\mid \zg) = \tmu -  \pa\tc$, 
which is independent of $\zg$. This implies that 
$\E(\aag \mid \zg) = \mua\ong$, where $\mua = \E(\ai) = \tmu -   \pa\tc$. 
Accordingly, by the law of iterated expectations, 
\begina
\cov(\zg, \aag) &=& \E\Big\{ \cov(\zg, \aag \mid \zg) \Big\} +  \cov\Big\{\E(\zg \mid \zg), \E(\aag \mid \zg) \Big\}\\
&=& 0 + \cov(\zg, \mua\ong) = 0,\\
\E(\dzi \dai) &=& \E\left\{\E(\dzi\dai \mid \zg)\right\}  
= \E\left\{\dzi \cdot \E(\dai \mid \zg)\right\} 
= \E\left\{\dzi \cdot \E(\dai)\right\} = 0. 
\enda

\paragraph*{Proof of \lemcef\eqref{it:lem_cef_EXZ}.}
The renewed \assmiv\eqref{it:assm_iv_random} ensures $\zg \indep (\xxi:\ig)$.
This ensures $\cov(\xxi,\zi) = 0$,  $\E(\dxi\dzi) = \E(\dxi)\cdot \E(\dzi) = 0$, and $\E(\sxz) = \sumi \E(\dxzi) = 0$. 
\end{proof}

\subsubsection{Asymptotics}
\begin{lemma}\label{lem:asym_nox}
Assume \assmhomo. Let $
\bi = \beginp 
\cai\\
\zi(\cai)\\
\dzi \dai 
\endp$ and $\szdg = \dsumig  \dzi \ddi$. 
Then 
\begine[(i)]
\item\label{it:lem_uniform_ZD} $Z_i$, $D_i$, $Z_i\di$, $\dzi\ddi$ are all \uni\ over $\otn$. 
%
\item\label{it:lem_asym_nox_S}
If \assm~\ref{assm:ng_main} holds, then  
\begineqs
\beginar{lll}
\bz = e + \op, \qquad \bd = \mud + \op, && \ds\olzd \equiv \meani \zi\di = \muzd + \op,\\
\szdt = \E(\szdt) + \op = \vzpc + \op, &&\szd = \eszd + \op = \kn\cdot \vzpc +\op,\medskip\\
\ds\nisumg \szdg^2 = \ooo. 
\endar
\endeqs

\item\label{it:lem_asym_nox_AB}
If $\E(Y_i^4) < \infty$ holds, then  $\E(\ai^4) = O(1)$,  $\supi\E(\dai^4) \leq 2^4 \cdot \E(\ai^4) = O(1)$; and 
$\{\|\bi\|^2: \otn\}$ is \uni\ with $\supi\E(\en{\bi}^4) = O(1).$
\ende
\end{lemma}

\begin{proof}[Proof of \lemasymnox]
We verify below \lemasymnox\eqref{it:lem_uniform_ZD}--\eqref{it:lem_asym_nox_AB} one by one.

\paragraph*{Proof of \lemasymnox\eqref{it:lem_uniform_ZD}.}
$|\zi|$, $|\di|$, $|\zidi|$, and $|\dzi\ddi|$ are all bounded between $[0,1]$, and therefore \uni\ over $\otn$. 

\paragraph*{Proof of \lemasymnox\eqref{it:lem_asym_nox_S}.}
\lem~\ref{lem:hetero} implies that $
\bz = \E(\bz) + \op = e + \op$,  $\bd = \E(\bd) + \op = \mud + \op$, $
\olzd = \E(\olzd) + \op = \muzd + \op,$ and $\szd = \E(\szd) + \op = \kvzpc + \op.$
Accordingly, $
\szdt = \meani \zidi - \bz \bd = \muzd - e\mud + \op =\cov(\zi,\di) + \op$,
where $\cov(\zi,\di) = \vzpc$ by \lemcefd. 
Lastly, it follows from $|\dzi\ddi| \leq 1$ that $|\szdg| \leq \ng$, so that
$\ds\nisumg \szdg^2 \leq \nisumg \ng^2$, 
which is $\ooo$ under \assm~\ref{assm:ng_main}.

\paragraph*{Proof of \lemasymnox\eqref{it:lem_asym_nox_AB}.} 
By \lembound, 
\begineq\label{eq:bound_dai}
\beginar{rcccl}
 \ai^4  =  (\yi - \di\tc)^4 
 &\oleqt{\lembound\eqref{it:lem_bound_1}}& 2^{4-1}\left( \yi^4 + \di^4\tc^4\right) 
 &\leq& 2^{4-1}\left( \yi^4 +  \tc^4\right), 
 \medskip\\
\E(\dai^4) 
 &\oleqt{\lembound\eqref{it:lem_bound_2}}& 2^4 \E(\ai^4) 
 &\leq& 2^4 \cdot 2^{4-1}\big\{\E(\yi^4) + \tc^4\big\}. 
\endar 
 \endeq
 This implies $\E(\ai^4) = O(1)$ and $\supi\E(\dai^4) = O(1)$ provided $\E(\yi^4) = O(1)$. 
In addition, given $\zi, \dzi^2 \in [0, 1]$, 
\begina
\|\bi\|^2 &=& (A_i-\mua)^2 + \zi(A_i-\mua)^2 + (\dzi\dai)^2 
 \leq     2\{(\ai - \mua)^2 +  \dai^2\} \\
 &\oleqt{\csp}& 2(2 \ai^2 + 2\mua^2  +  \dai^2),\nnb\\
\|\bi\|^4 &\leq&  2^2(2 \ai^2 + 2\mua^2 +  \dai^2)^2 \oleqt{\csp} 2^2\cdot 3 ( 4\ai^4 + 4\mua^4 + \dai^4).  
\enda
Combining this with \eqref{eq:bound_dai} implies $
\supi \E(\en{\bi}^4)\leq 12 \{ 4\E(\ai^4) + 4\mua^4 + \supi \E(\dai^4)\} =O(1)$, so that $\en{\bi}^2$ is \uni\ by the sufficient condition under \eqref{eq:uni_def}. 
\end{proof}

\begin{lemma}\label{lem:fex} 
Recall from \eqref{eq:rrifex_def} that $\rrifex = \ai - \xit\gxa$, where $\gxa = \gxaf$. 
Let $\bi = \dzdx \rrifex$, $\sxzg = \dsumig \dxi \dzi$,   $\sxdg = \dsumig  \dxi \ddi$, 
$\sxg = \dsumig \dxi \dxit$,
$\szrg = \dsumig  \dzi \rrifex$,  and $\sxrg = \dsumig  \dxi \rrifex$.
Under \assmasymfex,
\begine[(i)]
\item\label{it:lem_fex_uni_X}
Let $X_{ik}$ denote the $k$th element of $\xxi$ for $k = \ot{p}$.
Under \assmasym\eqref{it:assm_asym_fex}, $\|\dxi\|$, $\|\dxi\dzi\|$, $\|\dxi\ddi\|$, $\|\dx_{ik}\dx_{ik'} \|  \ (k, k' = \ot{p})$, $\|\dxi\dai\|$, and $\|\dxi\dai\|^2$ are all \uni\ over $\otn$, with
\begineqs 
\beginar{l}
\supi\E(\|\dxi\ddi\|) = O(1),\qquad
\supi\E(|\dx_{ik}\dx_{ik'}|) = O(1) \quad \forall   k, k' = \ot{p},
\medskip\\
\supi\E(\|\dxi\dai\|) = O(1), \quad \supi\E(\|\dxi\dai\|^{2+2q}) =O(1) \for\ q = \dfrac{s-2}{s+2} \in(0,1)
\endar
\endeqs
as $N \to \infty$, where $s \in (2,\infty)$ is the constant in \assmasym\eqref{it:assm_asym_fex} such that $\E(\|\xxi\|^{2s}) < \infty$.  
\item\label{it:lem_fex_S_plim}
$\sxz  = \E(\sxz) + \op = \op$, \quad $\sxd = \esxd + \op = \oop$,

$\sxa = \E(\sxa) + \op = \oop$, \quad $\sx = \E(\sx) + \op = \oop$,

where $\E(\sxz) = 0$, and $\esxd$, $\E(\sxa)$, and $\esx$ are  componentwise $O(1)$.
 
\item\label{it:lem_fex_uni_B} 
$\gxa$ is componentwise $O(1)$; \quad $\supi \E(\en{\bi}^{2+2q}) = O(1)$.   
\item\label{it:lem_fex_X_Y_R}
$\ds\nisumg \spn{\dxg}^{4} = \oop$,  
$\ds\nisumg \|\dyg\|^{4} = \oop$, 
$\ds\nisumg \en{\rrgfex}^{4} = \oop$. 
\item\label{it:lem_fex_bound_Sg^2} 
$\dnisumg \sxzg \sxzgt = \oop$, \quad 
$\dnisumg \sxdg \sxdgt = \oop$,

$\dnisumg \|\sxg\|_2^2 \leq \dnisumg \fn{\sxg}^2 = \oop$, \quad  
$\dnisumg \sxrg \sxrg^\T = \oop$. 
\ende 

\end{lemma}

\begin{proof}[Proof of \lemfex]
We verify below \lemfex\eqref{it:lem_fex_uni_X}--\eqref{it:lem_fex_bound_Sg^2}, respectively.

\paragraph*{Proof of \lemfex\eqref{it:lem_fex_uni_X}.}
\assmasym\eqref{it:assm_asym_fex} ensures $\E(\|\xxi\|^{2s}) = O(1)$, so that    
\begineq\label{eq:lem_fex_XX_2}
\sup_{\otn}  \E(\|\dxi\|^{2s}) 
\leq
 2^{2s} \cdot \E(\|\xxi\|^{2s}) = O(1) 
\endeq
by \lembound\eqref{it:lem_bound_2}. 
This implies \unin\ of $\|\dxi\|$. Combining \eqref{eq:lem_fex_XX_2} with 
\begineq\label{eq:lem_fex_XX_1}
\|\dxi\dzi\|, \|\dxi\ddi\| \leq \|\dxi\|, 
\quad 
|\dot X_{ik} \dot X_{ik'} | \leq 2^{-1} \left( \dot X_{ik}^2 + \dot X_{ik'}^2\right) \leq 2^{-1} \|\dxi\|^2,
\endeq 
implies \unin\ of $\|\dxi\dzi\|$, $\|\dxi\ddi\|$, and $|\dot X_{ik} \dot X_{ik'} |$. 
Also, it follows from 
\begineqs
\{\E(\|\dxi\|)\}^2 \leq \E(\|\dxi\|^2) \leq  \{\E(\|\dxi\|^{2s}) \}^{1/s}
\endeqs
by \lembound\eqref{it:lem_bound_3} that 
\begineqs 
\beginar{l}
\ds\supi \E(\|\dxi\ddi\|) \leq \supi \E(\|\dxi\|)
\oleqt{\text{\lembound\eqref{it:lem_bound_3}}} 
\supi \left\{\E(\|\dxi\|^{2s})\right\}^{\frac{1}{2s}} \oeq{\eqref{eq:lem_fex_XX_2}}
O(1), 
\medskip\\ 
\ds\supi \E(|\dot X_{ik} \dot X_{ik'} |)
\oleq{\eqref{eq:lem_fex_XX_1}} 
\supi \E(\|\dxi\|^2)
\oleqt{\text{\lembound\eqref{it:lem_bound_3}}} 
\supi \left\{\E(\|\dxi\|^{2s})\right\}^{1/s} \oeq{\eqref{eq:lem_fex_XX_2}} O(1). 
\endar 
\endeqs
In addition, let $g = \dfrac{s+2}{2} \in (1, \infty)$ and $h  = \dfrac{s+2}{s} \in (1, \infty)$ with  
\begineq\label{eq:holder_gh}
2 + 2q = \dfrac{4s}{s+2}, \quad 1/g + 1/h = 1, \quad (2+2q)g = 2s, \quad (2+2q)h = 4. 
\endeq
Given $
\|\dxi\dai\|  = \sqrt{\dai\dxit \dxi\dai} =  \|\dxi\| \cdot |\dai|$, 
Holder's inequality implies 
\begina
\E(\|\dxi\dai\|^{2+2q})
&\leq&  
\left[\E\left\{(\|\dxi\|^{2+2q})^g\right\}\right]^{1/g} \cdot \left[\E\left\{(|\dai|^{2+2q})^h\right\} \right]^{1/h}\nnb\\
&\oeq{\eqref{eq:holder_gh}}& 
\left\{\E(\|\dxi\|^{2s})\right\}^{1/g} \cdot \left\{\E( \dai ^4) \right\}^{1/h}. 
\enda
This, together with $\sup_{\otn}  \E(\|\dxi\|^{2s}) = O(1)$ from \eqref{eq:lem_fex_XX_2} and $\sup_{\otn}  \E(\|\dai\|^4) = O(1)$ from 
\lemasymnox\eqref{it:lem_asym_nox_AB}, implies 
\begineq\label{eq:bound_dxda_2q}
\supi \E(\|\dxi\dai\|^{2+2q}) \oeq{\eqref{eq:lem_fex_XX_2}+\text{\lemasymnox\eqref{it:lem_asym_nox_AB}}}
O(1)    
\endeq 
under \assmasym\eqref{it:assm_asym_Y} and \eqref{it:assm_asym_fex}, 
so that both $\dxi\dai$ and $\|\dxi\dai\|^2$ are \uni\ over $\otn$.
\lembound\eqref{it:lem_bound_3} further ensures $\E(\|\dxi\dai\|) \leq\{\E(\|\dxi\dai\|^{2+2q})\}^{1/(2+2q)}$, so that  
$\supi \E(\|\dxi\dai\|) \oleqt{\lembound\eqref{it:lem_bound_3}} \supi \{\E(\|\dxi\dai\|^{2+2q}\}^{1/(2+2q)} \oeq{\eqref{eq:bound_dxda_2q}}
O(1) $
by \eqref{eq:bound_dxda_2q}. 

\paragraph*{Proof of \lemfex\eqref{it:lem_fex_S_plim}.}
\lemfex\eqref{it:lem_fex_uni_X} ensures that, under \assmasym\eqref{it:assm_asym_fex}, 
$\dxi\dzi, \dxi\ddi, \dxi\dai$, and all elements in $\dxi\dxit$ are all \uni\ over $\otn$, and therefore satisfy the conditions required by \lemsm\eqref{it:lem_sm_wlln}.
Applying \lemsm\eqref{it:lem_sm_wlln} to $\dxi\dzi, \dxi\ddi, \dxi\dai$, and all elements in $\dxi\dxit$ implies  
\begineq\label{eq:fex_1}
\begin{array}{rclcl}
 \sxz &=& \meani \dxi\dzi &=& \E(\sxz) + \op,\\
 \sxd &=& \meani \dxi\ddi &=& \E(\sxd) + \op, \\
 \sxa &=& \meani \dxi\dai &=& \E(\sxa) + \op, \\
\sx &=& \meani \dxi\dxit &=& \E(\sx) + \op. 
\end{array}
\endeq
\lemcef\eqref{it:lem_cef_EXZ} implies $\E(\sxz) = 0$, so the first row in \eqref{eq:fex_1} implies $\sxz = \op$. 

Let $\sx[k,k']$ denote the $(k,k')$th element of $\sx$. 
The triangle inequality of norms implies 
\begineq\label{eq:bound_S_triangle}
\beginar{rclcl}
\|\sxd\|  
&=& \ds \left\|\meani  \dxi\ddi\right\|  
&\leq& \ds\meani  \| \dxi\ddi \|,\\
   \|\sxa\|  
 &=&  \ds\left\|\meani \dxi\dai\right\| 
 &\leq& \ds\meani  \|\dxi\dai\|  ,\\
  |\sx[k,k']|   
 &=& \ds\left|\meani \dx_{ik}\dx_{ik'}\right| 
 &\leq&
  \ds\meani |\dx_{ik} \dx_{ik'}|.
\endar
\endeq
Combining this with 
\begineqs
\supi \E(\| \dxi\ddi \|) = \ooo, \quad \supi \E(\|\dxi\dai\|) = \ooo, \quad \supi \E(|\dx_{ik} \dx_{ik'}|) = \ooo
\endeqs as $\ntinf$ by \lemfex\eqref{it:lem_fex_uni_X} implies  
\begineq\label{eq:fex_2}
\begin{array}{rcl}
\E(\|\sxd\|) 
&\oleq{\eqref{eq:bound_S_triangle}}&\ds\meani  \E(\| \dxi\ddi \|)  
\oeqt{\lemfex\eqref{it:lem_fex_uni_X}} O(1)
  ,\\
 \E(\|\sxa\|) 
&\oleq{\eqref{eq:bound_S_triangle}}& \ds\meani \E(\|\dxi\dai\|)
\oeqt{\lemfex\eqref{it:lem_fex_uni_X}} O(1)
  ,\\
 \E(|\sx[k,k']|) 
&\oleq{\eqref{eq:bound_S_triangle}}& \ds\meani \E\left(\left|\dx_{ik}\dx_{ik'}\right|\right) \oeqt{\lemfex\eqref{it:lem_fex_uni_X}} O(1)  
 \end{array}
\endeq
as $\ntinf$. 
Equation~\eqref{eq:fex_2} implies $
\E(\sxd)$, $\E(\sxa)$, and $\E(\sx)$ are all componentwise $O(1)$, so that $\sxd$, $\sxy$, and $\sx$ are all componentwise $\oop$ from \eqref{eq:fex_1}.

\paragraph*{Proof of \lemfex\eqref{it:lem_fex_uni_B}.}
\lemfex\eqref{it:lem_fex_S_plim} ensures $\E(\sxa) = O(1)$. 
This, together with $\lm\{\esx\} \geq \lambda > 0$ under \assmasym\eqref{it:assm_asym_fex}, ensures 
\begineq\label{eq:bound_gxa}
\gxa \oeq{\eqref{eq:rrifex_def}} \{\E(\sx)\}^{-1}\E(\sxa) = O(1).  
\endeq
In addition, given
$\bi  =  \dzdx (\rrifexf) = \beginp \dzi(\dai - \dxit\gxa) \\ \dxi\dai - \dxi\dxit\gxa\endp$, 
it follows from 
\begina
&& |\dxit\gxa| \leq \|\dxi\|\cdot \|\gxa\|,\\
&& \|\dxi\dxit\gxa\| = \|\dxi\| \cdot |\dxit\gxa| \leq \|\dxi\|^2\cdot \|\gxa\|,\\
&& \|\dxi\dai \|  = \|\dxi\| \cdot|\dai| 
\enda
that 
\begineqs
\beginar{rcccl}
\dzi^2(\dai - \dxit\gxa)^2 &\leq& 2(|\dai|^2 + |\dxit\gxa|^2) &\leq& 2(|\dai|^2 + \|\dxi\|^2\cdot \|\gxa\|^2),\\
\|\dxi\dai - \dxi\dxit\gxa \|^2 &\leq& 2\left( \|\dxi\dai \|^2 +   \|\dxi\dxit\gxa \|^2\right) &\leq& 2 \left( \|\dxi\|^2\cdot|\dai|^2 +   \|\dxi\|^4 \cdot \|\gxa\|^2\right),
\endar
\endeqs
with 
\beginy\label{eq:uni_bi_1}
\en{\bi}^2 &=&  \dzi^2(\dai - \dxit\gxa)^2 +  \|\dxi\dai - \dxi\dxit\gxa \|^2\nnb\\
&\leq& 2 \left(|\dai|^2 +  \|\dxi\|^2\cdot \|\gxa\|^2 +  \|\dxi\|^2\cdot|\dai|^2 +   \|\dxi\|^4 \cdot \|\gxa\|^2\right).
\endy
%
Recall from \lemfex\eqref{it:lem_fex_uni_X} that $q = \dfrac{s-2}{s+2} \in (0,1)$ with $1+ q = \dfrac{2s}{s+2}$. 
\lembound\eqref{it:lem_bound_1} ensures 
\begina
\dfrac{\en{\bi}^{2+2q}}{2^{1+q}}
&=& \left(\dfrac{\|\bi\|^2}{2}\right)^{1+q}\\
&\oleq{\eqref{eq:uni_bi_1}}&  \left(|\dai|^2 +  \|\dxi\|^2\cdot \|\gxa\|^2 +  \|\dxi\|^2\cdot|\dai|^2 +   \|\dxi\|^4 \cdot \|\gxa\|^2\right)^{1+q}\\
&\oleqt{\lembound}&  4^q \left(|\dai|^{2+2q} +  \|\dxi\|^{2+2q}\cdot\|\gxa\|^{2+2q} + \|\dxi\|^{2+2q} \cdot |\dai |^{2+2q} + \|\dxi\|^{4+4q}\cdot\|\gxa \|^{2+2q}\right), 
\enda
so that, after dividing both sides by $4^q$ and taking supreme over expectations, 
\begineq\label{eq:bound_B2_0}
\beginar{rcl}
\dfrac{\supi \E(\en{\bi}^{2+2q})}{2^{1+q}\cdot 4^q}
&\leq& \supi \E(|\dai|^{2+2q}) +  \|\gxa\|^{2+2q} \cdot\supi \E(\|\dxi\|^{2+2q})\\
&& + \supi \E\left(\|\dxi\|^{2+2q} \cdot |\dai |^{2+2q}\right) + \|\gxa \|^{2+2q}\cdot  \supi \E(\|\dxi\|^{4+4q}).  
\endar
\endeq
Note that $q < 1$ and $s > 2$,  so that $
2+2q < 4 < 2s$, and $4+4q = 4\cdot \dfrac{2s}{s+2} < 2s$.
By Holder's inequality and \lembound\eqref{it:lem_bound_3}, 
\begineq\label{eq:bound_B2_1}
\beginar{rcl}
\ds\supi \E(|\dai|^{2+2q}) &\leq& \ds\left\{\supi \E(|\dai|^4)\right\}^{(2+2q)/4} \oeqt{\lemasymnox\eqref{it:lem_asym_nox_AB}} O(1), \\
\ds\supi \E(\|\dxi\|^{2+2q}) &\leq& \ds\left\{\supi \E(|\dxi|^{2s})\right\}^{(2+2q)/2s} \oeq{\eqref{eq:lem_fex_XX_2}} O(1),\\
\ds\supi \E(|\dxi|^{4+4q}) &\leq& \ds\left\{\supi \E(|\dxi|^{2s})\right\}^{(4+4q)/2s} \oeq{\eqref{eq:lem_fex_XX_2}} O(1). 
\endar
\endeq
In addition, \eqref{eq:bound_gxa} and \lemfex\eqref{it:lem_fex_uni_X} imply
\begineq\label{eq:bound_gxa_2}
\en{\gxa}^{2+2q} = O(1), \quad \supi \E\left(\|\dxi\|^{2+2q} \cdot |\dai |^{2+2q}\right) = O(1). 
\endeq
Plugging \eqref{eq:bound_B2_1} and \eqref{eq:bound_gxa_2} into \eqref{eq:bound_B2_0} implies $
\supi \E(\en{\bi}^{2+2q}) = O(1)$ so that $\|\bi\|^2$ is \uni\ over $\otn$.

\paragraph*{Proof of \lemfex\eqref{it:lem_fex_X_Y_R}.}
We verify below the results for $\dxg$, $\dyg$, and $\rrgfex$, respectively. 

\underline{\it Proof of $\nisumg \spn{\dxg}^{4} = \oop$}.
Properties of norms ensure
\begineq\label{eq:oop_X_1}
 \spn{\dxg} \leq \fn{\dxg}. 
\endeq 
Also, it follows from 
$\fn{\dxg}^2 = \fn{(\dx_{g1}, \dx_{g2}, \ldots, \dx_{g,\ng})^\T}^2 = \sumig \en{\dxi}^2$ that 
\begineq\label{eq:oop_X_2}
\fn{\dxg}^4 = \left(\sumig \en{\dxi}^2\right)^2\oleqt{\csp} \ng\left(\sumig\en{\dxi}^{4} \right)
\endeq
by \csf. 
\lembound\eqref{it:lem_bound_2} further ensures that, under \assmim, 
\begineq\label{eq:oop_X_3}
\E\left(\|\dxi\|^4\right) \leq 2^4 \cdot \E\left(\|\xxi\|^4\right).  
\endeq
Combining \eqref{eq:oop_X_1}--\eqref{eq:oop_X_3} implies
\beginy\label{eq:oop_X_4}
 \spn{\dxg}^{4} 
 &\oleq{\eqref{eq:oop_X_1}}& 
  \fn{\dxg}^{4}  \oleq{\eqref{eq:oop_X_2}}
   \ng\left(\sumig\en{\dxi}^{4} \right), \nnb\\
 \E\left( \spn{\dxg}^{4} \right)
 &\leq& 
 \ng\left\{\sumig \E\left(\en{\dxi}^{4}\right) \right\} 
\oleq{\eqref{eq:oop_X_3}} 
16 \cdot \ng^2 \cdot \E(\en{\xxi}^{4}),
\endy
so that 
\beginy\label{eq:oop_X_5}
\E\left(\nisumg \spn{\dxg}^{4} \right) = \nisumg\E\left( \spn{\dxg}^{4} \right) \oleq{\eqref{eq:oop_X_4}} 
16\cdot \E(\en{\xxi}^{4}) \cdot \ni\sumg \ng^2.
\endy
\assmasym\eqref{it:assm_asym_123} and \eqref{it:assm_asym_fex} ensure $
\E(\en{\xxi}^{4}) \leq \left\{\E(\en{\xxi}^{2s})\right\}^{\frac{4}{2s}} < \infty$ and 
$\limsup_{\ntinf} \ni\sumg \ng^2 < \infty$
in  \eqref{eq:oop_X_5} so that $\limsup_{\ntinf} \E\left(\nisumg \spn{\dxg}^{4}\right) \overset{\eqref{eq:oop_X_5}}{<} \infty$. 
It then follows from \lem~\ref{lem:basic}\eqref{it:lem_basic_oop} that $
 \nisumg \spn{\dxg}^{4}\oeqt{\lem~\ref{lem:basic}}\oop.$

\underline{\it Proof of $\nisumg \|\dyg\|^{4} = \oop$}. Replacing $\dxg$ by $\dyg$ in the above proof with $
\spn{\dyg}^2 =  \lmax\left(\dyg^\T\dyg\right) =\dyg^\T\dyg = \en{\dyg}^2$ and $\spn{\dyg}^4 = \en{\dyg}^4$
ensures 
$\nisumg \en{\dyg}^{4} = \oop$
under \assmasym\eqref{it:assm_asym_123} and \eqref{it:assm_asym_Y}. 

\underline{\it Proof of $\nisumg \en{\rrgfex}^{4} = \oop$}.
Given $|\ddi|\leq 1$, we have $
\en{\ddg}^2 = \sumig\ddi^2 \leq \ng$ and $\en{\ddg}^{4}  \leq  \ng^2$, 
so that $
\nisumg \en{\ddg}^{4} \leq \limsup_{\ntinf} \nisumg n_g^2  = O(1)$
under \assmasym\eqref{it:assm_asym_123}. 
In addition, \lemfex\eqref{it:lem_fex_uni_B} ensures 
\begineq\label{eq:bound_gxa_3}
\gxa  = O(1), \quad \en{\gxa}^{4} = O(1).  
\endeq 
\lemen\ further implies 
\begineq\label{eq:bound_S2_gxa2}
\en{\dxg\gxa}^{4} \leq \spn{\dxg}^{4}\cdot \en{\gxa}^{4}.
\endeq 
It then follows from \lembound\eqref{it:lem_bound_1} 
that
\begina\label{eq:bound_S2_R}
\en{\rrgfex}^{4} 
&=& \|\dyg - \ddg\tc - \dxg\gxa\|^4\\
&\oleqt{\lembound\eqref{it:lem_bound_1}}&  3^{4-1}\cdot \left( \en{\dyg}^{4} + \en{\ddg}^{4}\cdot \tc^{4} + \en{\dxg\gxa}^{4}\right) \nnb\\
&\oleqt{\eqref{eq:bound_S2_gxa2}}&  3^{4-1}\cdot \left( \en{\dyg}^{4} + \en{\ddg}^{4}\cdot\tc^{4} + \spn{\dxg}^{4}\cdot \en{\gxa}^{4}\right),
\enda
so that
\begina
\nisumg \en{\rrgfex}^{4} 
\leq  3^{4-1}\left( \nisumg \en{\dyg}^{4} + \tc^{4}\cdot\nisumg \en{\ddg}^{4} + \en{\gxa}^{4}\cdot\nisumg\spn{\dxg}^{4}\right)\oeq{\eqref{eq:bound_gxa_3}} \oop  
\enda
from the results of $\en{\dyg}^{4}$ and $\spn{\dxg}^{4}$ and \eqref{eq:bound_gxa_3}.

\paragraph*{Proof of \lemfex\eqref{it:lem_fex_bound_Sg^2}.} We verify below that $\nisumg \sxrg \sxrg^\T = \oop$. The proof of the rest is similar hence omitted. 
Let $\sxrg[k]= \sumig \dot{X}_{ik}\rrifex$ denote the $k$th element of $\sxrg = \sumig \dxi\rrifex = \dxgt\rrgfex$.
\lemen\ implies 
\begineq\label{eq:bound_Sxrg}
\en{\sxrg} = \en{\dxg^\T\rrgfex} \leq \spn{\dxg} \cdot \en{\rrgfex} 
\endeq 
so that 
\begina
\big|\sxrg[k] \cdot \sxrg[k']\big| &\leq& 2^{-1}\left\{\big(\sxrg[k]\big)^2+ \big(\sxrg[k']\big)^2 \right\} \\
&\leq& 2^{-1}\|\sxrg\|^2
 \oleq{\eqref{eq:bound_Sxrg}} 2^{-1} \spn{\dxg}^2 \cdot \en{\rrgfex}^2
\leq  4^{-1}\left(\spn{\dxg}^{4} + \en{\rrgfex}^{4}\right).
\enda
Therefore, the $(k,k')$th element of $\dnisumg \sxrg \sxrg^\T$ satisfies
\begina
\left| \nisumg \sxrg[k] \cdot \sxrg[k']\right| 
\leq
4^{-1}\left(\nisumg \spn{\dxg}^{4} + \nisumg \en{\rrgfex}^{4}\right)\oeqt{\lemfex\eqref{it:lem_fex_X_Y_R}}\oop 
\enda
by \lemfex\eqref{it:lem_fex_X_Y_R}.
This implies $\dnisumg \sxrg \sxrg^\T = \oop$.   
\end{proof}

\subsubsection{Lemmas under \assmy}
\begin{lemma}\label{lem:A}
Under \assmy, 
\begineqs
A_i = \yiz + \oua(\tau_i - \tc), \quad
\cov(\aag\mid \zg) = \cov(\aag) = \ve   \ing +  \va   \og \ogt.
\endeqs
\end{lemma}

\begin{proof}[Proof of \lema]
The expression of $\aai$ follows from $
A_i = Y_i - \di\tc = \yiz + \di(\tau_i - \tc) = \yiz + (\oua + \zi\ouc)(\tau_i - \tc)$, which is independent of $\zg$ under \assmiv. 
The expression of $\cov(\aag)$ follows from $\aag = \ong\ag + \epg$ so that
\beginy\label{eq:lem_eff_4} 
\E(\aag \mid \ag) = \og\ag, \quad\cov(\aag\mid \ag) = \cov(\epg \mid \ag) = \ve  \ing,
\endy
with
$\cov(\aag)= \E\{\cov(\aag \mid \ag) \} + \cov\{\E(\aag \mid \ag) \} 
\overset{\eqref{eq:lem_eff_4}}{=}  \ve   \ing +   \cov(\ong \ag). 
$
\end{proof}

\begin{lemma}\label{lem:efficiency} 
Recall that $\phigb$. Under \assmhomo\ and \ref{assm:y}, 
\begina
\var\left\{\sumig(\czi)(\cai)\right\} 
&=& \vz  \left( \ve \cdot \ng + \va \cdot   \ngsq \phig \right),\\
\var  \left\{\sumig(\zi - \bzg)(\ai - \bag) \right\}  
&=& \ve\cdot \eszg = \vz  \ve \cdot \ng  ( 1 - \phig  ).
\enda
\end{lemma}

\begin{proof}[Proof of \lem~\ref{lem:efficiency}]
Let $
\czg = \zg - \og e$ and $\cag = \aag - \og \mua$. 
From \lem~\ref{lem:cef}\eqref{it:lem_cef_EAZ},
\beginy\label{eq:za_1}
\E(\cag \mid \zg) = \E(\cag) =0,
\quad 
\E(\dag\mid \zg) = \E(\dag) = 0, \quad 
 \dzgt\dag \oeq{\eqref{eq:apb}} \dzgt \aag. \qquad
\endy
This implies that 
\beginy\label{eq:za_10}
\var( \czgt \cag)
&=& 
\E\Big\{\var\left( \czgt \cag \mid \zg\right)  \Big\} + \var\Big\{\E\left( \czgt \cag \mid \zg\right)  \Big\} \nnb\\
&\overset{\eqref{eq:za_1}}{=}& 
\E\Big\{ \czgt \cdot \cov(\cag \mid \zg)\cdot \czg \Big\}
=
\E\Big\{ \czgt \cdot \cov(\aag \mid \zg)\cdot \czg \Big\}, \qquad\\
\label{eq:za_v_dzda_2}
\var( \dzgt\dag)&=&
\E\left\{\var(\dzgt\dag \mid \zg)\right\} + \var\left\{\E(\dzgt\dag\mid\zg)\right\}\nnb\\
&\oeq{\eqref{eq:za_1}}& 
\E\left\{\var(\dzgt\aag \mid \zg)\right\}
=
\E\Big\{ \dzgt \cdot \cov( \aag \mid \zg) \cdot \dzg \Big\}.
\endy
In addition, the definition of $\czg$ ensures
\begineq\label{eq:lem_eff_6}
\begin{array}{rcl}
\E\left(\czgt \czg\right) 
&=& \E\left\{\dsumig(\zi -e)^2 \right\} = \ng \cdot \vz, \medskip\\
\E\left( \czgt \cdot  \og \ogt \cdot \czg \right)
&=& \E\left(\ogt \cdot \czg \czgt  \cdot \og\right) 
= \ogt \cdot  \cov(\zg) \cdot \og \\
&=& \sumiipg \cov(\zi,\zip)
= \vz \cdot \ngsq\phig.
\end{array}
\endeq
Under \assmy, we have $
\cov(\aag\mid \zg) =  \ve   \ing +  \va   \og \ogt$ by \lema. 
Plugging this into \eqref{eq:za_10} and \eqref{eq:za_v_dzda_2} implies that 
\begina
\var\left\{\sumig(\czi)(\cai)\right\} &=& \var( \czgt \cag) 
\overset{\eqref{eq:za_10}}{=}  
\E\Big\{ \czgt \cdot \cov(\aag \mid \zg)\cdot \czg \Big\} \\
&\overset{\text{\lema}}{=}& 
\E\Big\{ \czgt \cdot \left(\ve  \ing +  \va   \og \ogt\right) \cdot \czg \Big\} \\
&=& \ve  \cdot \E\left(\czgt  \czg\right) + \va \cdot \E\left( \czgt \cdot  \og \ogt \cdot \czg \right)\\
&\overset{\eqref{eq:lem_eff_6}}{=}&   \ve\cdot \ng \cdot\vz  + \va \cdot \ngsq\phig \cdot \vz,\\
\var  \left\{\dsumig(\zi - \bzg)(\ai - \bag) \right\}   
 &=& \var( \dzgt\dag) 
 \overset{\eqref{eq:za_v_dzda_2}}{=} 
\E\Big\{ \dzgt \cdot \cov( \aag \mid \zg) \cdot \dzg \Big\}\\
&\overset{\text{\lema}}{=}& 
\E\Big\{ \dzgt \cdot \left(\ve  \ing +  \va   \og \ogt\right) \cdot \dzg \Big\} \\
&=& \ve  \cdot \E(\dzgt \dzg)
\oeq{\eqref{eq:szg}}  \ve \cdot \E(\szg) \\
&\overset{\text{\lem~\ref{lem:Z}}}{=}&  \ve \cdot \vz \cdot   \ng  ( 1 - \phig  ),
\enda 
where the last equality follows from  $\eszg =  \vz \cdot   \ng  ( 1 - \phig  )
$
by \lem~\ref{lem:Z}. 
\end{proof}

\section{Proofs of the results in \sec~\ref{sec:homo}}\label{sec:proof_homo}
\begin{proof}[\bf Proof of \propclt.]
\propclt\ follows from \thms~\ref{thm:lsx}--\ref{thm:fex} with $\xxi = \emptyset$. 
\end{proof}

\begin{proof}[\bf Proof of Theorem~\ref{thm:eff}]
Under \assmy, \lem~\ref{lem:efficiency} ensures that 
\beginy\label{eq:vls_vfe_N}
\beginar{rcl}
\ds\sumg\var\left\{\sumig(\czi)(\cai)\right\} 
&\overset{\text{Lem.~\ref{lem:efficiency}}}{=}& \ds\vz  \sumg  \left( \ve \cdot \ng + \va \cdot \ngsq\phig\right) = \ds N\vz  (\ve + \cn \va),\medskip\\
%
\ds  \sumg\var  \left\{\sumig(\zi - \bzg)(\ai - \bag) \right\} 
&\overset{\text{Lem.~\ref{lem:efficiency}}}{=}&  \ds\ve  \sumg \E(\szg)
\overset{\eqref{eq:szg}}{=} \ds  \ve N  \E( \sz)
\oeq{\eqref{eq:kn}} \ds N \vz   \kn \ve. 
\endar
\endy
This implies the simplified expressions of $\vls$ and $\vfe$, along with the ratio and the necessary and sufficient condition for $\vlsn/\vfen > 1$ in \eqref{eq:vr_def} and \eqref{eq:eff_cutoff}. 
%
\thmeff\eqref{it:eff_equal to 1} then follows from \eqref{eq:vr_def} and 
$
\{\cn = 0\} 
\Longleftrightarrow \{\kn = 1\} 
\Longleftrightarrow \{\text{$\phig = 0$ for all $g$}\}$ by \lemz.  
\end{proof}

\begin{proof}[\bf Proof of \propxeff]

\thmlsx\eqref{it:lsx_clt} ensures that 
\begineq\label{eq:effx_thm}
\dfrac{N\hsexsq}{\vlsxn} = 1+\op, \quad \vlsxn =   \dfrac{1}{\vzpcsq}\cdot \nifsumg \var\left\{\sumig (\zi-e)\rrilsx\right\},
\endeq 
with $\rrilsx = \res(\aai \mid 1, \xxi)$ as defined in \eqref{eq:rrilsx_def}.
Under \assmy, let 
\begineqs
\proj(\ag \mid 1, \xsg) =\boa + \xsgt\bxa, \quad
\agp  = \res(\ag\mid 1, \xsg) =\ag - \boa - \xsgt \bxa
\endeqs
denote the linear projection of $\ag$ onto $(1, \xsg)$ and the corresponding residual, with 
\begineq\label{eq:agp_e_var}
\E(\agp) = 0, \quad 
\E(\xsg \agp) = 0,\quad  \vap = \var(\agp) = \va - \var\left\{ \proj(\ag \mid 1, \xsg) \right\}. 
\endeq
Let $
\ds\czg = \zg - \ong e$ and $\rrglsx = (R_{g1,\lsx}, \ldots, R_{g,\ng,\lsx})^\T$
with
\begineq\label{eq:effx_z}
\beginar{l}
\ongt \cdot\cov(\czg)\cdot \ong = \ongt \cdot\cov(\zg)\cdot \ong = \dsumiipg \cov(\zi, \zip) = \vz \cdot \ngsq\phig,
\smallskip\\
\czgt \czg = \dsumig (\czi)^2, \qquad \E(\czgt \czg) = \dsumig \E\left\{(\czi)^2\right\} = \vz\cdot\ng. 
\end{array}
\endeq
We show below that, under \assms~\ref{assm:y}--\ref{assm:y_x},
\begineq\label{eq:effx_goal} 
\nifsumg \var\left\{\sumig (\zi-e)\rrilsx\right\}
=\nifsumg \var\left(\czgt\rrglsx\right)= \vz (\vap \cdot \cn + \ve ). 
\endeq
Combining \eqref{eq:effx_thm} and \eqref{eq:effx_goal} ensures $
\dfrac{\hselsx^2}{\hsels^2} = \dfrac{\ve + \vap \cdot \cn }{\ve + \va \cdot \cn }+ \op,$
where $
\dfrac{\ve + \vap \cdot \cn }{\ve + \va \cdot \cn }
= 1 +   \dfrac{(\vap-\va) \cdot\cn }{\ve + \va \cdot \cn}  
 \oeq{\eqref{eq:agp_e_var}}  1 - \var\left\{\proj(\ag\mid1, \xsg)\right\} \cdot\dfrac{\cn}{\ve + \va \cdot \cn}.$

\paragraph*{Proof of \eqref{eq:effx_goal}.}

 \assm~\ref{assm:y_x} ensures  $\E(\epg) = \E\big\{\E(\epg \mid \xg,\ag)\big\} = 0$, $\cov(\epg, \xxi) = \E(\epg\xit) = \E\big\{\E(\epg\mid \xg, \ag) \cdot \xit\big \} = 0$, and $\cov(\epg, \ag) \oeqt{symmetry} 0$, 
so that 
\begineq\label{eq:proj_epi}
\proj(\epi \mid 1, \xxi) = 0, 
\qquad 
\cov(\epg, \agp) =\cov(\epg, \ag) - \cov(\epg, \xsgt)\bxa = 0. 
\endeq
Combining this with $\ai = \aci + \epi$ implies that under \assmyyx,
\begineqs
\beginar{l}
\proj(\ai\mid \xxi) = \proj(\aci \mid \xxi) + \proj(\epi \mid \xxi) \oeq{\eqref{eq:proj_epi}} \proj(\aci \mid \xxi),\nnb\\
\rrilsx 
 \oeq{\eqref{eq:rrilsx_def}} \ai - \proj(\ai\mid \xxi)  
 =  \aci + \epi  - \proj(\aci\mid \xxi) 
 \oeq{\eqref{eq:agp_e_var}}   \acip + \epi,\nnb\\
\rrglsx = (R_{g1,\lsx}, \ldots, R_{g,\ng,\lsx})^\T = \ong\agp + \epg,  \endar
\endeqs
so that 
\begineq\label{eq:rrglsx_e_cov}
\beginar{l}
\rrglsx \indep \zg,
\quad
\E(\rrglsx) \oeq{\eqref{eq:agp_e_var}} \zng, \quad
\cov(\rrglsx) \oeq{\eqref{eq:proj_epi}+\eqref{eq:agp_e_var}}  \vap\jng + \ve\ing. 
 \endar
\endeq 
Combining \eqref{eq:rrglsx_e_cov} and \eqref{eq:effx_z} ensures
\begineqs
\beginar{rcl}
\E(\czgt\rrglsx \mid \zg) 
&\oeq{\eqref{eq:rrglsx_e_cov}}&
 \czgt\cdot \E( \rrglsx) \oeq{\eqref{eq:rrglsx_e_cov}} 0,
\smallskip\\
\var(\czgt\rrglsx \mid \zg) 
&\oeq{\eqref{eq:rrglsx_e_cov}}& \czgt \cdot \cov(\rrglsx ) \cdot \czg
\oeq{\eqref{eq:rrglsx_e_cov}} \czgt \cdot (\vap \jng + \ve \ing)\cdot \czg\\
&=& \vap\ongt\czg \czgt \ong + \ve \czgt\czg,
\smallskip\\
\E\{\var(\czgt\rrglsx \mid \zg)\} 
&=& \vap \cdot \ongt \cdot \cov(\zg)\cdot \ong + \ve \cdot \E(\czgt \czg)\\
&\oeq{\eqref{eq:effx_z}}& \vtgf,
\endar
\endeqs
so that $\var(\czgt\rrglsx) 
=\E\{\var(\czgt\rrglsx \mid \zg)\}+\var\{\E(\czgt\rrglsx\mid \zg)\}= \vtgf.$
This implies \eqref{eq:effx_goal}. 
\end{proof}

\section{Proof of \thmhetero\ in \sec~\ref{sec:hetero_plim}}\label{sec:proof_hetero}

\begin{lemma}\label{lem:hetero}
Let $\olzd = \meani \zidi$ and $\olzy = \meani \ziyi$. 
As $\ntinf$, if (i) \assmhetero\ hold; (ii) $\ds\max_{\otg} \ng/N \to 0$; (iii) $ \sup_{\otn}  \E ( \yi^2 ) < \infty$, then 
\begineqs
\beginar{lllll}
\bz = \E(\bz) + \op, &\quad& \bd = \E(\bd) + \op, && \by = \E(\by) + \op,\\
\olzd = \E(\olzd) + \op, &\quad& \olzy = \E(\olzy) + \op,\\
\szd = \E(\szd) + \op, &\quad& \szy = \E(\szy) + \op.  
\endar
\endeqs
\end{lemma}

\begin{proof}[Proof of \lemhetero]
Recall from \eqref{eq:sz_szy_szd} that $\szd$ and $\szy$ are the sample means of $\dzi\ddi$ and $\dzi\dyi$, respectively. 
First, $|\zi|$, $|\di|$, $|\zidi|$, and $|\dzi\ddi|$ are all bounded between $[0,1]$, and therefore \uni\ over $\otn$. 
In addition, 
\lem~\ref{lem:bound}\eqref{it:lem_bound_2} ensures
$\E (|\dzi\dyi|^2 )\leq \E ( \dyi ^2 ) 
\le 2^2 \E( \yi^2),$
so that condition (iii) $ \sup_{\otn}  \E ( \yi^2 ) < \infty$  ensures 
that $\ziyi$, $\yi$, and $\dzi\dyi$ are all \uni\ over $\otn$.
The results then follow from applying \lemsm\eqref{it:lem_sm_wlln} to 
$
\bi = \zi,\, \di,\, \yi,\, \zidi,\, \ziyi,\, \dzi\ddi,\, \dzi\dyi,
$
respectively. 
\end{proof}

\begin{proof}[\bf Proof of \thmhetero]
Let $
\tyi = \yiz + \oua \ti$.
Let $\tmug$ denote the common value of $\E(\tyi)$ for units within cluster $g$.
The cluster-specific results in \lemcef\ in the \sa\ ensure that 
\begineq
\label{eq:hetero_eyz}
\E(\yg \mid \zg) = \og \tmug + \zg\pcg\tcg, \quad
\E(\dg \mid \zg) = \og \pag + \zg \pcg.
\endeq
%

\paragraph*{\bf Probability limit of $\htfe$.} 
Recall from \eqref{eq:szg} that $\szg = \sumig (Z_i - \bzg)^2 = \zgt\png \zg$. 
Let $\szyg = \sumig(\zi - \bzg)(\yi - \byg) = \zgt \png \yg$. 
\eq~\eqref{eq:hetero_eyz} ensures that 
\begina
\E(\szyg \mid \zg) 
&=& \E(\zgt \png \yg \mid \zg)\\
&=& \zgt \png \cdot \E( \yg \mid \zg) 
\oeq{\eqref{eq:hetero_eyz}} \zgt\png\zg\cdot \pcg\tcg= \pcg \tcg \cdot \szg, 
\enda
so that $\E(\szyg) = \E\{ \E(\szyg \mid \zg)\} = \pcg \tcg \cdot \eszg$ and 
\begineq\label{eq:hetero_zpiny}
\E(\szy) = \ni \sumg \E\left(\szyg\right) = \nisumg \pcg\tcg \cdot \eszg. 
\endeq
Identical reasoning ensures 
\beginy\label{eq:hetero_zpind}
\E(\szd) = \nisumg \pcg \cdot \eszg
\endy
from \eqref{eq:hetero_eyz}.
It then follows from \eqref{eq:num} and \lemhetero\ that
\begina
\htfe 
\oeqt{\eqref{eq:num}}  \dfrac{\szy}{\szd}\oeqt{\lemhetero}  \dfrac{\E(\szy)}{\E(\szd)} + \op \oeq{\eqref{eq:hetero_zpiny} + \eqref{eq:hetero_zpind}} \sumg \kgfe \cdot \tcg + \op,
\enda
where
\begineq\label{eq:kgfe_def}
\kgfe = \dfrac{ \pcg \cdot \E(\szg)}{\sumg \pcg \cdot \E(\szg)}. 
\endeq
In addition, \lemz\ in the \sa\ implies that 
\beginy\label{eq:eszg_hetero}
\eszg \oeqt{\lemz} \vzg \cdot \ng(1-\phig)
\endy
under \assmhetero, so that $ 
\kgfe \overset{\eqref{eq:kgfe_def}+\eqref{eq:eszg_hetero}}{\propto}
\ng (1-\phig)\vzg \cdot  \pcg.$

%

\paragraph*{\bf Probability limit of $\htls$.} 
Recall from \eqref{eq:num} that  $\htols = \szyt / \szdt$. 
We compute below the probability limits of $\szyt$ and $\szdt$, respectively, which together imply the results. 
 
First, \lemcef\eqref{it:lem_cef_EYZ}--\eqref{it:lem_cef_EDZ} in the \sa\ imply that
\beginy\label{eq:ey_ed_hetero}
\begin{array}{llllll}
\E(\yi)= \tmug + \eg \pcg\tcg, &&
\E(\ziyi) = \eg (\tmug + \pcg\tcg),&&
\covzy = \vzg \cdot \pcg\tcg ,\\
 \E(\di) = \pag + \eg \pcg, &&
\E(\zidi) =  \eg (\pag + \pcg), &&
\covzd  = \vzg \cdot \pcg   
\end{array}\qquad
\endy
for $\ig$ under \assmhetero.
This ensures
\beginy\label{eq:hetero_zpd_0}
\begin{array}{rcl}
\E(\bz) &=&  \ds \meangi \E(\zi) = \nisumg \ng\eg,\medskip\\
\E(\bd) &=&  \ds \meangi \E(\di) \oeq{\eqref{eq:ey_ed_hetero}} \nisumg \ng(\pag + \eg \pcg),\\
\E\left(\dmeani \zidi\right) &=& \ds \meangi \E(\zidi)  \oeq{\eqref{eq:ey_ed_hetero}} \dnisumg \ng \eg (\pag + \pcg).
\end{array}
\endy
The probability limit for $\szdt$ in the denominator then follows from plugging \eqref{eq:hetero_zpd_0} into
\begineqs
\szdt =  \meani \zidi - \bz \bd 
\oeqt{\lemhetero} \E\left(\meani\zidi\right) - \E(\bz) \cdot \E(\bd) + \op. 
\endeqs
Identical reasoning ensures the expression for $\szyt$ from \eqref{eq:ey_ed_hetero}. 
\end{proof}

\section{Proof of Theorem~\ref{thm:joint} in \sec~\ref{sec:test}}\label{sec:app_proof_of_joint}

Under \assmtth, recall   from \thmjoint\ that $
\hsiglsfe =  
\beginp \hsesq_\ls & \hclsfe \\ \hclsfe & \hsesq_\fe\endp = \sum_{g=1}^{G} \hat{v}_g \hat{v}_g^{\top}$,
where 
$
\hat{v}_g = \dfrac{1}{N}\begin{pmatrix} S_{ZD}^{-1}  \sumig (Z_i - \bar{Z}) \rils \\ \szd^{-1} \sumig (Z_i - \bar{Z}_g) \rife \end{pmatrix}$. 
Let 
$
V_N =
\beginp \vls & \vlsfe \\ \vlsfe & \vfe\endp
= 
N^{-1}\sumg \cov(v_g),  
$
where
$v_g = \begin{pmatrix} (\vzpc)^{-1}  \sumig (Z_i - e) (A_i-\mua) \\ (\kvzpc)^{-1} \sumig (Z_i - \bar{Z}_g) (A_i - \bar{A}_g) \end{pmatrix}$. 
Intuitively, $(V_N, v_g)$ define the population analogs of $(\hsiglsfe, \hat v_g)$, respectively, scaled by a factor of $N$.
We will show that
\begineqs
\beginp 
\htls - \tc\\
\htfe - \tc
\endp 
\approx N^{-1}\sumg v_g,
\endeqs
so that $v_g$ is essentially the influence function of cluster $g$. 
We verify in \secs~\ref{sec:joint_1}--\ref{sec:joint_2} that 
\begineq\label{eq:goal_joint}
\vnsq \cdot \sqrtn \beginp 
\htls - \tc\\
\htfe - \tc
\endp 
\rs
 \mn(0_2, I_2),\quad 
  \vvhv = I_2 + \op, 
\endeq
respectively, which together imply $
\hsiglsfe^{-1/2} \beginp 
\htls - \tc\\
\htfe - \tc
\endp 
\rs
 \mn(0_2, I_2).$

\subsection{Verification of $\vnsq \cdot \sqrtn (\htls - \tc, \htfe - \tc)^\T\rs \mn(0_2, I_2)$ in \eqref{eq:goal_joint}}
\label{sec:joint_1}

Let $\hbo$ denote the coefficient on the constant term from $\tslst(\yi \sim 1 + \di \mid 1 + \zi )$.
Given that $\tslst(\yi \sim 1 + \di \mid 1 + \zi )$ is just-identified, the estimation equation implies
\begineq\label{eq:ee_ols}
\hbolsf
=
\left\{ \meani \oz \od \right\}^{-1}  
 \left\{\meani \oz Y_i\right\}\\
=
 \phini
 \left\{\meani \oz Y_i\right\}, 
\endeq
where $\phin =  \phinf$. 
Recall from \eqref{eq:ai} that $\ai = \yi - \di\tc$ with $\mua = \E(\ai)$. 
\eq~\eqref{eq:ee_ols} ensures
\begina
\dbolsf
&=& 
\hbolsf - \bolsf\\
&\overset{\eqref{eq:ee_ols}}{=}& 
\phini \left\{\meani \oz Y_i\right\} 
-
\phini  \underbrace{\left\{\phinf \right\}}_{=\phin}  \bolsf\\
&=& 
\phini \cdot \meani \oz \left\{ Y_i -  \od\bolsf \right\}
\nonumber\\
&\overset{\eqref{eq:ai}}{=}&  \phini \cdot \meani \oz (\aai - \mua),
\enda
so that 
\begineq\label{eq:htls_joint}
\htls - \tc
=
(0,1)\dbolsf
=
(0,1)\phini 
 \cdot \dmeani \beginp\aai - \mua\\ \zi(\aai-\mua)\endp.
\endeq
In addition, \eqref{eq:num} ensures that 
\begineq\label{eq:htfe_joint}
\htfe - \tc =
 \dfrac{\meani \dzi \dyi  }{\meani \dzi \ddi} -\tc \overset{\eqref{eq:ai}}{=}  \szd^{-1} \left( \dmeani\dzi \dai\right).
\endeq
Let 
\begineqs
\bi = \beginp 
\cai\\
\zi(\cai)\\
\dzi \dai 
\endp,
\quad 
\bb = \dmeani \bi,
\endeqs 
with
\begineq\label{eq:omgn_def_app}
\ds\omgn \oeqt{Def.} \ni\cov\left\{\sumi \beginp 
\cai\\
\zi(\cai)\\
\dzi \dai 
\endp \right\} = \cov(\sqrtn \bb). 
\endeq
\eqs~\eqref{eq:htls_joint}--\eqref{eq:htfe_joint} ensure that
\begineq\label{eq:dhtlsfe}
\beginp 
\htls - \tc\\
\htfe - \tc
\endp 
 \overset{\eqref{eq:htls_joint}+\eqref{eq:htfe_joint}}{=} 
\beginp
(0,1)\phini  & 0\\
0 & \szd^{-1}
\endp 
\cdot \meani 
\beginp 
\cai\\
\zi(\cai)\\
\dzi \dai 
\endp
=
\gn  \bb,
\endeq
where $\gn =  
\beginp
(0,1) \phini  & 0\\
0 & \szdi
\endp.$
On the other hand, 
\begineq\label{eq:vg}
\nisumg v_g = \nisumg\begin{pmatrix} (\vzpc)^{-1}  \sumig (Z_i - e) (A_i-\mua) \\ (\kvzpc)^{-1} \sumig (Z_i - \bar{Z}_g) (A_i - \bar{A}_g) \end{pmatrix} = \gns \bb,
\endeq 
where
\begineq\label{eq:gns_joint}
\gns = \gnsf. 
\endeq 
We show below that
\begineq \label{eq:joint_goal}
\vn^{-1/2}\cdot\sqrtn\gns\bb \rs \mn(0_2, I_2), \qquad
\vn^{-1/2}(\gn - \gns)\sqrtn\bb = \op,
\endeq
respectively, based on the building blocks
\begineq
\gn - \gns = \op, \qquad \omgn^{-1/2} \cdot \sqrtn \bb \rs \mn(0_3, I_3). \label{eq:joint_goal_clt_bi}
\endeq 
Combining \eqref{eq:joint_goal} with \eqref{eq:dhtlsfe} implies $
\vn^{-1/2}\sqrtn\beginp 
\htls - \tc\\
\htfe - \tc
\endp 
\oeq{\eqref{eq:dhtlsfe}}
\vn^{-1/2}\sqrtn\gn\bb =
\vn^{-1/2}\sqrtn\gns\bb + \vn^{-1/2}(\gn - \gns)\sqrtn\bb \rs \mn(0_2, I_2)$
by Slutsky's theorem. 

\subsubsection{Proof of $\gn - \gns = \op$ in \eqref{eq:joint_goal_clt_bi}.}
\lemasymnox\eqref{it:lem_asym_nox_S} implies  $
\phin \oeq{\eqref{eq:ee_ols}} \beginp 1 & \bd \\ \bz & \meani \zidi\endp =
\eozodf + \op$. 
Combining this with $\cov(\zi,\di) = \vzpc$ from \lemcef\eqref{it:lem_cef_EDZ} implies that
$
\text{plim} \phin^{-1} 
=\eozodf^{-1}
= \vzpci\beginp
 \muzd  & -\mud\\
 -e & 1
 \endp 
$
with
\beginy\label{eq:einv}
\text{plim} (0,1)\phin^{-1} 
=  \vzpci(- e, 1). 
\endy
In addition,  \lemcef\eqref{it:lem_cef_EDZ} ensures 
\begineq\label{eq:joint_szd}
\szd =\kvzpc + \op.  
\endeq
Plugging \eqref{eq:einv} and \eqref{eq:joint_szd} into \eqref{eq:dhtlsfe} implies the result. 

\subsubsection{Proof of $\omgn^{-1/2} \cdot \sqrtn \bb \rs \mn(0_3, I_3)$ in \eqref{eq:joint_goal_clt_bi}.}
\lemasymnox\eqref{it:lem_asym_nox_AB} ensures that $\{\en{\bi}^2: \otn\}$ is \uni\ under \assmhomo\ and $\E(Y_i^4)<\infty$.
In addition, \eqref{eq:omgn_def_app} ensures $\cov(\sqrtn \bb) = \omgn$, with $\inf_N\lm(\omgn) > 0$ by assumption.  
Therefore, $\bi$ satisfies the conditions in \lemsm\eqref{it:lem_sm_clt} with $r = 2$.  
\lemcefa\ ensures 
\begineq\label{eq:eb = 0}
\E(\bi) = \E\beginp 
\cai\\
\zi(\cai)\\
\dzi \dai 
\endp = \E\beginp 
0\\
\cov(\zi,\ai)\\
\E(\dzi \dai)
\endp \oeqt{\lemsm\eqref{it:lem_sm_clt}} 0, \qquad \E(\bb) = 0. 
\endeq
Applying \lemsm\eqref{it:lem_sm_clt}\ to  $\{\bi: \otn\}$ ensures
$\omgn^{-1/2} \cdot \sqrtn \bb   \rs  \mn(0_3, I_3)$. 

\subsubsection{Proof of $\vn^{-1/2} \cdot \sqrtn \gns\bb \rs \mn(0_2, I_2)$ in \eqref{eq:joint_goal}.} 
Let 
\begineq\label{eq:bis_joint}
\bis  = \gns \bi =  \dfrac1\vzpc\beginp (\czi)(\cai) \\ \kni \dzi\dai\endp,
\quad 
\bbs = \dmeani \bis = \gns \bb, 
\endeq 
with 
\beginy
\E(\bbs)  \oeq{\eqref{eq:eb = 0}+\eqref{eq:bis_joint}} 0,  \quad
%
\cov(\sqrtn\bbs)  \oeq{\eqref{eq:vg}} \cov\left(\dfrac{\sumg v_g}\sqrtn  \right) \oeqt{\text{Assm.~\ref{assm:cs}}} \nif \sumg \cov ( v_g )  = \vn. \quad\label{eq:cov_bi*}
\endy 
We verify below that, under the regularity conditions in \thmjoint\ with $\lm(\omgn) \geq \lmd > 0$ for some $\lmd >0$, we have
\begineq\label{eq:bis_clt_goal}
\supi\E(\en{\bis}^4) = O(1), \quad
\lm(\vn) \geq \dfrac{\lmd}\vzpcsq > 0,  
\endeq
so that $\{\bis: \otn\}$ satisfies the conditions in \lemsm\eqref{it:lem_sm_clt}. Applying \lemsm\eqref{it:lem_sm_clt} to $\{\bis: \otn\}$ implies the result.

\paragraph{Proof of $\supi\E(\en{\bis}^4) = O(1)$.} Direct algebra ensures
\begineq\label{eq:gns_gnst}
\gns\gnst = \dfrac1\vzpcsq \beginp e^2 + 1 & 0 \\ 0 & \kn^{-2}\endp.  
\endeq 
Given $\inf_N \kn >0$, 
\begineq\label{eq:lm_max}
a \equiv \spn{\gns}^2 = \lmax(\gns\gnst) = \dfrac1\vzpcsq\cdot \max(e^2+1, \kn^{-2})  = O(1).   
\endeq
\lemen\ ensures 
\begineq\label{eq:bis_bound}
\|\bis\| \oeq{\eqref{eq:bis_joint}} \|\gns\bi\| \oleqt{\eqref{eq:lm_max}+\lemen} a^{1/2} \|\bi\|.
\endeq 
Under \assmhomo\ and $\E(\yi^4) = O(1)$,  \lemasymnox\eqref{it:lem_asym_nox_AB} implies that $ \supi\E(\en{\bi}^4) = O(1)$. 
Combining this with \eqref{eq:bis_bound} implies the result. 

\paragraph{Proof of  the uniform boundedness of $\lm(\vn)$  in \eqref{eq:bis_clt_goal}.}
Recall from \eqref{eq:omgn_def_app} that $\omgn = \cov(\sqrt N\bb)$, so that  
\begineq\label{eq:cov_bbs_joint}
\vn \oeq{\eqref{eq:cov_bi*}}
\ds\cov(\sqrtn\gns\bb) 
 \oeq{\eqref{eq:omgn_def_app}}\gns \omgn \gnst. 
\endeq
Combining \eqref{eq:gns_gnst} with $\kn \leq 1$  from \lemz\  ensures   $
\lm(\gns\gnst) = \dfrac1\vzpcsq\cdot\min(e^2+1, \kn^{-2})  \geq  \dfrac1\vzpcsq$, 
so that, when $\lm(\omgn) \geq \lmd > 0$, \lemlm\ implies that   
$\lm(\vn) = \lm(\gns \omgn \gnst)  \ge \lm(\omgn)\cdot\lm(\gns \gnst)  \oeqt{\text{Assm.~\ref{assm:asym}}} \dfrac{\lmd}\vzpcsq$. 

\subsubsection{Proof of $\vn^{-1/2}(\gn - \gns)\sqrtn\bb = \op$ in \eqref{eq:joint_goal}.}
The second part of \eqref{eq:bis_clt_goal} implies that  $\vn^{-1/2} = \ooo$. 
Combining this with $\gn-\gns = \op$ and $\omgn^{-1/2} \cdot \sqrtn \bb \rs \mn(0_3, I_3)$ from \eqref{eq:joint_goal_clt_bi} implies that 
\begineqs
\ses  (\gn - \gns)   \sqrtn \bb 
= \underbrace{\ses}_{\textup{$\ooo$ by \eqref{eq:bis_clt_goal}}}  \underbrace{(\gn - \gns)}_{\textup{$\op$ by \eqref{eq:joint_goal_clt_bi}}}  
\omgn^{1/2}
\underbrace{\left( \omgn^{-1/2} \cdot \sqrtn \bb\right)}_{\textup{$\rs \mn(0_3, I_3)$ by \eqref{eq:joint_goal_clt_bi}}} 
\overset{\eqref{eq:joint_goal_clt_bi}+\eqref{eq:bis_clt_goal}+\eqref{eq:joint_goal_clt_bi}}{\rs} 0_2 
\endeqs
by Slutsky's theorem, so that 
$
\ses  (\gn - \gns)  \sqrtn \bb = \op 
$
by \lem~\ref{lem:basic}.

\subsection{Verification of \eqref{eq:goal_joint}: $\vvhv = I_2 + \op$}
\label{sec:joint_2}
Note from \eqref{eq:bis_joint} that $v_g = \sumig \bis$. Let 
\begineq\label{eq:tvn_def}
\tvn = \nisumg \tbgs  \tbgst 
\endeq
be a variant of $\vn = \nisumg \cov(v_g)$ to write
\begineq\label{eq:joint_cov_goal}
 \vvhv  = \vnsq (\hvn - \tvn) \vnsq + \vnsq  \tvn \vnsq. 
\endeq
Given that $\vn = \cov(\sqrtn \bb^*)$ from \eqref{eq:cov_bi*} and 
$
\E(\bi^*) \oeq{\eqref{eq:bis_joint}} \E(\gns \bi) = \gns \E(\bi) \oeq{\eqref{eq:eb = 0}} 0
$ from \eqref{eq:bis_joint} and \eqref{eq:eb = 0}, applying \lemsm\eqref{it:lem_sm_clt}\ to $\{\bi^*:\otn\}$ ensures 
\begineq\label{eq:tvn_plim}
\vn^{-1/2}\tvn\vn^{-1/2} = I_2 + \op. 
\endeq
Combining \eqref{eq:joint_cov_goal}, \eqref{eq:tvn_plim}, and $\lm(\vn) \geq \lmd/\vzpcsq > 0$ from \eqref{eq:bis_clt_goal} implies that 
to verify $\vvhv = I_2 + \op$, it suffices to verify that
\begineq\label{eq:sc_joint_1}
\hvn - \tvn = \op. 
\endeq
Write
\begineqs
\hvn = 
\beginp N\hsels^2 & N \hclsfe\\ N\hclsfe & N\hsefe^2 \endp, \quad \tvn = \beginp \tvn[1,1] & \tvn[1,2] \\ \tvn[1,2] & \tvn[2,2] \endp.  
\endeqs 
We verify in \sec~\ref{sec:joint_cov_1} that 
\begineq\label{eq:goal_joint_hse}
N \hclsfe = \tvn[1,2] + \op.
\endeq
The proofs for $N\hsels^2 = \tvn[1,1] + \op$ and $N\hsefe^2 = \tvn[2,2] + \op$ are similar and therefore omitted. Together, these componentwise convergence results imply \eqref{eq:sc_joint_1}.

\subsubsection{Proof of $N\hclsfe = \tvn[1,2] + \op$ in \eqref{eq:goal_joint_hse}.}\label{sec:joint_cov_1}
Let $\czg = (\czi:\ig) = \zg - \ong e$ and $\cag = (\cai:\ig) = \ag - \ong\mua$
to write  
\begineq\label{eq:tbg*}
\tbgs  
= \sumig\dfrac1\vzpc\beginp (\czi)(\cai) \\ \kni \dzi\dai\endp
= \dfrac1\vzpc \beginp
\czgt \cag\\
   \kni   \dzgt \dag
\endp. 
\endeq
Plugging \eqref{eq:tbg*} into the definition of $\tvn$ in \eqref{eq:tvn_def} ensures that 
\beginy\label{eq:tvn_12}
\tvn[1,2]
\oeq{\eqref{eq:goal_joint_hse}}
(1,0)\tvn\beginp 0 \\ 1\endp
&\oeq{\eqref{eq:tvn_def}}& 
(1,0)\left( \nifsumg \tbgs  \tbgst \right)\beginp 0 \\ 1\endp
\nnb\\
&\oeq{\eqref{eq:tbg*}}& 
 \nifsumg
 (1,0)v_g v_g^\T\beginp 0 \\ 1\endp 
\nnb\\
&=& 
\dfrac1{\vzpc} \cdot \dfrac1{\kvzpc} \cdot \nifsumg  \left(\czgt \cag\right)\left(\dzgt \dag\right)  
\nnb\\
&=& 
\dfrac1{\vzpc} \cdot \dfrac1{\kvzpc} \cdot \ttt,
\endy
where $
\ttt = \nisumg \tttg$ and $\tttg =  (\czgt \cag ) (\dzgt \dag)$. 
On the other hand, 
\beginy\label{eq:hclsfe_app}
N\hclsfe 
&=& 
\dfrac1\szdt\cdot \dfrac1\szd\cdot \nifsumg\left\{\sumig (\zi - \bz)\riols\right\} \left(\sumig  \dzi \rife \right)\nnb\\
&=& 
\dfrac1\szdt\cdot \dfrac1\szd\cdot \nifsumg\left\{(\zg - \ong \bz)^\T \rgols\right\} \left(\dzgt\rgfe \right)\nnb\\
&=& 
\dfrac1\szdt\cdot \dfrac1\szd\cdot \htt, 
\endy
where $\htt = \nisumg \httg$ and $\httg = \{(\zg - \ong \bz)^\T \rgols\} (\dzgt\rgfe)$. 
Given that $
\szdt = \vzpc + \op$ and $\szd = \kvzpc + \op$
by \lemasymnox\eqref{it:lem_asym_nox_S}, from  \eqref{eq:tvn_12}--\eqref{eq:hclsfe_app}, it suffices to verify that 
\beginy\label{eq:goal_T}
\htt - \ttt = \op. 
\endy

\paragraph{Proof of $\htt - \ttt = \op$ in \eqref{eq:goal_T}.} 
Let $\dols = \dbolsf$ and $\dfe = \htfe-\tc$. 
It follows from $\rils = \yi - \hbo - \di\htls$, $\ai = \yi - \di\tc$, and $\rife = \dyi - \ddi\htfe$,   $\dai = \dyi - \ddi\tc$
that 
\begina
\rils - (\cai) = - (\hbo - \mua) - \di(\htls - \tc) = - (1, D_i) \dols , \quad
\rife - \dai = - \ddi\dfe,
\enda
\begineq\label{eq:rgols_rgfe}
\rgls - \cag 
= -(\ong,\dg) \dols,
\quad
\rgfe - \dag  =  - \ddg \dfe. 
\endeq
Combining \eqref{eq:rgols_rgfe} with 
\begineq\label{eq:zg}
\zg - \ong \bz = \czg - \ong(\bz - e) 
\endeq
yields
\beginy\label{eq:zr_joint}
\beginar{rcl}
(\zg - \ong \bz)^\T \rgols
&\oeq{\eqref{eq:zg} + \eqref{eq:rgols_rgfe}}& 
\left\{\czgt  - (\bz - e)\ongt\right\}\left\{\cag - (\ong,\dg) \dols\right\}\\
&=& 
 \czgt \cag - \czgt(\ong,\dg) \dols \\
&& - (\bz - e)\ongt \cag + (\bz - e)\ongt (\ong,\dg) \dols \\
&=& \czag + \xigols,
\medskip\\
\dzgt\rgfe 
&\oeq{\eqref{eq:rgols_rgfe}}&
\dzgt (\dag -  \ddg \dfe) 
= \dzag + \xigfe,
\endar
\endy
where 
$\xigols = - \czgt(\ong,\dg) \dols  - (\bz - e)\ongt \cag + (\bz - e)\ongt (\ong,\dg) \dols$ and $\xigfe = - \dzgt\ddg\dfe$.
Plugging \eqref{eq:zr_joint} into the definition of $(\htt, \httg)$ in \eqref{eq:hclsfe_app} ensures
\beginy
\httg 
&\oeq{\eqref{eq:hclsfe_app}}& \left\{(\zg - \ong \bz)^\T \rgols\right\} \left(\dzgt\rgfe \right)
=
\left(\czag + \xigols\right)\left(\dzag + \xigfe\right)
\nnb\\
&=& 
\underbrace{\czagp \dzagp}_{\text{$\tttg$ by \eqref{eq:tvn_12}}} + \xigfe \czag 
+ \xigols \dzag  + \xigols \xigfe,
\nnb\\
\htt 
&=& \nifsumg \httg = 
\ttt + S_1 + S_2 + S_3, \label{eq:htt}
\endy
where  
\begineqs
S_1 = \nisumg\xigfe \czag, \quad S_2 = \nisumg\xigols \dzag, \quad S_3 = \nisumg\xigols \xigfe.
\endeqs 
From \eqref{eq:htt}, a sufficient condition for $\htt-\ttt = \op$ in \eqref{eq:goal_T} to hold is that 
the three terms on the right-hand side of \eqref{eq:htt} are all $\op$. 
We verify below that $S_1 = \op$. The proofs of the remaining two equalities are analogous and therefore omitted.

\paragraph{Proof of $S_1 = \op$ in \eqref{eq:htt}.}
Recall from \eqref{eq:zr_joint} that $\xigfe = - \dzgt\ddg\dfe$ with $\dfe = \htfe-\tc$ from \eqref{eq:rgols_rgfe} to write
\begineq\label{eq:joint_T3}
S_1 = - \dfe \cdot \nisumg \dzdgp \left(\czag\right).  
\endeq
The central limit theorem for $(\htols -\tc, \htfe-\tc)^\T$ in \eqref{eq:goal_joint} that we just verified implies 
\begineq\label{eq:htfe_plim}
\dfe  = \op.
\endeq
We show below that  
\begineq\label{eq:goal_joint_bounded}
T = \nisumg  \dzdgp \czagp = \oop. 
\endeq
Combining \eqref{eq:joint_T3}--\eqref{eq:goal_joint_bounded} implies the result. 

\paragraph{Proof of \eqref{eq:goal_joint_bounded}.}
Given $(\dzdg)(\czag)\leq 2^{-1}\{(\dzdg)^2 +  (\czag)^2\}$, 
\begineq\label{eq:bound_joint_1}
T = \nisumg \dzdgp \czagp  \leq 2^{-1}(T_1 + T_2), 
\endeq
where $T_1 = \nisumg\dzdgp^2$ and $T_2= \nisumg\czagp^2$.
It follows from $\dzi, \ddi \in (-1,1)$ that $|\dzdg | =  |\sumig \dzdi|  \leq \ng$, 
so that 
\begineq\label{eq:bound_joint_T1}
T_1 = \nisumg\dzdgp^2 \leq \nisumg\ngsq = O(1)
\endeq
under \assm~\ref{assm:ng_main}. 
Similarly, it follows from $|\zi - e| < 1$ that 
\begineqs
\left(\czag\right)^2 = \left|\sumig \czai\right|^2 
< \left(\sumig |\cai|\right)^2 
\leq \ng  \sumig  (\cai)^2
\endeqs
by \csf,
so that 
\begineqs
\beginar{c}
T_2 = \nisumg\czagp^2   \leq \nisumg \ng \left\{ \sumig  (\cai)^2\right\},\\
\E(|T_2|) 
 \leq  \nisumg \ng \left[\sumig \E\left\{(\cai)^2\right\} \right] = \var(A_i) \cdot \nisumg\ngsq = \ooo
 \endar
\endeqs
under \assm~\ref{assm:ng_main}. This implies 
\begineq\label{eq:bound_joint_T2}
T_2 = \oop
\endeq 
by \lem~\ref{lem:basic}. 
Plugging \eqref{eq:bound_joint_T1}--\eqref{eq:bound_joint_T2} into \eqref{eq:bound_joint_1} ensures $T = \oop$ in \eqref{eq:goal_joint_bounded}.

\section{Proof of the results in Appendix~\ref{sec:complete theory}}
\label{sec:proof_complete}
\subsection{Proof of \thmlsx}
Recall that $\zix$ denotes the residual from $\olst(\zi\sim 1+\xxi)$, with $
\szyx = \meani \zix \yi$ and $\szdx = \meani \zix \di$.
Letting $\dis = \di$, $\zis = \zi$, and $W_i = (1, \xxi)$ in \lemfwl\ ensures the numeric expressions of $(\htlsx, \hselsx)$ in \thmlsx\eqref{it:lsx_numeric}. 
We verify below \thmlsx\eqref{it:lsx_clt}. 

\paragraph*{Proof of \thmlsx\eqref{it:lsx_clt}.} Recall from \eqref{eq:rrilsx_decomp} that $
\rrilsx =  Y_i - \bo - \di\tc - \xit\bxaa$, 
where $\bo =  \mua - \muxt\bxaa$.
The \tslsxf, $\tslsxformula$, is a special case of \lemtsls\ in which 
\begini
\item $(u_i,v_i,w_i, \epi) = (\yi, \vilsx, \wilsx, \rrilsx)$,
where $\vilsx = \odxt$ and $\wilsx = \ozxt$,
\item the corresponding linear model is given by
\beginy\label{eq:model_lsx}
\beginar{l}
Y_i \oeq{\eqref{eq:rrilsx_decomp}}  \bo + \di\tc +\xit\bxaa + \rrilsx \;=\; \vilsx^\T\beta + \rrilsx,\smallskip\\
\E\left(\sumig \wilsx \rrilsx\right) = 0,
\endar
\endy
where $\beta = (\bo, \tc, \bxaa^\T)^\T$; 
\item  
$C_N = (0, 1, \zpt)^\T \in \mbr^{(2+p)}$.
\endi
Given \lemtsls, it suffices to verify that 
\begine[(i)]
\item \eqref{eq:model_lsx} is correct;  
\item the assumptions in \lemtsls\ are satisfied for \eqref{eq:model_lsx} under \assmasymlsx;
\item the expression of $\vlsxn$ in \thmlsx\eqref{it:lsx_clt} is correct. 
\ende 
We verify these three items in \secs~\ref{sec:lsx_1}--\ref{sec:lsx_3}, respectively. 

\subsubsection{Proof of the linear model in \eqref{eq:model_lsx}.}
\label{sec:lsx_1}
The form of the linear model in \eqref{eq:model_lsx} holds by the definition of $\rrilsx = \res(\ai\mid 1, \xxi)$ and $(\bo, \bxaa)$ in \eqref{eq:rrilsx_decomp}. 
In addition, properties of linear projection ensure 
\beginy\label{eq:proj_lsx}
\E\left\{\beginp 1 \\ \xxi \endp \rrilsx\right\}=0, \quad \E(\rrilsx) = 0. 
\endy
Under \assmtth, \lemcef\ ensures $\covza = 0$ and $\cov(\zi,\xxi) = 0$, so that  $\cov(\zi,\rrilsx) = \cov(\zi, \ai - \xit \bxaa) = 0$ and 
\begineq
\E(\zi \rrilsx) =\cov(\zi, \rrilsx) + \E(\zi)\cdot\E(\rrilsx) \oeq{\eqref{eq:proj_lsx}} 0.\label{eq:zr_lsx}
\endeq
\eqs~\eqref{eq:proj_lsx}--\eqref{eq:zr_lsx} together imply $\E(\wilsx \rrilsx) = 0$. 

\subsubsection{Proof of the assumptions in \lemtsls}\label{sec:lsx_2}
Note that 
\assmasym\eqref{it:assm_asym_123} ensures \lemtsls\eqref{it:hansen_cs}; 
\assmasym\eqref{it:assm_asym_123} ensures \lemtsls\eqref{it:hansen_ng} with $r=2$; 
\assmim\ ensures \lemtsls\eqref{it:hansen_identical},
so it suffices to verify \lemtsls\eqref{it:hansen_rank}--\eqref{it:hansen_bounded} with $s = 2$. 

\paragraph{Proof of \lemtsls\eqref{it:hansen_rank}.}
Under \eqref{eq:model_lsx} and \assmim, let
\begineq\label{eq:glsx_def}
\glsx = \E(\wilsx \vilsx^\T) = \E\left\{\ozx\odx\right\} = \eozxodxm 
\endeq 
denote the common value of $\E(\wilsx\vilsxt)$ across $\otn$, which is independent of $N$.  
The analog of $\gn$ in \lemtsls\eqref{it:hansen_rank}, as defined in \eqref{eq:hansen_def}, is
\beginy\label{eq:gnlsx}
\ds\gnlsx \oeq{\eqref{eq:hansen_def} + \eqref{eq:model_lsx}} \meani \E(\wilsx \vilsxt)
\oeqt{Assm.~\ref{assm:im}} \E(\wilsx \vilsxt) \oeq{\eqref{eq:glsx_def}} \glsx. \qquad
\endy
Under \assmtth, \lemcef\ ensures  
$\muzd - e\mud = \covzd = \vzpc$, and $\muxzt - e\muxt = \cov(\zi,\xxi) = 0$,
so that 
\beginy
\label{eq:det_glsx}
\det(\gnlsx) &\oeq{\eqref{eq:gnlsx}} &\det(\glsx) 
=
\det\beginp 
1 & \mud & \muxt \\ 
0 & \muzd - e\mud  & \muxzt - e\muxt\\
\mux & \muxd & \muxx
\endp
\nnb\\
&=&
\vzpc \cdot \det\vxf
\oeqt{\lemblock} \vzpc\cdot \det(\vx) > 0,
\endy
where the last equality follows from \lemblock\ with $\vx = \muxx - \mux\muxt = \cov(\xxi)$.
This ensures $\gnlsx = \glsx$ has full rank.

\paragraph{Proof of \lemtsls\eqref{it:hansen_lambda_min}.}
Under \eqref{eq:model_lsx} and \assmim, let
$\psilsx = \E(\wilsx \wilsx^\T)$ 
denote the common value of  $\E(\wilsx\vilsxt)$ across $\otn$, which is independent of $N$.  
The analogs of $\psin$ and $\omgn$ in \lemtsls\eqref{it:hansen_lambda_min}, as defined in \eqref{eq:hansen_def}, are
\begineq\label{eq:psinlsx}
\beginar{l}
\ds \psinlsx  
\oeq{\eqref{eq:hansen_def} + \eqref{eq:model_lsx}}  
\meani \E(\wilsx\wilsxt)  \oeqt{Assn.~\ref{assm:im}}  \E(\wilsx \wilsx^\T)  = \psilsx,\\
\omgnlsx^*  
\oeq{\eqref{eq:hansen_def} + \eqref{eq:model_lsx}}  
\dnisumg \cov\left(\dsumig \wilsx \rrilsx\right)
 \oeq{\eqref{eq:omgnlsx_omgnfex_def}}  \omgnlsx,  
\endar
\endeq
respectively,  
so that 
\lemtsls\eqref{it:hansen_lambda_min} holds if 
\begineq\label{eq:lambda_min_lsx}
\lm(\psinlsx) \oeq{\eqref{eq:psinlsx}} \lm(\psilsx) \geq\lambda, \quad  \lm(\omgnlsx)\geq\lambda \quad \text{for some $\lambda>0$.}
\endeq
\assmasym\eqref{it:assm_asym_lsx} ensures the uniform boundedness of $\lm(\omgnlsx)$ in \eqref{eq:lambda_min_lsx}.

In addition, given $\psilsx = \E(\wilsx \wilsx^\T)$, if there exists a nonzero constant vector $a \neq 0$ such that $
a ^\T\psilsx a  = \E\left\{(a^\T\wilsx)^2\right\} = 0$,
then $a^\T\wilsx = 0$ almost surely by \lem~\ref{lem:basic}\eqref{it:lem_basic_as}.
This implies $a^\T \wilsx\vilsx^\T=0$ almost surely, so that  
$
a^\T\glsx\oeq{\eqref{eq:gnlsx}} \E(a^\T \wilsx\vilsx^\T) \oeqt{\lem~\ref{lem:basic}\eqref{it:lem_basic_as}} 0 $
by the definition of $\glsx = \E(\wilsx\vilsxt)$ in \eqref{eq:gnlsx} and \lem~\ref{lem:basic}\eqref{it:lem_basic_as}. This contradicts \eqref{eq:det_glsx}.
Therefore,    $a^\T\psilsx a   = \E\left\{(a^\T\wilsx)(a^\T\wilsx)^\T\right\}> 0 $ for any nonzero constant vector $a$,
which implies $\lm(\psilsx) > 0$.

\paragraph{Proof of \lemtsls\eqref{it:hansen_bounded} with $s = 2$.}\label{sec:bound_joint_lsx}
Given $(u_i, v_i, w_i) = (Y_i, \vilsx, \wilsx)$, it suffices to verify that $\E(\yi^4) < \infty, \quad \E  (\|\vilsx\|^4 ) <\infty, \quad \E  (\|\wilsx\|^4 ) <\infty.$
First, \assmasym\eqref{it:assm_asym_Y} ensures $\E(\yi^4) < \infty$. 
Next, it follows from $
\|\vilsx\|^2 = 1 + \di^2 + \|\xxi\|^2 \le 2 + \|\xxi\|^2 \le 2 \, (1+\|\xxi\|^2)$
that 
\begina
\|\vilsx\|^4 &\leq&   4 (1+\|\xxi\|^2 )^2 \leq 8( 1 + \|\xxi\|^4 ), \quad
\E  (\|\vilsx\|^4 )\leq 8\left\{1 + \E(\|\xxi\|^4)\right\}, 
\enda
which ensures $\E(\|\vilsx\|^4) < \infty$ under \assmasym\eqref{it:assm_asym_lsx}. 
Identical reasoning after replacing $\di$ to $\zi$ ensures $\E(\|\wilsx\|^4) < \infty$.

\subsubsection{Verification of the expression of $\vlsxn$ in \thmlsx\eqref{it:lsx_clt}.}
\label{sec:lsx_3}
From \eqref{eq:gnlsx} and \eqref{eq:psinlsx}, under \eqref{eq:model_lsx} and \assmtth, the analog of $\sign$ in \lemtsls, as defined in \eqref{eq:hansen_def}, is
\beginy\label{eq:sign_lsx}
\signlsx 
\oeq{\eqref{eq:hansen_def}+\eqref{eq:gnlsx}+\eqref{eq:psinlsx}} \gnlsx^{-1} \omgnlsx (\gnlsx^{-1})^\T 
\oeq{\eqref{eq:gnlsx}} \glsxi \omgnlsx (\glsxi)^{\T}. \qquad
\endy
By \lemtsls, the asymptotic variance of $\htlsx$ is 
\beginy\label{eq:vlsxn_def}
\vlsxn = (0, 1, \zpt) \signlsx \beginp 0\\ 1\\ \zpt\endp.  
\endy
By \eqref{eq:glsx_def}, $(-e, 1, \zpt)\glsx = \beginp 0, & \czd, & \zpt \endp = \vzpc (0, 1, \zpt)$
under \assmtth, so that 
\begineq\label{eq:thmlsx_goal}
(0, 1, \zpt) \glsx^{-1} = \dfrac1\vzpc \cdot (-e, 1, \zpt).  
\endeq
Combining this with \eqref{eq:vlsxn_def} verifies the expression of $\vlsxn$ as follows: 
\begina
\vlsxn 
&\oeq{\eqref{eq:vlsxn_def}+\eqref{eq:vlsxn_def}}&  (0, 1, \zpt)\glsxi \omgnlsx (\glsxi)^{\T}\beginp 0\\ 1\\ \zpt\endp\\
&\oeq{\eqref{eq:thmlsx_goal}}&
\dfrac{1}{\vzpcsq}  (-e,1,\zpt) \cdot 
\dnisumg \cov\left\{\dsumig \wilsxf \rrilsx\right\}
\cdot \beginp -e \\ 1 \\ \zpt\endp\\
&=&
\dfrac{1}{\vzpcsq} \cdot \nisumg \var\left\{\dsumig (\zi - e) \rrilsx\right\}. 
\enda   

\subsection{Proof of \thmfex}
Applying \lemfwl\ with $\dis = \di$, $\zis = \zi$, and $W_i = (C_i, \xxi)$ ensures the numeric expressions of $(\htfex, \hsefex)$ in \thmfex\eqref{it:fex_numeric}. 
We verify below \thmfex\eqref{it:fex_clt}.  

Let $\hbxfex$ denote the estimated \coeffv\ of $\xxi$ from $\tsfexformula$. Applying \lemfwl\ with $\dis = (\di, \xxi)$, $\zis = (\zi, \xxi)$, and $W_i = C_i$  ensures that for the purpose of \thmfex, the \tsfexf\ is numerically equivalent to its \cpo\ variant, 
\beginy\label{eq:tscdx}
\tscdxformula,
\endy
where the estimated coefficients on $\ddi$ and $\dxi$ coincide with $\htfex$ and $\hbxfex$, and the \crse\ and IV residuals coincide with $\hsefex$ and $\rifex$, respectively.
Given that \eqref{eq:tscdx} is a just-identified \tsls, the estimation equation implies
\beginy\label{eq:fex_ee_1}
\beginp
\htfex \\ 
\hbxfex 
\endp 
=
\phini 
\left\{\meanif \beginp \dzi \\ \dxi \endp \dyi \right\},
\endy
where $
\phin = \meanif \beginp \dzi \\ \dxi \endp (\ddi, \dxit) = \beginp
\szd & \sxzt\\
\sxd & \sx
\endp$ with $
\sxz = \dmeani \dxzi$, $\sx = \dmeani \dxxi$, and $\sxd = \dmeani \dxdi$. 
%
%
Recall from \eqref{eq:rrifex_def} that $\rrifex = \dai - \dxit\gxa$, where $\gxa = \gxaf$.
Recall from \eqref{eq:omgnlsx_omgnfex_def} the definition of $\omgnfex$. 
Let 
\begineq\label{eq:bi_fex}
B_i = \dzdx \rrifex,\quad \bb = \dmeani B_i, \with \omgnfex  = \cov(\sqrtn \bb). 
\endeq 
Equation~\eqref{eq:fex_ee_1} ensures
\begina
\dbfexf&=& 
\hbffex - \bffex\\
&\overset{\eqref{eq:fex_ee_1}}{=}& \phini\left\{ 
\meani \beginp \dzi \\ \dxi \endp \dyi 
- \meani \beginp \dzi \\ \dxi \endp (\ddi, \dxit)  
\bffex \right\}\\
&\oeq{\eqref{eq:rrifex_def}}&
\phini 
 \left\{ \meani \dzdx \rrifex \right\}
\oeq{\eqref{eq:bi_fex}} \phini \bb, 
\enda
so that 
\beginy\label{eq:dhtfex}
\htfex - \tc 
=  (1,0_p^\T) \dbfexf = (1,\zpt)\phini \bb = \gn \bb,
\endy
where $
\gn =  (1,0_p^\T) \phini \in \mbr^{1\times (1+p)}$. 
We verify in \secs~\ref{sec:fex_1}--\ref{sec:fex_2} (i) $\vfexn^{-1/2} \cdot \sqrtn(\htfex - \tc)  \rs \sn$; and (ii) $N\hsefexsq/\vfexn  = 1+\op$,  
respectively, which together imply $(\htfex - \tc) / \hsefex  \rs \sn$.

\subsubsection{Proof of $\vfexn^{-1/2} \cdot \sqrtn(\htfex - \tc)  \rs \sn$}\label{sec:fex_1}
Let
\begineq\label{eq:bis_fex}
\beginar{l}
\gns =  \kvzpcif(1,0_p^\T),
\quad
\bis = \gns \bi \oeq{\eqref{eq:bi_fex}} \kvzpcif\dzi\rrifex, \quad \bb^* = \dmeani \bi^* = \gns \bb,
\endar
\endeq
with 
\begineq\label{eq:v_omg_fex}
\beginar{rcl}
\var(\sqrtn \bbs)  
&\oeqt{Assm.~\ref{assm:cs}} &  \ni \dsumg \var\left(\sumig \bis\right) \\  
&=& \dfrac1\kvzpcsq \cdot \dnisumg\var\left(\sumig \dzi\rrifex\right) 
= \vfexn,\\
\gns \omgnfex \gnst 
&\oeq{\eqref{eq:bi_fex}}&   \gns \cov(\sqrtn \bb) \gnst  
= \var(\sqrtn \bbs) = \vfexn. 
\endar
\endeq
Recall that $\gn =  (1,0_p^\T) \phini$ from \eqref{eq:dhtfex}. 
Given \lemfex,  similar reasoning as  the proof of \eqref{eq:joint_goal}--\eqref{eq:joint_goal_clt_bi} in Section~\ref{sec:app_proof_of_joint} implies that
\begineq\label{eq:goal_fex_clt}
\beginar{ll}
\gn = \gns + \op,\quad
\omgnfex^{-1/2}  \sqrtn \bb  \;\rs\; \mn(0, I),  \quad \inf_N\lm(\vfexn) > 0,\medskip\\
\vfexnsq \sqrtn \bb^* \;\rs\; \sn, \quad  
\vfexnsq  \sqrtn(\gn - \gns)   \bb = \op.
\endar
\endeq
Combining \eqref{eq:dhtfex}, \eqref{eq:bis_fex}, and \eqref{eq:goal_fex_clt}, Slutsky's theorem ensures that 
\begina
\vfexnsq \cdot \sqrtn (\htfex - \tc) 
&\overset{\eqref{eq:dhtfex}}{=}&
\vfexnsq \cdot \sqrtn\gn\bb \nnb\\
&=&
\vfexnsq\cdot \sqrtn(\gn - \gns)\bb 
+ \vfexnsq   \cdot\sqrtn \gns \bb \nnb\\
&\oeq{\eqref{eq:bis_fex}}&
\underbrace{\vfexnsq \cdot \sqrtn(\gn - \gns)\bb}_{\text{$=\op$ by \eqref{eq:goal_fex_clt}}} 
+ \underbrace{\vfexnsq \cdot \sqrtn \bb^*}_{\textup{$\rs \sn$ by \eqref{eq:goal_fex_clt}}} \nnb\\
&\overset{\eqref{eq:goal_fex_clt}+\eqref{eq:goal_fex_clt}}{\rs}&
 \sn. 
\enda

\subsubsection{Proof of $N\hsefexsq/\vfexn  = 1+\op$}
\label{sec:fex_2}
Let 
\beginy\label{eq:hsign_fex}
\begin{array}{rcccl}
\hsign &=& \szdcx^2 \cdot N\hsefexsq &\oeqt{\thmfex\eqref{it:fex_numeric}}& \dnisumg \left(\dsumig \zicx  \rifex\right)^2,\\
\sign &=& \kvzpcsq \cdot  \vfexn &\oeqt{\thmfex\eqref{it:fex_clt}}& \dnisumg \var\left(\dsumig \dzi \rrifex \right).  
\end{array}
\endy
We show in \secs~\ref{sec:fex_varest_1} and \ref{sec:fex_varest_2} that 
\begineq\label{eq:hsign_goal_fex}
\szdcx = \kvzpc + \op, \quad \hsign = \sign + \op, 
\endeq 
respectively, which together imply  $N\hsefexsq = \vfexn + \op$. When $\inf_N\lm(\vfexn) > 0$ from \eqref{eq:goal_fex_clt}, this implies that $N\hsefexsq/\vfexn = 1 + \op$. 

Let $\hgxz= \sx^{-1}\sxz$
denote the {\coeffv} of $\dxi$ from the \olsr\ of $\dzi$ on $\dxi$.
By FWL,  
\begineq\label{eq:zicx}
\zicx \oeq{\text{Def.}} \rest(\zi \sim \cci + \xxi) 
\oeq{\text{FWL}} \rest(\dzi \sim \dxi) 
=
\dzi  - \dxit \hgxz. 
\endeq
\lemfex\eqref{it:lem_fex_S_plim} ensures $\sxz = \op$ and $\sxx = \E(\sxx) + \op$, so that 
\begineq\label{eq:hgxz_plim}
\hgxz = \sxi \sxz = \{\esx\}^{-1} \op + \op = \op  
\endeq
given \assmlm{\E(\sxx)} under \assmasym\eqref{it:assm_asym_fex}. 

\paragraph{Proof of $\szdcx = \kvzpc + \op$ in \eqref{eq:hsign_goal_fex}.}
\label{sec:fex_varest_1}
\eq~\eqref{eq:apb} ensures
\beginy\label{eq:szdcx_0}
\meani \dzi \di \oeq{\eqref{eq:apb}} \meani\dzi \ddi = \szd, \quad \meani \dxi\di  \oeq{\text{symmetry}} \sxd, 
\endy
where $\szd = \kvzpc + \op$ and $\sxd = \oop$ from \lemfex\eqref{it:lem_fex_S_plim}. Combining these with \eqref{eq:zicx}--\eqref{eq:szdcx_0}, ensures
$\szdcx 
\oeq{\eqref{eq:zicx}}
\meani \left(\dzi - \dxit\hgxz \right)\di
\oeq{\eqref{eq:szdcx_0}}
\szd  - \hgxzt \sxd 
\oeqt{\lemfex\eqref{it:lem_fex_S_plim}+\eqref{eq:hgxz_plim}}\kvzpc+\op.$

\paragraph{Proof of $\hsign = \sign + \op$ in \eqref{eq:hsign_goal_fex}.}
\label{sec:fex_varest_2}
Let $S_i =  \dzi \rrifex$, 
\begineq\label{eq:dtg_fex}
\ss = \meani S_i, 
\quad \ssg = \sumig S_i, 
\quad
 \dtg = \sumig\zicx \cdot \rifex - \ssg  
\endeq
to write 
\begineq
\label{eq:hsign_fex_2}
\hsign 
\oeq{\eqref{eq:hsign_fex}+\eqref{eq:dtg_fex}}
\nisumg \left(\ssg +\dtg\right)^2
=
\nisumg\left\{\ssgsq + \dtg^2 + 2 \ssg\dtg\right\}. 
\endeq
We can show that $\{S_i: \otn\}$ satisfies the conditions in \lemsmclt\ with $\E(S_i) = 0$ as follows:
\begini
\item From \eqref{eq:bi_fex} and \lemfex\eqref{it:lem_fex_uni_B},  $\supi \E(\en{\bi}^{2+2q}) = O(1)$.
It follows from $S_i^2 \leq \|\bi\|^2$ that $\supi \E(S_i^{2+2q}) = O(1)$ so $\{S_i^2:\otn\}$ is  \uni.
\item 
\begineq\label{eq:vss}
\var(\sqrtn\ss) \oeq{\eqref{eq:dtg_fex}} \ds\nif\var\left( \sumi S_i\right)
 \oeq{\eqref{eq:hsign_fex}} \sign.  \quad
\endeq 
From \eqref{eq:omgnlsx_omgnfex_def} and \eqref{eq:vss}, $\sign$ is the (1,1)th element of $\omgnfex$. That $\lm(\omgnfex)\geq \lambda$ for some $\lambda >0$ under \assmasym\eqref{it:assm_asym_fex} implies that $\sign \geq\lambda > 0$. 
 
\item \lemcefa--\eqref{it:lem_cef_EXZ} ensure $\E(\dzi \dai) = 0$ and $\E(\dzi \dxi) = 0$ for all $\otn$, so that $\E(S_i) =\E\{ \dzi(\dai - \dxit \gxa)\}  = 0.$
\endi
Applying \lemsmclt\ to $\{S_i: \otn\}$ ensures the first term in \eqref{eq:hsign_fex_2} satisfies
\beginy
\label{eq:ssgsq_fex}
&&\dnisumg \ssgsq
=  \var (\sqrtn\ss ) + \op  
\oeq{\eqref{eq:vss}} \sign+\op. 
\endy
We show below that 
\beginy\label{eq:dtg_goal_fex}
\nisumg\dtg^2 = \op. 
\endy
Equations \eqref{eq:ssgsq_fex} and \eqref{eq:dtg_goal_fex} together ensure
\beginy\label{eq:ssg_dtg_bound}
\nisumg\ssg \dtg 
\leq
 \left\{\nisumg\ssgsq\right\}^{1/2}
\left(\nisumg\dtg^2\right)^{1/2}
\oeq{\eqref{eq:ssgsq_fex}+\eqref{eq:dtg_goal_fex}}
\op
\endy
by \csf.
Plugging \eqref{eq:ssgsq_fex}--\eqref{eq:ssg_dtg_bound} in  \eqref{eq:hsign_fex_2} ensures $\hsign = \sign + \op$. 

\paragraph*{Proof of \eqref{eq:dtg_goal_fex}.}
Applying \lembasu\ with $\dis = (\di, \xxi)$, $\zis = (\zi, \xxi)$, and $W_i = C_i$ ensures that $(\htfex, \hbxfex, \rifex)$ equal the \coeffs\ and IV residuals from $\tscdxformula$, respectively, so that 
\beginy\label{eq:rifex}
\rifex 
&\overset{\text{\lembasu}}{=}& \dyi - \ddi\htfex - \dxit\hbxfex\nnb\\
&=& \dai + \ddi\tc - \ddi\htfex - \dxit\hbxfex\nnb\\
&=& 
\rrifex - \ddi (\htfex - \tc) - \dxit(\hbxfex - \gxa).  
\endy
Let $\rrgfex = (\rrifex: \iig)$,  $\rgfex = (\rifex: \iig)$,  and $\zgcx = (\zicx:\iig)$
denote the vectors of $\rrifex$, $\rifex$, and $\zicx$ within cluster $g$, with 
\begineq\label{eq:zgcx_rgfex}
\beginar{rcl}
 \zgcx &\oeq{\eqref{eq:zicx}}&  \dzg  - \dxg \hgxz,\\
 \rgfex  
&\oeq{\eqref{eq:rifex}}&  \rrgfex -\ddg(\htfex-\tc) - \dxg (\hbxfex-\gxa)
\endar 
\endeq
from \eqref{eq:zicx} and \eqref{eq:rifex}, where $\dxg = \vjh{\dot X}$.
Let 
\begineqs
\szdg = \dzgt\ddg, \quad
\sxg   = \dxgt \dxg,\quad
\sxzg  = \dxgt \dzg,\quad\sxdg  = \dxgt\ddg,
\endeqs 
\begineqs
\szrg = \sumig  \dzi \rrifex,\quad\sxrg = \sumig  \dxi \rrifex.
\endeqs
Plugging \eqref{eq:zgcx_rgfex} into \eqref{eq:dtg_fex} yields
\begineqs
\beginar{rcl}
\dtg
&\oeq{\eqref{eq:dtg_fex}}& \zgcxt \rgfex - \dzgt\rrgfex\smallskip\\
&\oeq{\eqref{eq:zgcx_rgfex}}&  - \dzgt\ddg(\htfex-\tc)- \dzgt\dxg (\hbxfex-\gxa) 
\nnb\\
&&-  \hgxzt\dxgt\rrgfex + \hgxzt\dxgt\ddg(\htfex-\tc) + \hgxzt\dxgt\dxg (\hbxfex-\gxa)\smallskip\\
&=& \dgo + \dgt + \dgth + \dgfo + \dgfi,
\endar
\endeqs
where $\dgo = - (\htfex - \tc)\szdg$, 
\ $\dgt = -\dhbxfext \sxzg$, \ $\dgth = -\hgxzt \sxrg$, \ $\dgfo = (\htfex - \tc)\hgxzt \sxdg$, and $\dgfi = \hgxzt \sxg \dhbxfexp$.

\eq~\eqref{eq:dhtfex} and the central limit theorem of $\bb$ in \eqref{eq:goal_fex_clt} we just proved ensure  $
\htfex - \tc = \op$ and $\hbxfex - \gxa = \op$. 
Combining this with $\hgxz = \op$ from \eqref{eq:hgxz_plim},  $\nisumg \szdg^2 = \ooo$ from \lemasymnox\eqref{it:lem_asym_nox_S}, and 
\begin{itemize}
\item[]  $\nisumg \sxzg \sxzgt = \oop$, \quad $\nisumg \sxdg \sxdgt = \oop$, \\
$\nisumg \|\sxg\|_2^2 = \oop$, \quad $\nisumg \sxrg \sxrg^\T = \oop$
\endi 
from \lemfex\eqref{it:lem_fex_bound_Sg^2}  ensures that $
\nisumg \dgk^2 = \op$ for $k = 1, 2, 3, 4, 5$. 
This implies 
$\dtg^2 = (\dgo + \dgt + \dgth + \dgfo + \dgfi)^2
\leq  5 (\sum_{k=1}^5 \dgk^2)$ by \csf,
so that $\nisumg \dtg^2  \leq  5  \nisumg\sum_{k=1}^5 \dgk^2 =  5 \sum_{k=1}^5 \left(\nisumg \dgk^2\right) =\op.$

\end{document}

%% file: cmd_0208.tex
\newcommand{\sa}{supplemental appendix}
\newcommand{\diff}{\textup{diff}}\newcommand{\vs}{V_{*,N}}
\newcommand{\condsre}{condition~\eqref{eq:sre}}
\newcommand{\mig}{\mathcal I_g}
\newcommand{\propclt}{\prop~\ref{prop:clt}}
\newcommand{\assmngm}{\assm~\ref{assm:ng_main}}

\def\supi{\sup_{\otn}}
\def\lemblock{\lem~\ref{lem:block_mat}}

\def\cac{C^\T A C}
\def\uni{uniformly integrable}
\def\unin{uniform integrability}

\def\dbfexf{\beginp \htfex - \tc\\\hbxfex - \gxa\endp}

\def\hbffex{\beginp\htfex\\ \hbxfex\endp}

\def\bffex{\beginp \tc \\ \gxa \endp}
\def\dzdx{\beginp \dzi\\ \dxi \endp}

\def\ctc{C^\T C}
\def\lmax{\lambda_{\max}}
\def\olzd{\overline{ZD}}
\def\olzy{\overline{ZY}}

\def\vr{\dfrac{\vls}{\vfe}}
\def\vrf{\kn  \left(1 + \vavef  \cdot \cn \right)}
\def\xigols{\xi_{g,\ls}}
\def\xigfe{\xi_{g,\fe}}
\def\dols{\delta_\ls}
\def\dfe{\delta_\fe}

\def\htt{\widehat T}
\def\ttt{\widetilde T}
\def\httg{\htt_g}
\def\tttg{\ttt_g}
\def\ttg{\ttg}

\def\rgfe{r_{\sg,\fe}}
\def\rgls{r_{\sg,\ls}}
\def\rgols{\rgls}

\def\vvhv{\vnsq(N\hsiglsfe)\vnsq}

\def\ses{\vn^{-1/2}}

\def\tbgs{v_g}
\def\bis{\bi^*}
\def\bbs{\bb^*}
\def\tbgst{v_g^\T}
\def\tvn{\widetilde V_N}
\def\szdi{\szd^{-1}}
\def\gns{\gm^*_N}
\def\gnsf{\dfrac1\vzpc 
 \beginp -e & 1 & 0 \\
 0 & 0 & \kni \endp}

\def\eozodf{\beginp 1 & \mud \\ e & \muzd \endp}
\def\phinf{\meani \beginp 1 \\ \zi\endp (1, \di)}
\def\dbolsf{\beginp \hbo - \mua\\ \htols - \tc\endp}

\def\hbolsf{\beginp \hbo\\ \htols\endp}
\def\bolsf{\beginp \mua\\ \tc\endp}

\def\zpt{0_p^\T}

\def\vzpci{\dfrac1\vzpc}

\def\glsxi{\glsx^{-1}}

\def\enf{Euclidean norm}
\def\sumin{\sum_{i=1}^n}

\def\meanin{n^{-1}\sumin}
\def\abcd{\beginp A & B \\ C & D\endp}

\def\vi{v_i}
\def\swi{w_i}

\def\signlsx{\Sigma_{N,\lsx}}

\def\lm{\lambda_{\min}}
\def\wilsxf{\ozx}
\def\wifexf{
\beginp
\dzi\\
\dxi
\endp}

\def\eozxodxm{\beginp 
1 & \mud & \muxt\\ 
e & \muzd & \muxzt \\
\mux & \mudx & \muxx
\endp}

\def\wsg{W_\sg}
\def\wsgt{W_\sg^\T}
\def\vsg{V_\sg}

\def\ccn{C_N}
\def\ccnt{\ccn^\T}
\def\gnlsx{\Gamma_{N, \lsx}}
\def\glsx{\Gamma_\lsx}

\def\psinlsx{\Psi_{N, \lsx}}
\def\psilsx{\Psi_\lsx}

\def\omgnlsx{\Omega_{N, \lsx}}

\def\omgnfex{\Omega_{N, \fex}}

\def\wilsx{w_{i, \lsx}}
\def\wilsxt{w_{i, \lsx}^\T}
\def\vilsx{v_{i, \lsx}}
\def\vilsxt{v_{i, \lsx}^\T}

\def\ugj{U_{gj}}
\def\tauag{\tau_{\aa,g}}
\def\muyzg{\mu_{Y(0),g}}
\def\bxaa{\beta_{X, A}}

\def\cnf{\nsqphi}
\def\nsqphi{\ni\sumg \ngsq\phig}

\def\nphi{\ni\sumg\ng\phig}
\def\ngsq{\ng^2}
\def\ngsqi{\ng^{-2}}

\def\zidi{\zi\di}
\def\ziyi{\zi\yi}
\def\vave{\va/\ve}
\def\vavef{\dfrac{\va}{\ve}}

\def\va{\sigsq_\alpha}
\def\ve{\sigsq_\epsilon}

\def\ngg{{n_G}}

\newcommand{\vgh}[1]{(#1_{[1]}^\T, \ldots, #1_{[G]}^\T)^\T}
\newcommand{\vgj}[1]{(#1_{11}, \ldots, #1_{G,\ngg})^\T}

\newcommand{\vjh}[1]{(#1_{g1}, \ldots, #1_{g,\ng})^\T}

\def\vtgf{\vz \cdot (\vap \cdot \ngsq\phig + \ve \cdot\ng)}
\def\covxi{\{\cov(X_i)\}^{-1}}
\def\rrglsx{R_{\sg,\lsx}}
\def\boa{\beta_{1,\alpha}}
\def\bxa{\beta_{X,\alpha}}
\def\vap{\sigsq_{\alpha'}}
\def\xsg{X^*_g}
\def\xsgt{X^{*\T}_g}
\def\acip{\aci'}
\def\agp{\ag'}

\def\hvn{N\hsiglsfe}
\def\hsiglsfe{\widehat \Sigma_\lsfe}

\def\hclsfe{\hat\sigma_\lsfe} 
\def\vlsfe{V_{\lsfe,N}}

\def\vfen{\vfe}
\def\vlsn{\vls}

\def\vnsq{V_N^{-1/2}}

\def\vls{V_{\ls}}
\def\vfe{V_{\fe}}

\def\eszd{\E(\szd)}
\def\esxd{\E(\sxd)}

\def\esx{\E(\sx)}

\def\czi{\zi - e}

\def\cai{\ai - \mua}

\def\dgo{\delta_{g,1}}
\def\dgt{\delta_{g,2}}
\def\dgth{\delta_{g,3}}
\def\dgfo{\delta_{g,4}}
\def\dgfi{\delta_{g,5}}
\def\dgk{\delta_{g,k}}

\def\sx{S_{X, \win}}
\def\sxx{\sx}
\def\sxi{\sx^{-1}}
\def\sxz{S_{XZ, \win}}
\def\sxzt{\sxz^\T}
\def\sxd{S_{XD, \win}}

\def\sxa{S_{XA, \win}}

\def\dxxi{\dxi\dxit}
\def\dxdi{\dxi\ddi}
\def\dxzi{\dxi\dzi}
\def\dxai{\dxi\dai}

\def\dhbxfex{\hbxfex - \gxa}
\def\dhbxfexp{(\dhbxfex)}
\def\dhbxfext{\dhbxfexp^\T}

\def\hgxzt{\hgxz^\T}

\def\covzd{\cov(\zi,\di)}
\def\covzy{\cov(\zi,\yi)}
\def\covza{\cov(\zi,\ai)}

\def\muzd{\mu_{ZD}}

\def\muzy{\mu_{ZY}}
\def\muzx{\mu_{XZ}}
\def\muxz{\muzx}
\def\muzxt{\mu_{XZ}^\T}
\def\muxzt{\muzxt}

\def\vzpc{\vz \pc}
\def\kvzpc{\kn \vzpc}
\def\kvzpcsq{(\kvzpc)^2}

\def\kvzpcif{\dfrac1\kvzpc}

\def\vzpcsq{(\vzpc)^2}

\def\edf{\pa+e\pc}

\def\muzyf{\muzyf}

\def\rest{\texttt{res}}

\def\dxg{\dot\xg}
\def\dxgt{\dxg^\T}

\def\ss{\bar S^*}
\def\ssg{S^*_g}
\def\ssgsq{(\ssg)^2}
\def\szdg{S_{ZD,g}}
\def\szxg{S_{XZ,g}}
\def\szxgt{\szxg^\T}
\def\sxzgt{\szxgt}
\def\sxzg{\szxg}
\def\sdxg{S_{XD,g}}
\def\sxdg{\sdxg}
\def\sdxgt{\sdxg^\T}
\def\sxdgt{\sdxgt}

\def\sxg{S_{X,g}}

\def\sxrg{S_{XR,g}}
\def\szrg{S_{ZR,g}}

\def\dtg{\delta_g}

\def\gxa{\gamma_{X,A}}
\def\gxaf{\{\E(\sx)\}^{-1} \E(\sxa)}

\def\sxy{S_{XY,\win}}

\def\rrifex{R_{i,\fex}}
\def\rrifexf{\dai - \dxit\gxa}
\def\rrgfex{R_{\sg,\fex}}
\def\rifex{r_{i,\fex}}
\def\rgfex{r_{\sg,\fex}}

\def\hbxfex{\hb_{X,\fex}}

\def\dxit{\dot X^\T_i}

\def\vfexn{V_{\fex, N}}
\def\vfexnsq{\vfexn^{-1/2}}
\def\vlsxn{V_{\lsx, N}}

\def\phini{\phin^{-1}}
\def\phin{\Phi_N}
\def\psin{\Psi_N}

\def\xxit{\xxi\xit}

\def\oz{\beginp 1 \\ \zi\endp}
\def\od{(1,\di)}
\def\ozx{\beginp 1 \\ \zi \\ \xxi \endp}
\def\ozxt{(1, \zi, \xit)^\T}

\def\odxt{(1, \di, \xit)^\T}
\def\odx{(1, \di, \xit)}

\def\hbo{\hb_1}
\def\bo{\beta_1}

\def\sign{\Sigma_N}
\def\hsign{\widehat\Sigma_N}

\def\hgxz{\hg_{X,Z}}

\def\hsexsq{\hsex^2}
\def\hsefexsq{\hse_\fex^2}
\def\hselsxsq{\hse_\lsx^2}
\newcommand{\hses}{\hse_*}
\def\hselsx{\hse_\lsx}

\def\vn{V_N}

\def\mux{\mu_X}
\def\vx{V_X}

\def\vxf{
\beginp
1 &  \muxt \\ 
\mux  & \muxx
\endp}

\def\muxd{\mu_{XD}}

\def\muxx{\mu_{XX}}

\def\muxt{\mux^\T}
\def\mudx{\muxd}

\def\gm{\Gamma}

\def\rilsx{r_{i,\lsx}}
\def\rrilsx{R_{i,\lsx}}

\def\pz{P_Z}
\def\pzf{\zfwl(\zfwlt \zfwl)^{-1} \zfwlt}
\def\dpzdi{(\dsfwlt \pz \dsfwl)^{-1}}
\def\zfwlt{\zst_\fwl}
\def\dsfwl{D^*_\fwl}
\def\dfwl{\dsfwl}
\def\dsfwlt{\dst_\fwl}

\def\hdfwlsqi{(\hdfwlsq)^{-1}}

\def\dis{\di^*}
\def\zis{\zi^*}
\def\disw{D_{i\mid W}^*}
\def\zisw{Z_{i\mid W}^*}

\def\hgwz{\hg_{W,Z}}
\def\hgwd{\hg_{W,D}}

\def\hgwd{\hg_{W,D}}

\def\wit{\wi^\T}
 
\def\fwl{{\mid W}}

\def\yiw{Y_{i \fwl}}

\def\hdisw{\widehat D^*_{i \fwl}}
\def\hdgfwl{\widehat D^*_{\sg \fwl}}
\def\zgfwl{Z^*_{\sg \fwl}}
\def\zst{Z^{*\T}}
\def\dst{D^{*\T}}
\def\hds{\widehat D^*}
\def\hdst{\widehat D^{*\T}}

\def\zfwl{Z^*_\fwl}
\def\hdgfwlt{\hdst_{\sg\fwl}}
\def\hdfwl{\hds_\fwl}

\def\hdfwlt{\hdst_\fwl}
\def\hdfwlsq{\hdfwlt\hdfwl}
\def\hdfwlsi{(\hdfwlsq)^{-1}}

\def\hbzfwl{\hb}
\def\hsefwl{\hse_\fwl}

\def\rifwl{r_{i \fwl}}
\def\riw{\rifwl}
\def\rgfwl{r_{\sg  \fwl}}
\def\rgfwlt{\rgfwl^\T}
\def\htw{\htau_\fwl}
\def\htfwl{\htw}

\def\szw{S_{Z \fwl}}
\def\szdw{S_{ZD\mid W}}
\def\szyw{S_{ZY\mid W}}

\def\mud{\mu_D}

\def\muy{\mu_Y}

\def\czd{\cov(\zi, \di)}
\def\covzd{\czd}

\def\ho{\widehat \Omega}

\def\gnst{\gn^{*\T}}
\def\omgn{\Omega_N}
\def\tomgn{\widetilde \Omega_N}
\def\gn{\Gamma_N}
\def\gnt{\gn^\T}

\def\sza{S_{ZA,\win}}

\def\mua{\mu_A}

\def\tbg{\widetilde B_g}

\newcommand{\tmug}{\tmu_g}
\def\szyg{S_{ZY,g}}

\def\dsag{\dot a_\sg}
\def\sag{a_\sg}
\def\sagt{\sag^\T}
\def\sbg{b_\sg}

\def\bsag{\bsa_g}
\def\bsbg{\bsb_g}

\def\sai{a_i}
\def\sbi{b_i}
\def\bsa{\bar a}
\def\bsb{\bar b}
\def\dsai{\dot a_i}
\def\dsbi{\dot b_i}

\def\pw{P_W}
\def\pwf{W(\wt W)^{-1}\wt}

\def\szg{S_{Z,g}}

\def\bb{\bar B}
\def\bbn{\bb_N}

\def\xpwxi{(V^\T \pw V)^{-1}}

\def\hg{\hat\gamma}

\def\dai{\dot A_i}

\def\ho{\widehat \Omega}

\def\ddg{\dot D_\sg}
\def\dzg{\dot Z_\sg}
\def\dyg{\dot Y_\sg}
\def\dxg{\dot X_\sg}

\def\dzgt{\dzg^\T}

\def\czg{\check Z_\sg}
\def\czgt{\czg^\T}
\def\cag{\check A_\sg}
\def\czag{\czgt\cag}
\def\dzag{\dzgt\dag}
\def\dzdg{\dzgt\ddg}
\def\czai{(\czi)(\cai)}
\def\dzdi{\dzi\ddi}
\def\dzdgp{\left(\dzdg\right)}
\def\dzagp{\left(\dzag\right)}
\def\czagp{\left(\czag\right)}

\def\og{1_{\ng}}
\def\ogt{\og^\T}

\def\jng{J_{\ng}}
\def\ing{I_{\ng}}
\def\ong{1_{\ng}}
\def\ongt{\ong^\T}
\def\zng{0_{\ng}}
\def\png{P_{\ng}}

\def\mig{\mathcal I_g}

\def\mud{\mu_D}
\def\mual{\mu_\alpha}

\def\da{\dot A}

\def\dag{\da_{\sg}}

\def\aag{A_{\sg}}

\def\baci{\bar A_\ci}

\def\bag{\bar A_g}

\def\lsx{{\ls\textup{-x}}}
\def\fex{{\fe\textup{-x}}}
\def\hsex{\hse_\lsx}
\def\hselsx{\hsex}
\def\hsefex{\hse_\fex}
\def\htfex{\htau_\fex}
\def\htx{\htau_\lsx}
\def\htlsx{\htx}

\def\szdx{S_{ZD\mid X}}
\def\szyx{S_{ZY\mid X}}
\def\szdcx{S_{ZD\mid (C,X)}}
\def\szycx{S_{ZY\mid (C,X)}}

\def\zix{Z_{i\mid X}}

\def\zicx{Z_{i\mid (C,X)}}
\def\zgcx{Z_{\sg\mid (C,X)}}
\def\zgcxt{\zgcx^\T}

\def\wi{W_i}

\def\kgfe{\kappa_{g,\fe}}
\def\kgls{\kappa_{g,\ls}}

\def\pcg{\pi_{\cc,g}}
\def\tcg{\tau_{\cc,g}}
\def\pag{\pi_{\aa,g}}
\def\vzg{\sigma_{Z,g}^2}
\def\eg{e_g}

\def\dyi{\dot Y_i}
\def\dzi{\dot Z_i}
\def\ddi{\dot D_i}
\def\dxi{\dot X_i}
\def\dx{\dot X}

\def\dyif{\yi - \byci}
\def\dzif{\zi - \bzci}
\def\ddif{\di - \bdci}
\def\ddifp{(\ddif)}
\def\dyifp{(\dyif)}
\def\dzifp{(\dzif)}

\def\epi{\ep_i}

\def\eszg{\E(\szg)} 
\def\lsfe{{\ls, \fe}}

\def\sumiipg{\sum_{i,i'\in\mig }}
\def\dsumiipg{\ds\sum_{i,i'\in\mig }}

\def\knf{1 - \nphi} 
\def\phigb{\phig = \phigf}
\def\phigf{\ng^{-2}\sumiipg \corr(\zi, \zip)}
\def\phig{\phi_g}

\def\vzinv{(\vz)^{-1}}
\def\vzi{\vzinv}
\def\esz{\E(\sz)}

\def\hsed{\hse_{\textup{diff}}}
\def\hsels{\hse_\ls}

\def\hsefe{\hse_\fe}
\def\hsefesq{\hsesqfe}

\def\hsesqls{\hsesq_\ls}
\def\hsesqfe{\hsesq_\fe}

\newcommand{\ci}{{c(i)}}

\newcommand\assmlm[1]{$\inf_{N} \lm(#1) > 0$}
\newcommand{\ouc}{1_{\{U_i = \cc\}}}
\newcommand{\oua}{1_{\{U_i = \aa\}}}

\newcommand{\oco}{1_{\{\ci = 1\}}}
\newcommand{\ocgg}{1_{\{\ci = G\}}}

\def\cci{C_i}
\def\ccif{(\oco, \ldots, \ocgg)^\T}

\newcommand{\di}{D_i}

\def\dsumi{\displaystyle\sumi}
\def\dsumig{\displaystyle\sumig}

\def\wt{W^\T}

\def\rg{r_\sg}

\def\ni{N^{-1}}
\def\nif{\dfrac{1}{N}}

\def\tyi{\widetilde Y_i}

\def\ai{A_i}
\def\aai{A_i}
\def\bi{B_i}

\def\rils{r_{i, \ls}}
\def\riols{\rils}
\def\rife{r_{i, \fe}}
\def\ri{r_i}

\def\bn{\bar n}
\def\nb{\bn}
\def\bni{\nbi}
\def\nbi{\nb^{-1}}
\def\hts{\htau_*}

\def\win{\textup{in}}
\def\sz{S_{Z,\win}}
\def\szd{S_{ZD,\win}}
\def\szdf{\meani \dzifp\ddifp}

\def\szy{S_{ZY,\win}}
\def\szt{S_{Z}}
\def\szdt{S_{ZD}}
\def\szyt{S_{ZY}}
\def\szytf{\meani(\zi - \bz)(\yi - \by)}
\def\szdtf{\meani(\zi - \bz)(\di - \bd)}

\def\mi{\mathcal I}
\def\mif{\{gj: \otg; \, j = \ot{\ng}\}}

\def\dmeani{\displaystyle\meani}

\def\dsumg{\displaystyle\sumg}

\def\rz{\rho_Z}

\def\tmu{\widetilde\mu}

\def\for{\quad \text{for} \ \ }
\def\where{\quad\text{where} \ \ }
\def\with{\quad\text{with} \ \ }

\def\kn{\kappa_N}
\def\kni{\kn^{-1}}

\def\cn{c_N}

\def\vfe{V_{\fe,N}}
\def\vls{V_{\ls,N}}

\def\ls{\textup{2sls}} 
\def\fe{\textup{2sfe}}

\def\htls{\htau_\ls}

\def\htfe{\htau_{\fe}}

\def\sg{{[g]}}

\def\xci{X^*_{\ci}}

\def\aci{\alpha_{\ci}}

\def\vz{\sigma^2_{Z}}

\def\zip{Z_{i'}}

\def\ep{\epsilon}

\def\zg{Z_{\sg}}
\def\dg{D_{\sg}}
\def\zgt{\zg^\T}
\def\yg{Y_g}

\def\pctc{\pc\tc}
\def\pc{\pi_{\cc}}
\def\pcsq{\pc^2}

\def\doi{\Delta\omega_i}

\def\sumg{\sum_{g=1}^G}
\def\nisumg{\ni\sumg}
\def\nifsumg{\nif\sumg}
\def\dnisumg{\ni\dsumg}

\def\meangi{N^{-1}\sum_{g=1}^G \sumig}

\def\yg{Y_{\sg}}
\def\xg{X_{\sg}}

\def\ng{n_g}

\def\ig{{i \in \mig }}
\def\iig{\ig}
\def\sumig{\sum_{\ig}}
\def\dsumig{\displaystyle\sumig}
\def\meanig{\ng^{-1}\sumig}
\def\dmeanig{\ng^{-1}\displaystyle\sumig}

\def\bxg{\bar X_g}

\def\ag{\alpha_g}

\def\zgj{Z_{gj}}
\def\sigsq{\sigma^2}

\def\aa{\textup{a}}
\def\cc{\textup{c}}
\def\nn{\textup{n}}

\def\tc{\tau_{\cc}}

\def\aa{\textup{a}}
\def\pa{\pi_{\aa}}

\def\ua{\ui = \aa}
\def\uc{\ui = \cc}
\def\un{\ui = \nn}

\def\lst{\texttt{ols}}
\def\tslst{\texttt{2sls}}
\def\olst{\lst}

\def\beginp{\begin{pmatrix}}
\def\endp{\end{pmatrix}}

\def\beginy{\begin{eqnarray}}
\def\endy{\end{eqnarray}}

\def\T{\top}

\newcommand{\bz}{\bar Z}
\newcommand{\bd}{\bar D}
 
\def\yizd{Y_i(z,d)}
\newcommand{\yid}{Y_{i}(d)}

\def\dizz{D_i(z)}
\def\diz{D_i(0)}
\def\dio{D_i(1)}
\def\ui{U_i}

\def\tyi{\widetilde Y_i}
\newcommand{\ti}{\tau_i}

\def\gj{{gj}}

\def\epg{\epsilon_{\sg}}

\def\epgt{\epg^\T}

\def\mq{\mathcal Q}
\def\mqg{\mq_{\sg}}
\def\mqgp{\mq_{[g']}}
\def\mqi{\mq_i}

\def\pypd{\yizd, \dizz}
\def\pypdz{\pypd,\zi}
\def\mqif{\mqifx}
\def\mqifx{\mqi = \{\pypdz, \xxi: z, d = 0, 1\}}

\def\mqgj{\mathcal Q_{gj}}

\def\byg{\by_g}
\def\bzg{\bz_g}
\def\bdg{\bd_g}

\def\byci{\by_{\ci}}
\def\bzci{\bz_{\ci}}
\def\bdci{\bd_{\ci}}

\newcommand{\tsls}{\tslsc}
\newcommand{\tslsc}{\textsc{2sls}}
\newcommand{\ols}{\textsc{ols}}

\def\tsfe{\textup{\textsc{2sfe}}}

\def\otn{i\in\mi}
\def\otg{g = \ot{G}}
\def\otng{j = \ot{\ng}}

\newcommand{\meani}{N^{-1}\sumi}
\newcommand{\meanif}{\dfrac{1}{N}\sumi}
\def\sumi{\sum_{i\in \mi}}

\def\hb{\widehat\beta}

\def\ag{\alpha_g}

\def\hsesq{\hse^2}
\def\ui{U_i}
\def\hse{\hat{\textup{se}}}

\def\htau{\hat\tau}
\def\htols{\htau_\ls}

\def\hb{\hat\beta}

\def\sqrtn{\sqrt N}
\def\sn{\mn(0,1)}

\def\sigesq{\ve}

\newcommand{\mn}{\mathcal N}

\def\ca{covariate-adjusted}
\def\tslsf{two-stage least squares}
\def\tsfef{two-stage least squares with fixed effects}

\def\tslsxformula{\ctsxformula}
\def\ctsxformula{\tslst(\yi \sim 1 + \di + \xxi \mid 1 + \zi + \xxi)}
\def\tsfexformula{\tslst(\yi \sim \di + \cci + \xxi \mid  \zi + \cci + \xxi)}
\def\tscdxformula{\tslst(\dyi \sim \ddi + \dxi \mid  \dzi + \dxi)}
\def\cpo{$C$-partialled out}

\def\tslsxf{\ca\ \cts}
\def\tsfexf{\ca\ \tsfe}

\def\cts{canonical \tsls}
\def\crse{cluster-robust standard error}

\def\crses{{\crse}s}

\def\baumid{baum2003instrumental}

\def\hansen{\cite{hansen2019asymptotic}}
\def\hansenid{hansen2019asymptotic}

\def\tslsr{\tsls\ regression}
\def\olsr{\ols\ regression}
\def\olsc{\textsc{ols}}

\def\cfes{{\cfe}s}
\def\cfe{cluster-specific effect}
\def\dgp{Under a linear IV model \citep{imbens2005robust}}
 
\def\condeff{the variation in \cfes\ dominates idiosyncratic variation}
\def\condiv{the IV has sufficient within-cluster variation}
\def\vaveword{the relative importance of cluster-level heterogeneity}

\def\ch{cluster heterogeneity}

\def\propxeff{Proposition~\ref{prop:x_eff}}

\def\defls{Definition~\ref{def:ls}}

\def\deffe{Definition~\ref{def:fe}}

\def\defcrse{Definition~\ref{def:crse}}

\def\thmhetero{\thm~\ref{thm:hetero}}

\def\thmlsx{\thm~\ref{thm:lsx}}
\def\thmfex{\thm~\ref{thm:fex}}
\def\thmlsxfex{\thms~\ref{thm:lsx}--\ref{thm:fex}}
\def\thmjoint{\thm~\ref{thm:joint}}

\def\thmeff{Theorem~\ref{thm:eff}}

\def\assmasymlsx{\assmasym\eqref{it:assm_asym_123}--\eqref{it:assm_asym_lsx}}
\def\assmasymfex{\assmasym\eqref{it:assm_asym_123}--\eqref{it:assm_asym_Y} and \eqref{it:assm_asym_fex}}

\def\lembound{\lem~\ref{lem:bound}}

\def\lemasymnox{\lem~\ref{lem:asym_nox}}
\def\lemen{\lem~\ref{lem:lm}\eqref{it:lem_lm_Cx}}
\def\lemlm{\lem~\ref{lem:lm}\eqref{it:lem_lm_CAC}}

\def\lemfwl{\lem~\ref{lem:fwl}}

\def\lembasu{\lem~\ref{lem:basu}}

\def\lema{\lem~\ref{lem:A}}
\def\lemz{\lem~\ref{lem:Z}}

\def\lemtsls{\lem~\ref{lem:hansen_2sls}}
\def\lemsm{\lem~\ref{lem:hansen_sm}}
\def\lemsmclt{\lem~\ref{lem:hansen_sm}\eqref{it:lem_sm_clt}}

\def\lemfex{\lem~\ref{lem:fex}}

\def\lemhetero{\lem~\ref{lem:hetero}}

\def\lemcef{\lem~\ref{lem:cef}}

\def\lemcefd{\lemcef\eqref{it:lem_cef_EDZ}}
\def\lemcefa{\lemcef\eqref{it:lem_cef_EAZ}}

\def\assmasym{Assumption~\ref{assm:asym}}

\def\assmcs{Assumption~\ref{assm:cs}}
\def\ccd{demeaned}

\def\assmim{Assumption~\ref{assm:im}}

\def\assmhm{Assumption~\ref{assm:hetero_marginal}}

\def\assmng{Assumption~\ref{assm:ng}} 

\def\assmiv{Assumption~\ref{assm:iv}}

\def\assmy{Assumption~\ref{assm:y}}
\def\assmyyx{\assms~\ref{assm:y} and \ref{assm:y_x}}
\def\assmyx{Assumption~\ref{assm:y_x}}

\def\assmyall{\assmhomo, \ref{assm:ng_main}--\ref{assm:y}, and proper moment and rank conditions that ensure $\inf_{N}\kn > 0$}

\def\assmhetero{\assms~\ref{assm:cs}--\ref{assm:iv} and \ref{assm:hetero_marginal}}

\def\assmtth{Assumptions~\ref{assm:iv}--\ref{assm:im}}

\def\assmhomo{\assms~\ref{assm:cs}--\ref{assm:im}}

%% file: cmd_housekeeping.tex
\newcommand{\wrt}{with respect to}

\newcommand{\lmd}{\lambda}

\def\simiid{\overset{\textup{\iid}}{\sim}}

\def\eq{Equation}
\def\eqs{Equations}
\def\sec{Section}
\def\secs{Sections}
\def\prop{Proposition}

\def\thm{Theorem}
\def\thms{{\thm}s}

\def\lem{Lemma}

\def\assm{Assumption}
\def\assms{{\assm}s}

\def\fwlf{Frisch--Waugh--Lovell}

\def\csf{Cauchy--Schwarz inequality}
\def\csp{\text{Cauchy--Schwarz}}

\def\coeffv{\coeff\ vector}
\def\coeffs{coefficients}
\def\coeff{coefficient}
\def\nnb{\nonumber}

\def\begina{\begin{eqnarray*}}
\def\enda{\end{eqnarray*}}
\def\beginals{\begin{align*}}
\def\endals{\end{align*}}

\def\beginar{\begin{array}}
\def\endar{\end{array}}

\def\begineqs{\begin{equation*}}
\def\endeqs{\end{equation*}}
\def\begineq{\begin{equation}}
\def\endeq{\end{equation}}
\newcommand{\sm}{Supplementary Material}
\def\begini{\begin{itemize}}
\def\endi{\end{itemize}}
\def\begine{\begin{enumerate}}
\def\ende{\end{enumerate}}

\def\zi{Z_i}
\def\yi{Y_i}
\def\xxi{X_i}
\def\xit{X_i^\T}

\def\yio{Y_i(1)}
\def\yiz{Y_i(0)}

\newcommand{\by}{\bar Y}

\newcommand{\mbr}{\mathbb R}

\newcommand{\fn}[1]{ \|#1 \|_{\textsc{f}}}
\newcommand{\en}[1]{ \|#1 \|}

\newcommand{\spn}[1]{ \|#1\|_2}

\newcommand{\oeqt}[1]{\overset{\textup{#1}}{=}}
\newcommand{\oeq}[1]{\overset{#1}{=}}
\newcommand{\oleq}[1]{\overset{#1}{\le}}
\newcommand{\oleqt}[1]{\overset{\textup{#1}}{\le}}

\newcommand{\ot}[1]{1, \ldots, #1}
\newcommand{\indep}{\perp \!\!\! \perp}
\DeclareMathOperator\cov{cov}
\DeclareMathOperator\argmin{argmin}

\DeclareMathOperator{\diag}{\textup{diag}}

\DeclareMathOperator\corr{corr}
\DeclareMathOperator\res{Res} 
\DeclareMathOperator\proj{Proj} 
\DeclareMathOperator\E{\mathbb E}    
    
\DeclareMathOperator\var{var}

\def\ntinf{N\to\infty}

\def\iid{i.i.d.}
\newcommand{\pr}{\mathbb P}
\newcommand{\op}{o_{\pr}(1)}
\newcommand{\oop}{O_{\pr}(1)}
\newcommand{\ooo}{O(1)}

\newcommand{\rs}{\rightsquigarrow}

\def\ep{\epsilon}
\def\mn{\mathcal N}

\newcommand{\ds}{\displaystyle}